\documentclass[dvipsnames,letterpaper,11pt]{article}
\usepackage[margin=1in]{geometry}
\usepackage{booktabs} 
\usepackage[ruled]{algorithm2e} 
\usepackage{tikz}
\usepackage{tikz-network}
\usetikzlibrary{patterns}
\usetikzlibrary{arrows}
\usepackage[hidelinks]{hyperref}
\usepackage{cleveref}
\usepackage{amsthm}
\usepackage{stmaryrd}
\usepackage[longnamesfirst,sort]{natbib}
\usepackage{caption}
\usepackage{subcaption}
\usepackage{float}
\usepackage{cancel}

\usepackage{xcolor}

\definecolor{blue}{HTML}{34495E}
\definecolor{red}{HTML}{C0392B}
\definecolor{green}{HTML}{F1C40F}
\definecolor{grey}{HTML}{BDC3C7}

\usepackage{amsmath}
\usepackage{tikz}
\usepackage{pgfplots}
\usepackage{soul}

\usepgfplotslibrary{external}

\usepackage{thmtools} 
\usepackage{thm-restate}

\SetAlFnt{\small}
\SetAlCapFnt{\small}
\SetAlCapNameFnt{\small}
\SetAlCapHSkip{0pt}
\IncMargin{-\parindent}
\usepackage{enumitem}
\usepackage{bbold}

\DeclareMathOperator{\EX}{\mathbb{E}}
\newcommand{\NN}{\mathbb{N}}
\newcommand{\calF}{\mathcal{F}}
\newcommand{\PP}{\mathbb{P}}
\newcommand{\RR}{\mathbb{R}}
\newcommand{\EE}{\mathbb{E}}
\newcommand{\calE}{\mathcal{E}}
\newcommand{\calL}{\mathcal{L}}
\newcommand{\calU}{\mathcal{U}}
\newcommand{\calA}{\mathcal{A}}
\newcommand{\calS}{\mathcal{S}}
\newcommand{\calM}{\mathcal{M}}
\newcommand{\calO}{\mathcal{O}}

\newtheorem{lemma}{Lemma}
\newtheorem{corollary}{Corollary}

\newtheorem{definition}{Definition}

\newtheorem{obs}{Observation}
\newtheorem{claim}{Claim}

\crefname{claim}{Claim}{Claims}
\crefname{obs}{Observation}{Observation}

\begin{document}
	
\title{New Combinatorial Insights for Monotone Apportionment\footnote{This work was partially funded by the ANID (Chile) through Grants FB210005, AFB230002, and ACT210005; by NSF (USA) under Grant No.\ DMS-1928930; and by the Alfred P.\ Sloan Foundation under grant G-2021-16778, while one of the authors was in residence at the Simons Laufer Mathematical Sciences Institute (formerly MSRI) in Berkeley, California, during the Fall 2023 semester. 
This work was partially done while the third and fourth authors were with the Center for Mathematical Modeling at the University of Chile. The authors are grateful to Ekin Ergen for pointing them to the problem of the $k$-level in line arrangements.}
}
\author{Javier Cembrano\thanks{Department of Algorithms and Complexity, Max Planck Institute for Informatics.}
\and José Correa\thanks{Department of Industrial Engineering, Universidad de Chile.}
\and Ulrike Schmidt-Kraepelin\thanks{Department of Mathematics and Computer Science, TU Eindhoven.}
\and Alexandros Tsigonias-Dimitriadis\thanks{European Central Bank.}
\and Victor Verdugo\thanks{Institute for Mathematical and Computational Engineering and Department of Industrial and Systems Engineering, PUC Chile.}
}
\date{}
\maketitle

\begin{abstract}
The \emph{apportionment problem} constitutes a fundamental problem in democratic societies: How to distribute a fixed number of seats among a set of states in proportion to the states' populations? This---seemingly simple---task has led to a rich literature and has become well known in the context of the US House of Representatives. 
In this paper, we connect the design of monotone apportionment methods to classic problems from discrete geometry and combinatorial optimization and explore the extent to which randomization can enhance proportionality.

We first focus on the well-studied family of \emph{stationary divisor methods}, which satisfy the strong \emph{population monotonicity} property, 
and show that this family produces only a slightly superlinear number of different outputs as a function of the number of states.
While our upper and lower bounds leave a small gap, we show that---surprisingly---closing this gap would solve a long-standing open problem from discrete geometry, known as the complexity of $k$-levels in line arrangements. 
The main downside of divisor methods is their violation of the \emph{quota} axiom, i.e., every state should receive $\lfloor q_i \rfloor$ or $\lceil q_i \rceil$ seats, where $q_i$ is the proportional share of the state. As we show that randomizing over divisor methods can only partially overcome this issue, we propose a relaxed version of divisor methods in which the total number of seats may slightly deviate from the house size. By randomizing over these methods, we can simultaneously satisfy population monotonicity, quota, and ex-ante proportionality. 

Finally, we turn our attention to quota-compliant methods that are \emph{house-monotone}, i.e., no state may lose a seat when the house size is increased. We provide a polyhedral characterization based on network flows, which implies a simple description of all ex-ante proportional randomized methods that are house-monotone and quota-compliant.
\end{abstract}

\thispagestyle{empty}
\newpage
\setcounter{page}{1}

\section{Introduction}

Deciding how to allocate seats in legislative bodies has been a central topic of discussion in the political organization of democratic societies for over 200 years. This has come to be known as the \emph{apportionment problem}, and, despite its apparent simplicity, its various aspects have challenged mathematicians for several decades. Proportionality stands as a fundamental criterion, demanding that each state receive seats in accordance with its population.\footnote{The notion also applies, for instance, to allocate seats proportionally to the votes obtained by political parties.} Back in 1792, Alexander Hamilton proposed a simple method that first assigns the lower quota---the exact proportional value rounded down---to every state and then allocates the remaining seats to the states with the largest remainders.

Despite its clear intuition and ease of implementation, the Hamilton method later led to unexpected outcomes, commonly referred to as apportionment paradoxes. The first one, known as the ``Alabama paradox'', occurred during the United States congressional apportionment in 1880. C.\ W.\ Seaton, chief clerk of the Census Office, observed that, when transitioning from 299 to 300 representatives, the Hamilton method resulted in Alabama losing a seat. The second paradox, termed the ``population paradox'', occurred between 1900 and 1901 and involved the states of Virginia and Maine. Despite Virginia's larger proportional growth in population compared to Maine, the Hamilton method would have taken a seat away from Virginia and allocated it to Maine.

These paradoxes, together with the fundamental nature and ubiquity of the apportionment problem, sparked an interest in its mathematical study. Subsequently, monotonicity and proportionality have become cornerstone goals when devising apportionment methods. In particular, two desirable properties are house monotonicity and population monotonicity: a method is said to satisfy the former if it escapes the Alabama paradox and the latter if it avoids the population paradox. The notion of proportionality does not extend in any obvious way to integers. However, a commonly adopted approach is the one of quota compliance, ensuring that each state receives its exact proportional value rounded up or down.

The most prevalent population-monotone methods are divisor methods, which entail scaling the population of each state by a common factor and rounding the resulting values. In fact, \citet{balinski2010fair} showed that, subject to what they termed ``rock-bottom requirements'', divisor methods are the only ones satisfying population monotonicity. Each rounding rule yields a particular divisor method. For example, Spain and Brazil use the Jefferson/D'Hondt method, based on downward rounding, to distribute the seats of their Chamber of Deputies across political parties, whereas Germany and New Zealand use the Webster/Sainte-Lagu\"e method, based on nearest-integer rounding. 
Despite their ubiquity, little is known regarding the diversity of apportionments generated by different divisor methods.
Furthermore, with the emerging interest in randomized apportionment methods 
\citep[e.g.][]{aziz2019random, correa2024monotone, goelz2022apportionment, hong2023random}, 
it is natural to ask whether divisor methods with randomized rounding rules may give best-of-both-worlds guarantees by ensuring population monotonicity and getting closer to exact proportionality.

House monotonicity, unlike population monotonicity, is compatible with quota compliance. To the best of our knowledge, three characterizations of house-monotone and quota-compliant methods have been proposed in the literature. Two of these characterizations rely on recursive constructions \citep{balinski2010fair, still1979} while the third one associates apportionment vectors generated by these methods with extreme points of a fractional matching polytope \citep{goelz2022apportionment}.

\subsection{Our Contribution and Techniques}

We present combinatorial descriptions of the space of outcomes generated by the two most common families of monotone apportionment methods and study methods that randomize over this space.

In \Cref{sec:pop-mon}, we focus on divisor methods with stationary rounding rules---rules that round a fractional value upwards if its fractional part exceeds a fixed threshold $\delta\in [0,1]$ and downwards otherwise.
We establish as \Cref{thm:main} a link between the apportionment output by these methods and the $k$-level in a line arrangement, thus drawing a novel connection between two fundamental problems in social choice and computational discrete geometry. 
This provides valuable insights into the behavior of divisor methods as a function of $\delta$.
Specifically, it implies that for any population vector $p$ and house size $H$, it is possible to partition the interval $[0,1]$ into polynomially (almost linearly) many intervals, such that each of them yields a unique common output when $\delta$ lies in its interior. It further implies a superlinear lower bound on the number of such intervals, where the gap comes from a long-standing open question about the complexity of the \mbox{$k$-level} in a line arrangement. The $k$-level in an arrangement of $n$ lines is the closure of all line segments that have exactly $k$ lines strictly below them; its complexity corresponds to its (worst-case) number of vertices and is known to lie between $n\exp(\Omega(\sqrt{\ln n}))$ and $\calO(n^{4/3})$ \citep{dey1998improved, toth2000point}. 
Our upper bound leads to an efficient algorithm to compute all apportionment vectors that are realized for a non-zero measure domain for $\delta$, as well as a concise combinatorial description of the whole set of outputs. 
Even though non-stationary divisor methods may produce exponentially many outcomes, we discuss in \Cref{sec:power} how the upper bound that follows from \Cref{thm:main} can be extended to the family of \emph{power-mean} divisor methods, this time through a connection to pseudoline arrangements.
This family is particularly relevant as it includes the five traditional divisor methods: Adams, Dean, Huntington-Hill, Webster/Sainte-Lagu\"e, and Jefferson/D'Hondt. 
Furthermore, we generalize results regarding the lower quota compliance of the Jefferson/D'Hondt method and the upper quota compliance of the Adams method (\Cref{prop:partition}). We show that for every instance, there is a partition of the interval $[0,1]$ into two intersecting intervals, one containing $0$ and the other containing $1$, such that upper quota is satisfied when $\delta$ lies in the former and lower quota is satisfied when $\delta$ lies in the latter.

Building upon our comprehension of the set of apportionment vectors produced by stationary divisor methods, we study methods that randomize over this set in \Cref{sec:randomization-variable}. Randomizing over well-studied and widely-used apportionment methods, such as divisor methods, constitutes a natural step to make randomized apportionment more applicable in practice. The natural goal of randomization in this context is to mitigate the main drawback of divisor methods: the potential violation of quota by up to $H$ seats. While positive results arise when state populations differ by a constant factor (\Cref{prop:dev-quota-fixed-pop}), we show that the minimum worst-case deviation from quota that randomized stationary divisor methods can achieve is still linear in $H$ (\Cref{prop:rand-dev-quota}). This bound, roughly $H/2$, is nearly matched by a straightforward method that takes $\delta$ equal to either $0$ or $1$, each with probability $1/2$ (\Cref{prop:jeff-adams}).
Motivated by this negative result, our focus shifts to methods that meet the house size in expectation but may slightly deviate from it ex-post. Within this class, akin to divisor methods, certain methods satisfy population monotonicity, quota compliance, and ex-ante proportionality. 
In essence, these methods allocate each state its lower quota and subsequently assign an additional seat with a probability equal to its remainder. We carefully implement the sampling and rounding schemes to guarantee population monotonicity and ex-ante proportionality, accompanied by probabilistic bounds on the deviation from the house size (\Cref{thm:variable-size}). Apart from their ease of implementation, these methods can be seen as a randomized version of both divisor and Hamilton methods, which provides an intuitive understanding of the underlying design principles that enable our method to satisfy these three properties simultaneously.

We finally provide in \Cref{sec:HM} a particularly simple polyhedral characterization of house-monotone and quota-compliant methods, alternative to that by \citeauthor{goelz2022apportionment} We show as \Cref{thm:hm-LP-characterization} that, for any given instance, the set of apportionment vectors output by such methods is exactly the set of (integral) extreme points of a network flow polytope. Our approach is very flexible as it remains valid when incorporating additional properties in our method as long as these properties can be expressed as constraints that preserve the network flow structure. Moreover, combining \Cref{prop:phi} and \Cref{prop:tightness} we provide tight bounds on the size of the linear program needed (which can be thought of as the worst-case lookahead value), so that we can have a polyhedral description of all such possible apportionment vectors for a given house size in a way that house monotonicity and quota compliance are not violated for any higher number of seats. As a consequence of \Cref{thm:hm-LP-characterization}, in \Cref{thm:rand-characterization} we also fully characterize the set of randomized apportionment methods that respect house monotonicity, quota compliance, and ex-ante proportionality up to a house size equal to the total population. Informally, the theorem states that since the proportional fractional allocation is a feasible point within our network flow polytope, a randomized method with the described properties can be obtained by taking any convex combination of its extreme points that results in this point. Conversely, any randomized method respecting house monotonicity, quota compliance, and ex-ante proportionality generates the same allocation, for any number of seats not greater than the total population, as some randomization over extreme points of this polytope. 

\subsection{Related Work}

There is a rich body of literature on the theory and applications of apportionment methods; for a comprehensive treatment of this topic, we refer the reader to the book of \citet{balinski2010fair} and the book of \citet{pukelsheim2017}. Closely related to our work is the stream of literature dealing with the design of house-monotone and quota-compliant methods. The existence of such a method was first shown in \citet{balinskiyoung1974} and~\citet{balinskiyoung1975}. Subsequently (and, in fact, in parallel), \citet{balinski1979} and \citet{still1979} provided simple characterizations of all methods satisfying the two properties. 
Regarding the stronger population monotonicity axiom, \citet{balinski2010fair} showed that, under basic axioms (symmetry and exactness), divisor methods are the unique family satisfying this property but, unfortunately, fail to be quota-compliant. Divisor methods are well known and widely used at national and regional levels in many democracies around the world \citep{balinski2010fair,pukelsheim2017}. If one relaxes the notion of quota compliance, the Jefferson/D'Hondt method has been shown to be the unique divisor method satisfying upper quota compliance, whereas the Adams method has been shown to be the unique divisor method satisfying upper quota compliance \citep{balinski2010fair}. \citet{marshall2002majorization} studied the behavior of apportionments produced by certain families of divisor methods such as stationary divisor methods, showing that seat transfers go from smaller to larger states as a parameter of the method increases.

As is the case of this paper, linear programming and discrete optimization have proven to be powerful tools in the design of apportionment methods. For the biproportional apportionment problem \citep{balinskidemange1989a, balinskidemange1989b}, in which proportionality is ruled by two dimensions (typically states and political parties), \citet{guenterzachariasen2007, gaffkepukelsheim2008a}, and \citet{gaffkepukelsheim2008b} developed a network flow approach to compute a solution. Recently, \citet{cembranocorreaverdugo2022} and \citet{cembranocorreadiazverdugo2021} extended these ideas to the multidimensional case, studying a discrepancy problem in hypergraphs. Furthermore, network flow techniques have been employed in other questions related to the biproportional apportionment problem by \citet{pukelsheimricca2011} and \citet{serafini2011}. Similarly, \citet{mathieu2022} studied the classic apportionment problem with the extra constraint of achieving parity between the representatives of two parts of the population. \citet{shechter2024congressional} has recently proposed a multiobjective optimization approach, studying apportionment vectors at the Pareto frontier between fairness axioms inspired by traditional divisor methods.
Apportionment also has a strong connection to just-in-time sequencing, and apportionment theory has been employed for the design algorithms in this setting \citep{jozefowska2006characterization,bautista1996note,li2022webster}. 

For the case of randomized apportionment, \citet{grimmett2004} first suggested such a method to overcome fairness issues caused by the use of deterministic methods. Despite being ex-ante proportional, quota-compliant, and easy to implement, the method proposed by \citeauthor{grimmett2004} does not satisfy the essential notions of house and population monotonicity. \citet{goelz2022apportionment} developed a randomized, ex-ante proportional method that satisfies population monotonicity but is not quota-compliant. Their method is essentially a divisor method where the signposts are sampled from independent Poisson processes of the same rate. Moreover, they provided a house-monotone, quota-compliant, and ex-ante proportional method based on a dependent rounding approach inspired by the bipartite pipage rounding procedure of \citet{gandhi2006pipage}. 
Their result is based on finding a bipartite matching description of the apportionment vectors. 
Recently, \citet{hong2023random} proposed a randomized method that is quota-compliant and satisfies stochastic versions of house monotonicity and a weaker version of population monotonicity, while \citet{correa2024monotone} studied monotonicity axioms regarding the number of seats assigned to a coalition of districts/parties for randomized apportionment methods.
Finally, \citet{aziz2019random} studied the strategyproof peer selection problem and pointed out that their proposed mechanism contains a randomized allocation subroutine, which can serve by itself as a randomized apportionment method satisfying ex-ante proportionality and quota compliance. 

For an overview of the $k$-level in line arrangement problem, we refer to \citet{matousek2013lectures}. The upper and lower bounds on this problem that we apply in \Cref{sec:pop-mon} were given by \citet{dey1998improved} and \citet{toth2000point}, respectively. 

\section{Preliminaries}\label{sec:prelims}

We denote by $\NN$ the set of strictly positive integer values and by $\NN_0=\NN \cup \{0\}$ the set of non-negative integer values. 
We also denote by $\RR_{+}$ the set of non-negative real numbers and by $\RR_{++}$ the set of strictly positive real numbers.
An instance of the apportionment problem is given by a pair $(p,H) \in \NN^n \times \NN$ for some positive integer $n$, where we refer to $p$ as the \emph{population vector} and $H$ is the so-called \emph{house size}. 
We let $[n]$ denote the set $\{1,\ldots,n\}$ for any natural value $n$ and use $P=\sum_{i=1}^n p_i$ as a shortcut for the total population. 
An \emph{apportionment method} is given by a family of multi-valued functions $f$ mapping an instance in $\NN^n \times \NN$ to a subset of $\NN^n_0$, such that for every $p\in \NN^n,~H\in \NN$, and $x \in f(p,H)$ we have $\sum_{i=1}^{n}x_i = H$.
In a slight abuse of notation, we use $f$ both to refer to a method and to individual functions of the family. For a given method, we use \textit{outcome} to refer to the set that the method outputs for a given instance and \textit{apportionment vector} to refer to the individual vectors that belong to this outcome.

For an instance $(p,H)$ and a state $i$, the \textit{quota} of $i$, given by $q_i=\frac{p_i}{P}H$, corresponds to the number of seats that the state would obtain in a proportional fractional allocation.
In this work, we consider the following axioms for apportionment methods:
\begin{enumerate}[label=(\alph*)]
    \item {\bf Quota compliance.} We say that a method $f$ is \emph{quota-compliant} if every state receives a number of seats equal to the rounding (up or down) of its quota, namely, for every population vector $p=(p_1,\ldots,p_n)$, every house size $H$, every $x\in f(p,H)$, and every $i\in [n]$, it holds $x_i \in \{ \lfloor q_i \rfloor, \lceil q_i \rceil\}$.
    \item {\bf Lower (upper) quota compliance.} We say that $f$ is lower (resp. upper) quota-compliant if every state receives a number of seats greater or equal to the floor (resp. lower or equal to the ceiling) of its quota.
    Namely, for every population vector $p=(p_1,\ldots,p_n)$, every house size $H$, every $x\in f(p,H)$, and every $i\in [n]$, it holds $x_i \geq \lfloor q_i \rfloor$ (resp. $x_i \leq \lceil q_i \rceil$).
    \item {\bf House monotonicity.} We say that $f$ is \emph{house-monotone} if no state receives fewer seats when the house size is incremented: for every $p$, every house sizes $H_1,H_2,H_3$ with $H_1<H_2<H_3$, and every $y \in f(p,H_2)$, there exist $x \in f(p,H_1)$ and $z \in f(P,H_3)$ such that $x\leq y \leq z$.\footnote{Here and throughout the paper, inequalities between vectors denote component-wise inequalities.}
    \item {\bf Population monotonicity.} We say that $f$ is \emph{population-monotone} if, whenever the populations change in a way that the population ratio between two states $i$ and $j$ increases in favor of $i$, it does not occur that state $i$ receives strictly fewer seats and $j$ receives strictly more seats. Formally, for every population vectors $p,p'$, every house sizes $H,H'$, every apportionments $x\in f(p,H), x'\in f(p',H')$, and every pair of states $i,j\in [n]$, whenever $p'_i/p'_j \geq p_i/p_j$ it holds either (i) $x_i\leq x'_i$, (ii) $x_j\geq x'_j$, or (iii) $p'_i/p'_j = p_i/p_j$ and $x' \in f(p,H)$.
\end{enumerate}

Denoting as $\calF$ the set of methods, a \emph{randomized method} consists of $F$, a random variable on $\calF$, and a tie-breaking distribution $B$ on subsets of $\NN^n_0$; we write both compactly as $F^B$.
For each possible population vector $p$ and house size $H$, we write $F^B(p,H)$ for the random variable corresponding to the output of the method realized by $F$ evaluated in $(p,H)$, breaking ties according to $B$ in case the method outputs multiple apportionment vectors.\footnote{This definition is made to ensure that a randomized method outputs a single apportionment vector for every instance, which will turn useful, for example, when computing the expected outcome of a randomized divisor method.}
We omit the superscript corresponding to the tie-breaking distribution whenever we work with methods that output a single vector for every instance.
We remark that any randomness involved in the randomized method is realized independently of the specific instance.
A randomized method $F^B$ is (lower/upper) quota-compliant, house-monotone, or population-monotone if $F$ is such that the methods that do not satisfy these properties are realized with probability zero.
Furthermore, we say that $F^B$ is \emph{ex-ante proportional} if every state receives its quota in expectation: for every $p=(p_1,\ldots,p_n)$, every house size $H$ and every $i\in [n]$, it holds $\EE(F^B_i(p,H))=q_i$.

\section{The Combinatorial Structure of Stationary Divisor Methods}\label{sec:pop-mon}

The most well-known population-monotone apportionment methods are divisor methods. Among these, prominent examples are the Jefferson/D'Hondt method, the Webster/Sainte-Laguë method, and the Adams method, all of which fall into the subclass of \textit{stationary divisor methods}. Despite being widespread, little is known about the diversity of outcomes that can be derived from these methods. In particular, \textit{how many different stationary divisor methods may return different outcomes in the worst case?} Given the omnipresence of divisor methods in real-world politics, answering this question can help us to understand the influence of the choice between different divisor methods. While we cannot answer this question exactly, we show that---perhaps surprisingly---doing so would solve a long-standing open question from discrete geometry. On our way towards finding almost matching upper and lower bounds (\Cref{subsec:ub-bp} and \Cref{subsec:lb-bp}), we derive several structural insights into the space of outcomes of stationary divisor methods through the lens of our new geometric perspective. In \Cref{sec:power}, we discuss how we can extend our upper bound from \Cref{subsec:ub-bp} to the class of \emph{power-mean} divisor methods by applying results for pseudoline arrangements.

The space of stationary divisor methods is parameterized by $\delta \in [0,1]$, and we refer to the stationary divisor method with parameter $\delta$ as the $\delta$-divisor method. To introduce them, we define a \emph{rounding rule} parameterized by $\delta$, which simply rounds up a number if its fractional part is strictly above $\delta$, downwards if it is strictly below $\delta$, and either of them if it is equal to $\delta$. Formally, 
\[
    \llbracket r \rrbracket_\delta = \begin{cases}
        \{0 \} & \text{ if } r < \delta,\\
        \{t \} & \text{ if } t-1 + \delta < r < t + \delta \text{ for some } t\in \NN_0,\\
        \{t,t+1\} & \text{ if } r = t + \delta \text{ for some } t\in \NN_0.
    \end{cases}
\]

For $\delta \in [0,1]$, the \textit{$\delta$-divisor method} is a family of functions $f(\cdot,\cdot;\delta)$ (i.e., one function for each number $n \in \mathbb{N}$)
such that for every $p\in \NN^n$ and $H\in \NN$
\[
    f(p,H;\delta) = \bigg\{ x\in \NN^n_0 ~\bigg|~ \text{there exists } \lambda>0 \text{ s.t.\ } x_i \in \llbracket \lambda p_i\rrbracket_\delta \text{ for every }i\in [n] \text{ and } \sum_{i=1}^{n} x_i = H \bigg\}.
\]
The Jefferson/D'Hondt method corresponds to $f(\cdot,\cdot;1)$ in our notation, the Webster/Sainte-Lagu\"e method corresponds to $f\big(\cdot,\cdot;\frac{1}{2}\big)$, and the Adams method corresponds to $f(\cdot,\cdot;0)$.
For an instance $(p,H)$, a value $\delta\in [0,1]$ and a vector $x\in f(p,H;\delta)$, we let $$\Lambda(x;\delta) = \left\{ \lambda\in \RR_{++} ~|~ x_i\in \llbracket \lambda p_i\rrbracket_{\delta} \text{ for all } i\in [n] \right\}$$ denote the set of multipliers producing this output via the $\delta$-divisor method.

\paragraph{Breaking points.} To study the diversity of stationary divisor methods, we introduce the notion of \textit{breaking points} of an instance $(p,H)$, which informally correspond to the values of $\delta$ at which the outputs of stationary divisor methods change.

\begin{definition}
    The \textit{breaking points} of an instance $(p,H)$ are defined by $\tau_0=0$ and inductively by 
    \[\tau_i = \max \{\delta \in (\tau_{i-1},1] \mid f(p,H;\delta_1) = f(p,H;\delta_2) \; \forall \; \delta_1,\delta_2 \in (\tau_{i-1},\delta)\}\]
    until $\tau_i = 1$, at which point we fix $B=i$.
\end{definition}
By the definition it follows directly that $0=\tau_0<\tau_1<\tau_2<\cdots<\tau_{B-1}<\tau_{B}=1$ and that, for all $i \in \{0,\dots,B-1\}$ and $\delta, \delta' \in (\tau_i,\tau_{i+1})$, $f(p,H;\delta) = f(p,H;\delta')$. The following observation states a natural convexity notion of stationary divisor methods.

\begin{restatable}{obs}{obsconvexity}\label{obs:convexity}
Let $(p,H)\in \NN^n\times \NN$ and $\delta_1,\delta_2\in [0,1]$ be arbitrary values with $\delta_1<\delta_2$.
Then, for every $x\in f(p,H;\delta_1)\cap f(p,H;\delta_2)$ and every $\delta\in [\delta_1,\delta_2]$, we have that $x\in f(p,H;\delta)$.
\end{restatable}

To prove this observation, we show that an appropriate convex combination of the two multipliers associated with $x$ in $f(p,H;\delta_1)$ and in $f(p,H;\delta_2)$ is itself a multiplier associated with $x$ in $f(p,H;\delta)$; see \Cref{app:obs-convexity} for the proof.
The example in \Cref{fig:line-arrangement-a} shows that the interval $[\delta_{-}(x),\delta_{+}(x)]$ for which a vector $x$ is output might actually be a single breaking point, i.e., we may have $\delta_{-}(x) = \delta_{+}(x)$. 

\begin{figure}[t]
\centering
\begin{subfigure}[b]{.45\textwidth}
    \centering
    \scalebox{0.9}{
    \begin{tikzpicture}
    \node at (0,0) {}; 
    \draw[fill=grey!40,draw=grey, thick] (0.6,-.5) rectangle (3.3,-1); 
    \draw[fill=grey!40,draw=grey, thick] (3.4,-.5) rectangle (6.3,-1); 
    \node at (1.95,-0.75) {$(\textcolor{blue}{2},\textcolor{green}{1},\textcolor{red}{1})$}; 
    \node at (4.75,-0.75) {$(\textcolor{blue}{3},\textcolor{green}{1},\textcolor{red}{0})$}; 

    \draw[fill=black,draw=black] (3.32,-.5) rectangle (3.38,-1.25); 
    \node at (3.35,-1.5) {$(\textcolor{blue}{2},\textcolor{green}{1},\textcolor{red}{1})$}; 
    \node at (3.35,-1.9) {$(\textcolor{blue}{2},\textcolor{green}{2},\textcolor{red}{0})$}; 
    \node at (3.35,-2.3) {$(\textcolor{blue}{3},\textcolor{green}{1},\textcolor{red}{0})$}; 
    
    \begin{axis}[xlabel=$\delta$,ylabel=$\frac{t+\delta}{p_i}$, xmin=0,xmax=1,ymin=0,ymax=1, axis lines=center, axis equal]
    \def\H{6}

    \addplot[domain=0:0.25,color=grey,line width=6pt]{1/5 + x/5}; 
   \addplot[domain=0.25:0.5,color=grey,line width=6pt]{x}; 
    \addplot[domain=0.5:1,color=grey,line width=6pt]{2/5 + x/5};

    \foreach \t in {0,...,\H}{
        \addplot[domain=0:1, color=blue,line width=1.5pt]{\t/5 + x/5}; 
        \addplot[domain=0:1, color=green,line width=1.5pt]{\t/3 + x/3}; 
        \addplot[domain=0:1, color=red,line width=1.5pt]{\t/1 + x/1}; 
    }

    \node[circle, thick, draw =blue, fill=blue] at (0,20){}; 
    \node[circle, thick, draw =blue, fill=blue] at (100,60){}; 
    \node[circle, thick, draw =blue, fill=blue] at (50,50){}; 
    \node[circle, thick, draw =blue, fill=white] at (25,25){}; 
    \end{axis}
    \end{tikzpicture}}
    \caption{Illustration of instance $(p=(5,3,1),H=4)$.}
    \label{fig:line-arrangement-a}
\end{subfigure} \quad  \centering \begin{subfigure}[b]{.45\textwidth}
    \centering
    \scalebox{0.9}{
    \begin{tikzpicture}
    \node at (0,0) {}; 
    \draw[fill=grey!40,draw=grey, thick] (0.6,-.5) rectangle (1.65,-1); 
    \draw[fill=grey!40,draw=grey, thick] (1.75,-.5) rectangle (3.35,-1); 
    \draw[fill=grey!40,draw=grey, thick] (3.45,-.5) rectangle (5.05,-1); 
    \draw[fill=grey!40,draw=grey, thick] (5.15,-.5) rectangle (6.3,-1); 
    
    \node at (1.125,-0.75) {$(\textcolor{blue}{3},\textcolor{green}{2},\textcolor{red}{1})$}; 
    \node at (2.5,-0.75) {$(\textcolor{blue}{4},\textcolor{green}{1},\textcolor{red}{1})$}; 
    \node at (4.2,-0.75) {$(\textcolor{blue}{4},\textcolor{green}{2},\textcolor{red}{0})$}; 
    \node at (5.675,-0.75) {$(\textcolor{blue}{5},\textcolor{green}{1},\textcolor{red}{0})$}; 

    \draw[fill=black,draw=black] (1.67,-.5) rectangle (1.73,-1.25); 
    \draw[fill=black,draw=black] (3.37,-.5) rectangle (3.43,-1.25); 
    \draw[fill=black,draw=black] (5.07,-.5) rectangle (5.13,-1.25); 
    \node at (1.7,-1.5) {$(\textcolor{blue}{3},\textcolor{green}{2},\textcolor{red}{1})$}; 
    \node at (1.7,-1.9) {$(\textcolor{blue}{4},\textcolor{green}{1},\textcolor{red}{1})$}; 
    \node at (1.7,-2.28) {}; 

    \node at (3.4,-1.5) {$(\textcolor{blue}{4},\textcolor{green}{1},\textcolor{red}{1})$}; 
    \node at (3.4,-1.9) {$(\textcolor{blue}{4},\textcolor{green}{2},\textcolor{red}{0})$}; 

    \node at (5.1,-1.5) {$(\textcolor{blue}{4},\textcolor{green}{2},\textcolor{red}{0})$}; 
    \node at (5.1,-1.9) {$(\textcolor{blue}{5},\textcolor{green}{1},\textcolor{red}{0})$}; 
    \begin{axis}[xlabel=$\delta$,ylabel=$\frac{t+\delta}{p_i}$, xmin=0,xmax=1,ymin=0,ymax=1, axis lines=center, axis equal]
    \def\H{6}

    \addplot[domain=0:0.2,color=grey,line width=6pt]{1/3 + x/3}; 
    \addplot[domain=0.2:0.44,color=grey,line width=6pt]{3/8 + x/8}; 
    \addplot[domain=0.42:0.5,color=grey,line width=6pt]{0/1 + x/1};
    \addplot[domain=0.5:0.8,color=grey,line width=6pt]{1/3 + x/3};
    \addplot[domain=0.8:1,color=grey,line width=6pt]{4/8 + x/8};

    \foreach \t in {0,...,\H}{
        \addplot[domain=0:1, color=blue,line width=1.5pt]{\t/8 + x/8}; 
        \addplot[domain=0:1, color=green,line width=1.5pt]{\t/3 + x/3}; 
        \addplot[domain=0:1, color=red,line width=1.5pt]{\t/1 + x/1}; 
    }

    \node[circle, thick, draw =blue, fill=blue] at (0,33){}; 
    \node[circle, thick, draw =blue, fill=blue] at (20,40){}; 
    \node[circle, thick, draw =blue, fill=white] at (42,42){}; 
    \node[circle, thick, draw =blue, fill=blue] at (50,50){}; 
    \node[circle, thick, draw =blue, fill=blue] at (80,60){}; 
    \node[circle, thick, draw =blue, fill=blue] at (100,62){}; 
    
    \end{axis}
    \end{tikzpicture}}
    \caption{Illustration of instance $(p=(8,3,1),H=6)$.}
    \label{fig:line-arrangement-b}
\end{subfigure}
\caption{Illustration of the outputs and breaking points of apportionment instances. Linear functions in $\mathcal{L}$ corresponding to state $1$ ($2$,$3$, respectively) are illustrated by blue (yellow, red, respectively) lines. The function $\lambda_H$ corresponding to the $(H-1)$-level is illustrated by thick light gray segments. 
Filled circles correspond to breaking points of $(p,H)$; unfilled circles correspond to vertices of $\lambda_H$ that are not breaking points of $(p,H)$. Outputs for each value of $\delta\in [0,1]$ are shown below the plots.}
\label{fig:line-arrangement}
\end{figure}
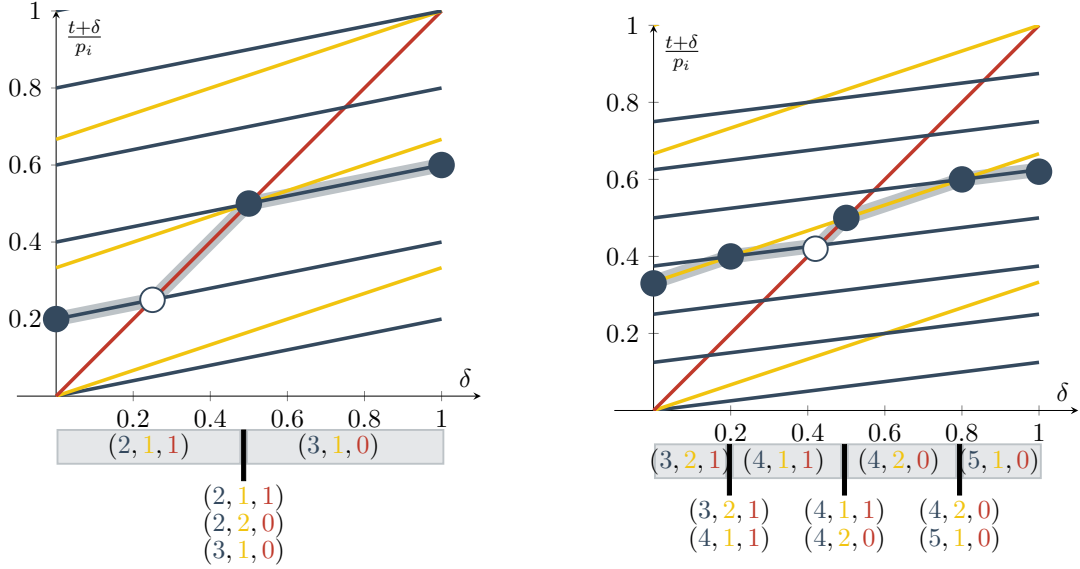

\paragraph{Majorization and a first upper bound.} \citet{marshall2002majorization} showed a particular relation between apportionments produced by different divisor methods. For vectors $x,y\in \NN^n_0$ whose components sum up to the same value, we say that~$x$ \emph{majorizes}~$y$ if, for every $i\in[n]$, the sum of the $i$ largest components of $x$ is at least the sum of the $i$ largest components of $y$. The authors showed that, for all values $0\leq \delta \leq \delta' \leq 1$ and all instances $(p,H)$, any vector in $f(p,H;\delta')$ majorizes any vector in $f(p,H;\delta)$. In simple terms, increasing $\delta$ can only produce seat transfers from smaller to larger districts. 
An upper bound of $\calO(nH)$ on the number of breaking points follows immediately from this property, as seat transfers go in a single direction. In order to beat this upper bound and get rid of the dependence on~$H$, we will explore in the remainder of this section the geometric structure of divisor methods through a connection to line arrangements.

\subsection{Connection to Line Arrangements}\label{subsec:line-arrangements}

In the following, we uncover structural insights about breaking points and the space of outcomes of $\delta$-divisor methods. To do so, we first draw a connection to line arrangements. 
For an apportionment instance $(p,H)$, we introduce the following family of linear functions with domain in $[0,1]$: \[\mathcal{L}(p,H) = \bigg\{\ell_{i,t}(\delta) = \frac{t}{p_i} + \frac{\delta}{p_i}\;\bigg\vert\; i \in [n], t \in \{0,\dots,H-1\}\bigg\}.\] 
These functions are illustrated in \Cref{fig:line-arrangement}. When clear from the context, we omit the apportionment instance and write $\mathcal{L}$ instead of $\mathcal{L}(p,H)$. 
For some $x \in \mathbb{N}^{n}_0$ with $\sum_{i =1}^n x_i = H$, it holds that $x \in f(p,H;\delta)$ if and only if there exists $\lambda \in \mathbb{R}_{++}$ such that for every $i \in [n]$ we have that $x_i - 1 + \delta \leq \lambda p_i \leq x_i + \delta$. Dividing the constraint for $i \in [n]$ by $p_i$, this is equivalent to the condition that \[\ell_{i,x_i-1}(\delta) \leq \lambda \leq \ell_{i,x_i}(\delta).\]
Interpreting these constraints geometrically, we get the following equivalence: For some $x \in \mathbb{N}^{n}_0$ with $\sum_{i =1}^n x_i = H$, it holds that $x \in f(p,H;\delta)$ if and only if there exists $\lambda \in \mathbb{R}_{++}$ such that for every $i\in [n]$ we have that $\ell_{i,0}(\delta)< \dots < \ell_{i,x_i-1}(\delta) \leq \lambda$ and $\ell_{i,x_i}(\delta) \geq \lambda$. Hence, another interpretation of the multiplier $\lambda$ is that of a threshold such that all lines in $\mathcal{L}$ below $\lambda$ get assigned one seat and those lines at $\lambda$ potentially get a seat assigned. Here, \textit{assigning} a seat to a line refers to assigning a seat to the corresponding state. In order to meet the house size $H$, $\lambda$ must be chosen such that the number of lines strictly below $\lambda$ is at most $H$ and the number of lines weakly below $\lambda$ is at least $H$. 
Hence, the minimum choice of $\lambda$ as a function of $\delta$ is given by
$$\lambda_H(\delta) = \min \big\{\lambda \in \mathbb{R}\;\big\vert\; |\{\ell \in \mathcal{L}(p,H) \mid \ell(\delta) \leq \lambda\}| \geq H \big\}.$$ 

We also illustrate $\lambda_H$ in \Cref{fig:line-arrangement}. 
Recall that, for a given $x \in f(p,H;\delta)$, $\Lambda(x;\delta)$ is the set of feasible multipliers for $x$ in the $\delta$-divisor method. We immediately get that $\Lambda(x;\delta) = [\lambda_H(\delta), \lambda_{H+1}(\delta)]$ and this set does not depend on $x$.
Now, let $\mathcal{L}_H(\delta) = \big \{ \ell \in \mathcal{L} \mid \ell(\delta) = \lambda_H(\delta)\big\}$ and, similarly, $\mathcal{L}_{<H}(\delta) = \big \{ \ell \in \mathcal{L} \mid \ell(\delta) < \lambda_H(\delta)\big\}$, $\mathcal{L}_{\leq H}(\delta) = \mathcal{L}_{< H}(\delta) \cup \mathcal{L}_{H}(\delta)$, and $\mathcal{L}_{\geq H}(\delta) = \mathcal{L} \setminus \mathcal{L}_{<H}(\delta)$. We obtain the following characterization of the apportionment output by the $\delta$-divisor method.

\begin{obs} \label{obs:lineSets}
For $(p,H)\in \NN^n\times \NN$, $\delta\in[0,1]$, and $x \in \mathbb{N}^n_0$ with $\sum_{i \in [n]} x_i = H$, it holds that $$x \in f(p,H;\delta) \Longleftrightarrow \ell_{i,x_i-1} \in \mathcal{L}_{\leq H}(\delta) \text{ and } \ell_{i,x_i} \in \mathcal{L}_{\geq H}(\delta) \text{ for every } i \in [n].$$ 
\end{obs}

Note that \Cref{obs:lineSets} restricted to fixed $\delta$ is well known. Yet, to the best of our knowledge, we are the first to systematically study the set of outcomes derived from stationary divisor methods for varying $\delta$. This brings us to the connection to line arrangements.

\paragraph{The $k$-level in line arrangements.} Consider a set of $n$ lines in the plane. The intersections of two lines are called \emph{vertices} and the \emph{edges} are the line segments between any two vertices. 
For some $k \in [n]$, the \emph{$k$-level} of the line arrangement is the closure of all edges that have exactly $k$ lines strictly below them.
For some apportionment instance $(p,H)$, the set of lines $\mathcal{L}(p,H)$ can be directly interpreted as a line arrangement of $nH$ lines in the plane. Moreover, unless $p_i = p_j$ for some $i,j \in [n]$, the $(H-1)$-level of $\mathcal{L}(p,H)$ is well defined and equals exactly $\lambda_H(\delta)$. Note that $\lambda_H(\delta)$ is well defined even when $p_i=p_j$ for some $i,j \in [n]$. In the literature on line arrangements, it is often assumed that lines are in \textit{general position}, i.e., no three lines intersect at the same point. While we cannot make this assumption for $\mathcal{L}(p,H)$ (see, e.g., \Cref{fig:line-arrangement-a}), whenever we apply results from the literature the general position assumption is without loss of generality. 

\paragraph{Vertices and breaking points.} In the next section, we will show that if $\tau \in [0,1]$ is a breaking point of $(p,H)$, then $(\tau,\lambda_H(\tau))$ is a vertex of $\lambda_H(\delta)$. Conversely, not every vertex of $\lambda_H(\delta)$ is located at a breaking point. Consider for example the unfilled circle in \Cref{fig:line-arrangement-a}. Since the set $\mathcal{L}_{\leq H}(\delta)$ is exactly of size $H$ and does not change at this intersection point, $f(p,H;\delta)$ does not change at this point. In general, vertices of $\lambda_H(\delta)$ at which the slope increases (also referred to as \emph{convex} vertices) do not correspond to breaking points of $(p,H)$ while points at which the slope decreases (also referred to as \emph{concave} vertices) correspond to breaking points. The number of vertices of a $k$-level is also referred to as its \emph{complexity}. Establishing tight worst-case bounds on the complexity of a $k$-level (in terms of the number of lines and $k$) is a long-standing open problem in discrete geometry. For an overview of known results, we refer to Section 11 by \citet{matousek2013lectures}. 

The connection we have sketched between the breaking points of an apportionment instance and the complexity of the $k$-level in a line arrangement for some $k$ is formally stated in the following theorem, which constitutes the main result of this section.

\begin{restatable}{theorem}{thmMain}\label{thm:main}
Let 
\begin{itemize}
    \item $g: \NN \to \RR$ be such that, for any arrangement of $m$ lines and any $k\in \{0,\ldots,m\}$, the complexity of the $k$-level is bounded by $\calO(g(m))$, and
    \item $h: \NN \to \RR$ be such that, for any $m\in \NN$, there exists an arrangement of $m$ lines whose $k$-level has complexity $\Omega(h(m))$ for some $k\in \{0,\ldots,m\}$.
\end{itemize}
Then, for any apportionment instance $(p,H)$, the number of breaking points of $(p,H)$ is upper bounded by $\mathcal{O} (g(n))$, where $n$ is the number of states. Conversely, for any $n\in \NN$, there exists an apportionment instance with $n$ states and $\Omega(h(n))$ breaking points.
\end{restatable}

We can now directly apply the best-known bounds for the complexity of the $k$-level in an arrangement of $m$ lines. \citet{dey1998improved} proved that, for any arrangement of $m$ lines and any $k \in \{0,\dots,m\}$, the complexity of the $k$-level is bounded by $\mathcal{O}(m^{4/3})$, while \citet{toth2000point} showed that, for any $m\in \mathbb{N}$, there exists an arrangement of $m$ lines  whose $k$-level has complexity $m \exp({\Omega(\sqrt{\ln m})})$ for some $k\in \{0,\ldots,m\}$.
We obtain the following corollary.

\begin{corollary}\label{cor:BestKnownBounds}
Let $n\in \mathbb{N}$. For any apportionment instance $(p,H)$ with $n$ states, the number of breaking points of $(p,H)$ is upper bounded by $\mathcal{O} (n^{4/3})$. Conversely, there exists an apportionment instance with $n$ states and $n e^{\Omega(\sqrt{\ln n})}$ breaking points.
\end{corollary}

We formally prove \Cref{thm:main} in \Cref{subsec:ub-bp,subsec:lb-bp}. First, we derive further consequences of the geometric approach we have developed. 

\subsection{Structural Insights}

In this section, we provide structural insights into the space of outcomes induced by $\delta$-divisor methods using the geometric interpretation of the assignment of seats performed by these methods.

\paragraph{Quota intervals.} It is well known that the Jefferson/D'Hondt method satisfies lower quota and that the Adams method satisfies upper quota. We now show that the whole set of stationary divisor methods, given by $\delta\in [0,1]$, can be partitioned into three instance-specific subintervals, depending on whether the output satisfies lower quota, upper quota, or both.

\begin{restatable}{proposition}{quotaintervals}\label{prop:partition}
    For every $(p,H)$, there exist $\tau,\bar{\tau}\in [0,1]$ with $\tau\leq \bar{\tau}$ such that:
    \begin{enumerate}[label=(\roman*)]
        \item For every $\delta\in [0,\bar{\tau}]$, every $x\in f(p,H;\delta)$, and every $i\in [n]$, we have $x_i \leq \lceil q_i \rceil$;\label{partition-uq}
        \item For every $\delta\in [\tau,1]$, every $x\in f(p,H;\delta)$, and every $i\in [n]$, we have $x_i \geq \lfloor q_i \rfloor$.\label{partition-lq}
    \end{enumerate}
\end{restatable}

The proof of \Cref{prop:partition} can be found in Appendix \ref{app:prop-partition}. The main idea is to consider $\lambda_1=\min\{\lambda \in \RR_{++} ~|~ \lambda p_i \geq \lfloor q_i\rfloor \text{ for all } i\in [n]\}$ and $\lambda_2=\max\{\lambda \in \RR_{++} ~|~ \lambda p_i \leq \lceil q_i\rceil \text{ for all } i\in [n]\}$ and to show that the claim holds for $\tau = \lambda_H^{-1}(\lambda_1)$ and $\bar{\tau} = \lambda_H^{-1}(\lambda_2)$. Note that $\lambda_H(\delta)$ is strictly increasing, which is why the inverse of the function exists. 

\paragraph{Breaking points and ties.} We come back to the task of determining the breaking points of an instance, i.e., those values of $\delta \in [0,1]$ for which $f(p,H,\delta)$ changes. We show that any breaking point corresponds to a vertex in $\lambda_H$, and, under the assumption that all populations differ from one another, breaking points are exactly those $\delta \in [0,1]$ for which $f(p,H;\delta)$ contains more than one apportionment vector. See \Cref{fig:line-arrangement} for an illustration. The proof of \Cref{prop:BreakingPointsAndTies} can be found in \Cref{app:breakingTies} and makes use of the alternative definition of $\delta$-divisor methods given in \Cref{obs:lineSets}.

\begin{restatable}{proposition}{tiesAndBreaking} \label{prop:BreakingPointsAndTies}
    Let $(p,H)$ be an apportionment instance. If $\tau \in [0,1]$ is a breaking point, then $\lambda_H(\delta)$ has a vertex at $\tau$ and $|f(p,H;\tau)|>1$. Furthermore, if $p_i \neq p_j$ for all $i,j \in [n]$ with $i \neq j$ and $\delta \in [0,1]$ is such that $|f(p,H;\delta)|>1$, then $\delta$ is a breaking point. 
\end{restatable}

\paragraph{Number of outcomes vs.\ apportionment vectors.} In \Cref{subsec:ub-bp}, we will bound the number of breaking points by a polynomial in $n$, which by definition, yields an upper bound on the number of outcomes of a method (namely, of twice the number of breaking points). 
This is in contrast to the number of apportionment vectors, which may be exponential. While this is easy to see when the populations of states can be equal, we provide an example where all populations are different. This happens when a high number of lines from $\mathcal{L}$ intersect with $\lambda_H$ at the same point. Our example extends the one in \Cref{fig:line-arrangement-a} and proves the following observation in \Cref{app:obs-number-solutions}.

\begin{restatable}{obs}{obsnumbersolutions}\label{obs:number-solutions}
    For every $n\in \NN$, there exist $p\in \NN^n$ and $H\in \NN$ with $|f(p,H;0.5)| = \Omega(2^n / \sqrt{n})$.
\end{restatable}

We remark that such examples are highly unlikely to occur in practice since they require lines not to be in general position. For example, sampling populations from some uniform distribution would almost surely lead to lines in general position (besides the intersection of the lines $\ell_{i,0}$ in $0$).

\subsection{Upper Bound on the Number of Breaking Points}\label{subsec:ub-bp}

In this section, we exploit the connection outlined in previous sections to prove the upper bound on the number of breaking points for any apportionment instance in terms of the number of states established in \Cref{thm:main}. 
Observe that our construction of a line arrangement from an apportionment instance in \Cref{subsec:line-arrangements} involves $nH$ lines. Thus, in order to directly apply an upper bound on the complexity of the $k$-level in an arrangement of $m$ lines, we would need to replace $m$ by $nH$. Instead, we show that we can reduce $\mathcal{L}$ to an arrangement of $2n-1$ lines such that the $(H-1)$-level of $\mathcal{L}$ exactly corresponds to some $k$-level of the reduced line arrangement. The upper bound then follows directly. An illustration is shown in \Cref{fig:fixedLines}.

\newcommand{\clp}{\mathcal{L}_{\leq \Lambda}(\delta)}
\newcommand{\clpp}{\hat{\mathcal{L}}(\delta)}

\begin{proof}[Proof of the upper bound in \Cref{thm:main}]
    Let $g$ be a function as in the statement, let $n\in \NN$, and let $(p,H)$ be an apportionment instance with $n$ states. We claim that there exist at most $2n-1$ lines in $\mathcal{L}(p,H)$ that intersect with $\lambda_H$. If true, this implies that we can reduce the line arrangement $\mathcal{L}(p,H)$ of some apportionment instance $(p,H)$ to a line arrangement $\mathcal{L}'$ with at most $2n-1$ lines and such that some $k$-level for $k \in [2n-1]$ has the same complexity as the $(H-1)$-level in $\mathcal{L}(p,H)$. From the definition of $g$, we obtain that the complexity of the $(H-1)$-level in $\mathcal{L}(p,H)$ is bounded by $\calO(g(m))$. From \Cref{prop:BreakingPointsAndTies}, we conclude that this bound holds for the number of breaking points of $(p,H)$ as well.

    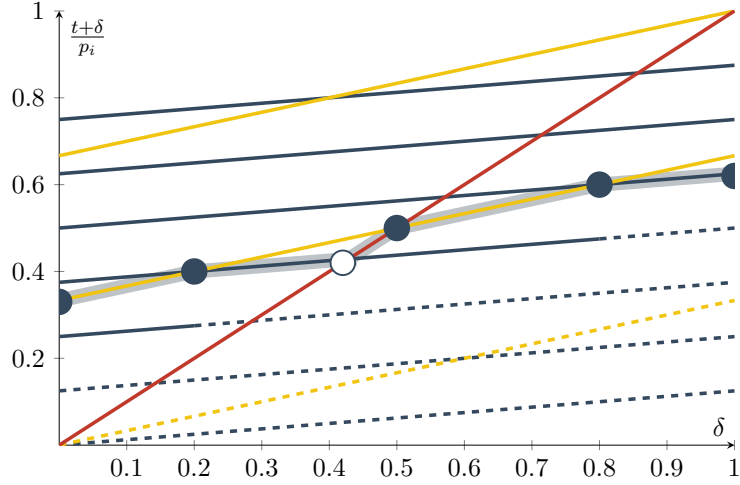
\begin{figure}[t]
\centering
    \scalebox{0.9}{
    \begin{tikzpicture}
    \node at (0,0) {}; 

    \begin{axis}[xlabel=$\delta$,ylabel=$\frac{t+\delta}{p_i}$, xmin=0,xmax=1,ymin=0,ymax=1, axis lines=center, height=8cm,width=0.7\textwidth]
    \def\H{6}

    \addplot[domain=0:0.2,color=grey,line width=6pt]{1/3 + x/3}; 
    \addplot[domain=0.2:0.44,color=grey,line width=6pt]{3/8 + x/8}; 
    \addplot[domain=0.42:0.5,color=grey,line width=6pt]{0/1 + x/1};
    \addplot[domain=0.5:0.8,color=grey,line width=6pt]{1/3 + x/3};
    \addplot[domain=0.8:1,color=grey,line width=6pt]{4/8 + x/8};

    \addplot[domain=0:1, color=blue,line width=1.5pt, dashed]{0/8 + x/8}; 
    \addplot[domain=0:1, color=blue,line width=1.5pt, dashed]{1/8 + x/8}; 
    \addplot[domain=0:0.2, color=blue,line width=1.5pt]{2/8 + x/8}; 
    \addplot[domain=0.2:1, color=blue,line width=1.5pt, dashed]{2/8 + x/8}; 
    \addplot[domain=0:0.8, color=blue,line width=1.5pt]{3/8 + x/8}; 
    \addplot[domain=0.8:1, color=blue,line width=1.5pt,dashed]{3/8 + x/8}; 
    \addplot[domain=0:1, color=blue,line width=1.5pt]{4/8 + x/8}; 
    \addplot[domain=0:1, color=blue,line width=1.5pt]{5/8 + x/8}; 
    \addplot[domain=0:1, color=blue,line width=1.5pt]{6/8 + x/8}; 

    \addplot[domain=0:1, color=green,line width=1.5pt,dashed]{0/3 + x/3}; 
    \addplot[domain=0:1, color=green,line width=1.5pt]{1/3 + x/3}; 
    \addplot[domain=0:1, color=green,line width=1.5pt]{2/3 + x/3}; 

    \addplot[domain=0:1, color=red,line width=1.5pt]{0/1 + x/1};

    \node[circle, thick, draw =blue, fill=blue] at (0,33){}; 
    \node[circle, thick, draw =blue, fill=blue] at (20,40){}; 
    \node[circle, thick, draw =blue, fill=white] at (42,42){}; 
    \node[circle, thick, draw =blue, fill=blue] at (50,50){}; 
    \node[circle, thick, draw =blue, fill=blue] at (80,60){}; 
    \node[circle, thick, draw =blue, fill=blue] at (100,62){}; 
    
    \end{axis}
    \end{tikzpicture}}
\caption{Illustration of the set $\clpp$ from the proof of the upper bound in \Cref{thm:main} via the same example as given in \Cref{fig:line-arrangement-b}. For each $\delta$, $\clpp$ are those functions from $\mathcal{L}$ for which there exists a function with higher index that is included in $\mathcal{L}_{\leq H}(\delta')$ for some $\delta' \in [0,\delta]$. We illustrate $\clpp$ by dashed lines. The important property of this set is that once a line is included for some $\delta$, it will not intersect with the $(H-1)$-level for any $\delta' \geq \delta$.}
\label{fig:fixedLines}
\end{figure}

    We use a potential function argument to show that the number of lines in $\mathcal{L}$ intersecting with $\lambda_H$ is upper bounded and prove the claim. For $\delta \in [0,1]$, we define $\clpp$ as those functions from $\mathcal{L}$ for which there exists a function with higher index that is included in $\mathcal{L}_{\leq H}(\delta')$ for some $\delta' \in [0,\delta]$, i.e., 

    $$\clpp =\{\ell_{i,t'} \mid \exists \; t > t', \delta' \leq \delta : \ell_{i,t} \in \mathcal{L}_{\leq H}(\delta')\}.$$

    Note that $\clpp$ is monotone by definition, i.e., $\clpp \subseteq \hat{\mathcal{L}}(\delta')$ for all $\delta \in [0,1], \delta' \in [\delta,1]$. In the following, we show two further observations: 
    \begin{enumerate}[label=(\roman*)]
        \item Lines in $\clpp$ are \emph{fixed} in the sense that they are included in any $\mathcal{L}_{< H}(\delta')$ for $\delta' \in [\delta,1]$. Formally, $\clpp \subseteq \mathcal{L}_{< H}(\delta')$ for all $\delta \in [0,1], \delta' \in [\delta,1]$.
        \item The size of $\clpp$ is bounded from both sides, i.e., $H -n \leq |\clpp| \leq H-1$ for all $\delta \in [0,1]$. 
    \end{enumerate}
    
    We start by proving (i). Let $\delta \in [0,1], i \in [n]$, and $t \in \{0,\dots,H-1\}$ be such that $\ell_{i,t}(\delta) \leq \lambda_H(\delta)$. We claim that for any $\delta' \in [\delta,1]$, $\ell_{i,0}(\delta'), \dots, \ell_{i,t-1}(\delta')$ are strictly below $\lambda_H(\delta')$. For $\delta' = \delta$, this is true because $\ell_{i,t'}(\delta) < \ell_{i,t}(\delta)$ for any $t' <t$. For $\delta' \in (\delta,1]$ this is true since (a) all functions in $\mathcal{L}$ are increasing, thus, $\lambda_H$ is increasing, and (b) $\ell_{i,t'}(\delta') \leq \ell_{i,t'}(1) \leq \ell_{i,t}(0) \leq \ell_{i,t}(\delta)$ for all $t' < t$. 
    
    For (ii), note that $|\mathcal{L}_{\leq H}(0)| \geq H$ and each state $i \in [n]$ can have at most one line in $\mathcal{L}_{\leq H}(0)\setminus \hat{\mathcal{L}}(0)$, thus $|\hat{\mathcal{L}}(0)| \geq H-n$. Moreover, by (i) we know that $\hat{\mathcal{L}}(1) \subseteq \mathcal{L}_{<H}(1)$, where the cardinality of the latter set is upper bounded $H-1$. Statement (ii) then follows by the monotonicity of $\clpp$.

    We now turn to bounding the number of different lines intersecting with $\lambda_H$. By (i), for any $\delta \in [0,1]$, any line that intersects with $\lambda_H$ in $\delta$ is in particular included in $\mathcal{L}_{\leq H}(\delta)$ but not in $\clpp$. Hence, the total number of lines in $\mathcal{L}$ intersecting with $\lambda_H$ is upper bounded by \begin{equation*}
    \Big|\bigcup_{\delta \in [0,1]} \mathcal{L}_{\leq H}(\delta) \setminus \clpp\Big|.
    \label{eq:boundOnLines}
    \end{equation*}
    For some line $\ell_{i,t}, i \in [n], t \in \{0,\dots,H-1\}$, let $\delta_0 \in [0,1]$ be the smallest value such that $\ell_{i,t} \in \mathcal{L}_{\leq H}(\delta_0) \setminus \hat{\mathcal{L}}(\delta_0)$. This is also the first time that $\ell_{i,t}$ is included in $\mathcal{L}_{\leq H} (\delta)$ (because $\clpp$ is monotone). Thus, $\delta_0$ is also the smallest value at which $\ell_{i,t-1}$ is included in $\hat{\mathcal{L}}(\delta)$ and the size of $\clpp$ increases. By (ii), this can happen at most $n-1$ times. Moreover, by definition it clearly holds that $|\mathcal{L}_{\leq H}(0) \setminus \hat{\mathcal{L}}(0)| \leq n$. This yields an upper bound of $2n-1$ lines intersecting with $\lambda_H$ and finishes the proof of the claim.
\end{proof}

\paragraph{Computation of all outcomes.} Our results give rise to a polynomial-time algorithm for computing all outcomes returned by stationary divisor methods. 
Indeed, using the argument from the previous proof, we can reduce an apportionment instance to a line arrangement with $\mathcal{O}(n)$ lines by computing the outcome of two divisor methods (for $\delta=0$ and $\delta=1$). Then, applying an algorithm to find the breaking points due to \citet{edelsbrunner1986constructing} (later improved by \citet{chan1999remarks}) yields a running time of $\mathcal{O}(n^{4/3}\log^{1+\varepsilon}(n))$. 
After identifying the breaking points, \Cref{obs:lineSets} and \Cref{prop:BreakingPointsAndTies} yield a complete description of the space of outcomes.

\subsection{Lower Bound on the Number of Breaking Points}\label{subsec:lb-bp}

\begin{figure}[t]
\centering
\begin{subfigure}[b]{.45\textwidth}
    \centering
    \scalebox{0.9}{
    \begin{tikzpicture}
    
    \begin{axis}[xlabel=$\delta$, xmin=0,xmax=1,ymin=0,ymax=4, axis lines=center]
    \def\H{6}

    \addplot[domain=0:0.33,color=grey,line width=6pt]{1/(4/7) + x/(4/7)}; 
   \addplot[domain=0.33:0.66,color=grey,line width=6pt]{2/1 + x/1}; 
    \addplot[domain=0.66:1,color=grey,line width=6pt]{3/(11/8) + x/(11/8)};

        \addplot[domain=0:1, color=blue,line width=1.5pt]{1/(4/7) + x/(4/7)}; 
        \addplot[domain=0:1, color=green,line width=1.5pt]{2/1 + x/1}; 
        \addplot[domain=0:1, color=red,line width=1.5pt]{3/(11/8) + x/(11/8)}; 

    \node[circle, thick, draw =blue, fill=blue] at (33,233){}; 
    \node[circle, thick, draw =blue, fill=blue] at (66,266){}; 
    \end{axis}
    \end{tikzpicture}}
    \caption{Line arrangement with its $0$-level marked via thick light gray segments.}
    \label{fig:line-arrangement-lb-a}
\end{subfigure} \quad  \centering \begin{subfigure}[b]{.45\textwidth}
    \centering
    \scalebox{0.9}{
    \begin{tikzpicture}
    \node at (0,0) {}; 

    \begin{axis}[xlabel=$\delta$,ylabel=$\frac{t+\delta}{p_i}$,ymin=0,ymax=4, xmin=0,xmax=1, axis lines=center,y label style={at={(axis description cs:0.1,1.1)},anchor=north}]

      \addplot[domain=0:0.33,color=grey,line width=6pt]{1/(4/7) + x/(4/7)}; 
   \addplot[domain=0.33:0.66,color=grey,line width=6pt]{2/1 + x/1}; 
    \addplot[domain=0.66:1,color=grey,line width=6pt]{3/(11/8) + x/(11/8)};

    \foreach \t in {0,...,5}{
        \addplot[domain=0:1, color=blue,line width=1.5pt]{\t/(4/7) + x/(4/7)}; 
        \addplot[domain=0:1, color=green,line width=1.5pt]{\t/1 + x/1}; 
        \addplot[domain=0:1, color=red,line width=1.5pt]{\t/(11/8) + x/(11/8)}; 
    }

         \addplot[domain=0:1, color=blue,line width=1.5pt]{1/(4/7) + x/(4/7)}; 
        \addplot[domain=0:1, color=green,line width=1.5pt]{2/1 + x/1}; 
        \addplot[domain=0:1, color=red,line width=1.5pt]{3/(11/8) + x/(11/8)}; 

    \node[circle, thick, draw =blue, fill=blue] at (33,23.3){}; 
    \node[circle, thick, draw =blue, fill=blue] at (66,26.6){}; 
    
    \end{axis}
    \end{tikzpicture}}
    \caption{Line arrangement corresponding to the instance $(p=(4/7,1,11/8),H=7)$.}
    \label{fig:line-arrangement-lb-b}
\end{subfigure}
\caption{Illustration of the construction of a line arrangement corresponding to an apportionment instance from an arbitrary line arrangement satisfying the conditions in \Cref{lem:wlog-form-of-arrangement}, such that the $k$-level of the original arrangement corresponds to the $(H-1)$-level of the new one. For illustration purposes, one of the slopes has been chosen out of the range specified in \Cref{lem:wlog-form-of-arrangement} (smaller than $1$).}
\label{fig:line-arrangement-lb}
\end{figure}
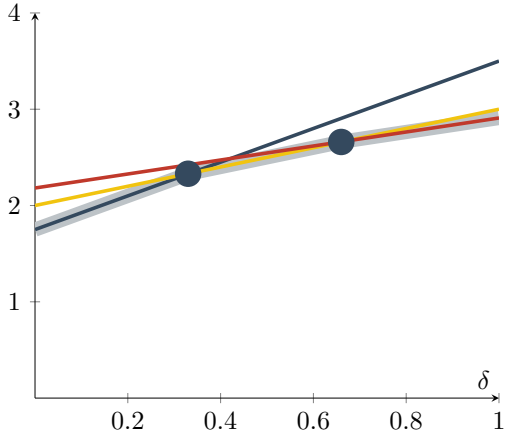
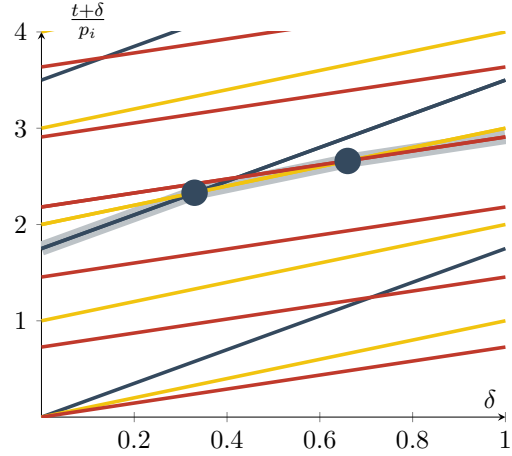

In this section, we prove the lower bound on the worst-case number of breaking points stated in \Cref{thm:main}.  
To do so, we first show that, w.l.o.g., any line arrangement has a specific form.

\begin{restatable}{lemma}{lemwlogFormArrangement} \label{lem:wlog-form-of-arrangement}
    Let $\mathcal{L}$ be a line arrangement of $n$ lines in general position and rational slopes. There exists another line arrangement $\mathcal{L}' = \{\ell_i(\delta) = m_i \delta + c_i, i \in [n]\}$ such that: 
    \begin{enumerate}[label=(\roman*)]
        \item For all $i \in [n]$ it holds that $m_i \in (1,2)$; \label{en:con1}
        \item For all $i \in [n]$ it holds that $\frac{c_i}{m_i} \in \mathbb{N}$; \label{en:con2}
        \item For any $k \in [n]$ the complexity of the $k$-level in $\mathcal{L}$ equals the complexity of the $k$-level in $\mathcal{L}'$. \label{en:con3}
    \end{enumerate}
\end{restatable}

Building upon \Cref{lem:wlog-form-of-arrangement}, whose proof can be found in \Cref{app:lem-wlog-form-of-arrangement}, we construct an apportionment instance from any line arrangement with the number of breaking points being equal up to a factor of $2$ to the complexity of some $k$-level. An illustration of this construction is shown in \Cref{fig:line-arrangement-lb}.

\begin{proof}[Proof of the lower bound in \Cref{thm:main}.]
Let $h$ be a function as in the statement and $n\in \NN$. Let $\mathcal{L} = \{\ell_i(\delta) = m_i \delta + c_i: i \in [n]\}$ be a line arrangement and $k \in [n]$ such that the complexity of the $k$-level of $\mathcal{L}$ is $z=\Omega(h(n))$. We assume without loss of generality that every line in $\mathcal{L}$ intersects with the $k$-level; otherwise, we could just remove such a line and adjust $k$ if necessary. By \Cref{lem:wlog-form-of-arrangement} we can assume without loss of generality that $m_i \in (1,2)$ and $c_i/m_i \in \mathbb{N}$ for all $i \in [n]$. We construct the apportionment instance as follows:\footnote{Note that, even though these populations might not be integers, they can be scaled for this purpose without affecting the breaking points.} Set $p_i = 1/m_i$ for all $i \in [n]$ and $H = \sum_{i \in [n]} \frac{c_i}{m_i} + k + 1$.

We aim to show that the $(H-1)$-level of $\mathcal{L}(p,H)$ equals the $k$-level of $\mathcal{L}$. We first show that no two lines in $\mathcal{L}(p,H)$ that correspond to the same state can both intersect with the $k$-level of $\mathcal{L}$, which we denote by $\gamma_k(\delta)$. Assume for contradiction that $\ell_{i,t}(\delta)$ and $\ell_{i,t'}(\delta)$ with $t,t' \in \mathbb{N}_0$, $t \neq t'$ both intersect with $\gamma_k(\delta)$ at $\delta_1, \delta_2 \in [0,1], \delta_1 < \delta_2$, respectively. In the following, we derive a contradiction by applying the facts that $|\ell_{i,t}(\delta) - \ell_{i,t'}(\delta)| \geq m_i > 1$, $\delta_2 - \delta_1 <  \gamma_k(\delta_2) - \gamma_k(\delta_1) < 2(\delta_2 - \delta_1)$, and $\delta_2 - \delta_1 < \ell_{i,t}(\delta_2)- \ell_{i,t}(\delta_1) < 2(\delta_2 - \delta_1)$, which hold since all $\ell_{i,t}, \ell_{i,t'}$, and $\gamma_k(\delta)$ are (piecewise) linear functions with slopes in $(1,2)$ and the images of $\ell_{i,t}$ and $\ell_{i,t'}$ intersect in at most one point. First, assume $t<t'$. Then, $$2 (\delta_2 - \delta_1) \leq 1 + \delta_2 - \delta_1 < 1 + \ell_{i,t}(\delta_2) - \ell_{i,t}(\delta_1)< \ell_{i,t'}(\delta_2) - \ell_{i,t}(\delta_1) = \gamma_k(\delta_2) - \gamma_k(\delta_1) < 2 (\delta_2 - \delta_1),$$ yielding a contradiction. Otherwise, assume $t > t'$. Then, $$ \delta_2 - \delta_1 < \gamma_k(\delta_2) - \gamma_k(\delta_1)= \ell_{i,t'}(\delta_2) - \ell_{i,t}(\delta_1) < \ell_{i,t}(\delta_2) -1  - \ell_{i,t}(\delta_1) < 2(\delta_2 - \delta_1) - 1 \leq \delta_2 - \delta_1,$$ again, yielding a contradiction. 

Now, we show that for each $i \in [n]$, exactly the line $\ell_{i,c_i/m_i}$ intersects with the $(H-1)$-level. Note that each line of the form $m_i \delta + c_i$ for each $i \in [n]$ from $\mathcal{L}$ equals exactly the line $\ell_{i,c_i/m_i}(\delta)$ from $\mathcal{L}(p,H)$. Therefore, we know for each $\ell_{i,c_i/m_i}$ that it intersects with $\gamma_k(\delta)$. By the observation above, this implies that for each $i \in [n]$, the lines $\ell_{i,t}$ for $t \in \{0,\dots,c_i/m_i\}$ are below $\gamma_k(\delta)$ for the entire interval $[0,1]$ and $\ell_{i,t}$ for $t \in \{c_i/m_i + 1, \dots, H-1\}$ are above $\gamma_k(\delta)$ for the entire interval $[0,1]$. By definition of $H$ this implies that $\gamma_k(\delta)$ equals the $(H-1)$-level of $\mathcal{L}(p,H)$. Hence, $z$ is the number of vertices on the $(H-1)$-level of $\mathcal{L}(p,H)$.
Observe that any concave vertex at $\delta$ is such that $|f(p,H;\delta)|>1$: For any pair of states $i,j\in [n]$ such that $\ell_{i,c_i/m_i}(\delta) = \ell_{j,c_j/m_j}(\delta) = \gamma_k(\delta)$, 
there must be apportionment vectors $x,y \in f(p,H;\delta)$ such that $x_{i} = c_{i}/m_{i} = y_{i}-1$ and $x_{j}-1 = c_{j}/m_{j} = y_{j}$.
Therefore, if at least $\lceil z/2 \rceil$ of the vertices on the $(H-1)$-level of $\mathcal{L}(p,H)$ are concave, then we are done due to \Cref{prop:BreakingPointsAndTies}. Otherwise, note that every convex vertex of this level corresponds to a concave vertex of the $(H-2)$-level of $\mathcal{L}(p,H)$, and thus to a breaking point of the instance $(p,H-1)$. Thus, in this case, the instance $(p,H-1)$ has at least $\lceil z/2 \rceil$ breaking points due to \Cref{prop:BreakingPointsAndTies}.
\end{proof}
 
We remark that the apportionment constructed in the previous proof may require a large number of seats $H$. Bounding the number of breaking points in both $H$ and $n$ remains an intriguing direction for future work. 

{
\subsection{Beyond Stationary Divisor Methods} \label{sec:power}
In general, one can consider rounding rules $\llbracket \cdot \rrbracket$ given by values~$\delta_t\in [0,1]$ for each $t\in \NN_0$ such that for $r\in [t,t+1]$ we have $\llbracket r \rrbracket = \{t\}$ if $r<t+\delta_t$, $\llbracket r \rrbracket = \{t+1\}$ if $r>t+\delta_t$, and $\llbracket r \rrbracket = \{t,t+1\}$ if $r=t+\delta_t$. These rules give rise to non-stationary divisor methods, for which the number of breaking points may become exponential in the number of districts $n$. To see this, we can consider the same instance as in the proof of \Cref{obs:number-solutions}, with $p_i=2i-1$ for every $i\in [n]$ and $H=\lfloor n^2/2\rfloor$. For any $S\subset [n]$ with $|S|=\lfloor n/2\rfloor$, one can obtain the apportionment vector $x\in \NN^n_0$ given by
\[
    x_i = \begin{cases}
        i & \text{if } i\in S,\\
        i-1 & \text{otherwise,}
    \end{cases}
    \qquad \text{for every } i\in [n]
\]
by setting $\delta_i<0.5$ for every $i\in S$ and $\delta_i>0.5$ for every $i\in\NN \setminus S$. Nevertheless, our upper bound on the number of breaking points still holds for other families of divisor methods that exhibit the same majorization property discussed at the beginning of this section, such as the family of \textit{power-mean} divisor methods. This family uses rounding rules given by $\delta_t=\big(\frac{t^q}{2}+\frac{(t+1)^q}{2}\big)^{1/q}-t$ for some $q\in \RR$, and one can define the breaking points analogously as before when~$q$ varies.\footnote{The value $\delta_t$ is not well defined when $q=0$, as well as when $t=0$ and $q<0$. When $q=0$, the standard definition \citep{marshall2002majorization} is $\delta_t = \lim_{q\to 0} \big(\frac{t^q}{2}+\frac{(t+1)^q}{2}\big)^{1/q}-t = \sqrt{t(t+1)}-t$. When $q<0$, we consider $\delta_0=0$ to maintain monotonicity, since $\lim_{q\to 0^+} \big(\frac{0^q}{2}+\frac{1^q}{2}\big)^{1/q} = 0$.} Although this does not define a line arrangement anymore as varying~$q$ leads to curves instead, it is not hard to see that the curves defined in this way for $q>0$, as well as for $q<0$ and $t\geq 1$, are \textit{pseudolines}, i.e., any pair of them intersects at most once. 
Thus, the $\calO(n^{4/3})$ upper bound on the number of breaking points remains valid due to a result by \citet{tamaki2003characterization} on the complexity of the $k$-level in a pseudoline arrangement. Importantly, this family contains traditional divisor methods such as the Dean method (where $q=-1$) and the Huntington-Hill method (where $q=0$), the latter of which is currently used to apportion the House of Representatives of the United States. In \Cref{app:power-mean}, we prove that the curves are indeed pseudolines and argue why an adaptation of the proof of the upper bound stated in \Cref{thm:main} still holds.}

\section{Randomization and Divisor Methods}\label{sec:randomization-variable}

In the deterministic apportionment setting, population monotonicity is essentially incompatible with quota compliance.
\citet{balinski2010fair} showed that, under mild assumptions, divisor methods are the unique population-monotone methods.
Among them, the Jefferson/D'Hondt method is the unique one satisfying the lower quota axiom but is known to violate upper quota compliance, whereas the Adams method is the unique one satisfying upper quota compliance but is known to violate lower quota compliance.

In \Cref{subsec:deviations-quota}, we first study \textit{how large} these deviations from quota can be for any $\delta\in [0,1]$. The negative results in this regard, showing a worst-case deviation of almost $H$ seats, motivate the incorporation of randomization. In particular, we explore the possibility of defining a random variable $\delta\in [0,1]$ in a way that, for the $\delta$-divisor method, a smaller deviation from quota is achieved in expectation. Even though constant worst-case deviations can be achieved when the ratio between the populations of two states is constant, in general, it turns out that these deviations are still linear in $H$ for any randomized stationary divisor method. Hence, in \Cref{subsec:variable-H} we take one step further and relax the requirement of exactly fulfilling the house size. We present a simple method, that can be seen as a mixture between a divisor method and the Hamilton method, that satisfies quota compliance, population monotonicity, and ex-ante proportionality. We further give both ex-post and probabilistic bounds on the deviation from the house size.

\subsection{Deviation Bounds from Quota}
\label{subsec:deviations-quota}

It is well known that divisor methods satisfy several desirable axioms, including the strong population monotonicity property, but fail to satisfy quota compliance \citep{balinski2010fair}. In \Cref{sec:pop-mon}, we showed that every stationary divisor method is such that for any instance, all states either receive at least their lower quota or at most their upper quota but may deviate from the other one. In this section, we further study \textit{how large} this deviation can be, i.e., how far the number of seats allocated to some state can be from its quota. Since this deviation cannot be larger than the total number of seats $H$, we take this number as given and ask the following: What is the largest difference that can occur between the number of seats allocated to a state and its quota, over all possible number of states $n$ and all population vectors $p\in \NN^n$?

Our first proposition answers this question: The maximum deviation is arbitrarily close to $H-1$ when $\delta=0$ and $H$ when $\delta>0$.

\begin{restatable}{proposition}{propdetdevquota}\label{prop:det-dev-quota}
    Consider $H\in \NN$. Then, for every $\varepsilon>0$ the following hold:
    \begin{enumerate}[label=(\roman*)]
        \item There exist $n\in \NN,~ p\in \NN^n$, and $i\in [n]$ such that 
        $| x_i - q_i | \geq H-1-\varepsilon$ for every $x \in f(p,H;0)$.\label{prop:det-dev-quota-i}
        \item For every $\delta\in (0,1]$, there exist $n\in \NN,~ p\in \NN^n$, and $i\in [n]$ such that 
        $| x_i - q_i | \geq H-\varepsilon$ for every $x \in f(p,H;\delta)$.\label{prop:det-dev-quota-ii}
    \end{enumerate}
\end{restatable}

The proof is deferred to \Cref{app:prop-det-dev-quota}. For the case $\delta=0$, the deviation of $H-1$ is achieved on an instance with one state with a large population and $H-1$ states with a small population, as the divisor method gives one seat to all states. When $\delta>0$, we take values $n$ and $M$ with $1\ll M\ll n$ and consider one state with population $n-1$ and $n-1$ states with population $M$. For carefully chosen $n$ and $M$, the $\delta$-divisor method assigns all seats to the state with population $n-1$, even though its quota is arbitrarily close to $0$.

Since the aforementioned worst-case instances are tailored for specific values of $\delta$, randomizing over $\delta\in [0,1]$ and returning the stationary divisor method corresponding to the realized $\delta$ arises as a natural attempt to lower the expected deviation from quota while keeping all ex-post guarantees of divisor methods. When $\delta\sim G$, we call such method the $G$-randomized divisor method $F^B$, where $B$ corresponds to a tie-breaking distribution, i.e., a distribution over subsets of apportionment vectors that are output by the method in case they are more than one. The following proposition, whose proof can be found in \Cref{app:prop-rand-dev-quota}, combines the previous worst-case instances appropriately to provide a lower bound for the expected deviation from the quota of such methods, that approaches $H/2$ as $H$ grows regardless of the tie-breaking distribution.

\begin{restatable}{proposition}{propranddevquota}\label{prop:rand-dev-quota}
    Let $G$ be an arbitrary probability distribution over $[0,1]$,
    $B$ be an arbitrary probability distribution over subsets of $\NN^n_0$, and $H\in \NN$ be arbitrary. Then, for every $\varepsilon>0$ there exist $n\in \NN,~ p\in \NN^n$, and $i\in [n]$ such that the $G$-randomized divisor method $F^B$ satisfies $| \EX(F^B_i(p,H)) - q_i | \geq \left(1-1/(2H-1)\right)H/2 - \varepsilon.$
\end{restatable}

This lower bound can be almost matched for large $H$ with a very simple randomized method, that runs either the Adams method or the Jefferson/D'Hondt method, each with probability $1/2$, and breaks ties arbitrarily. We formally state this in the following proposition. The proof, based on the fact that the Adams method satisfies upper quota and the Jefferson/D'Hondt method satisfies lower quota, is deferred to \Cref{app:prop-jeff-adams}.

\begin{restatable}{proposition}{propjeffadams}\label{prop:jeff-adams}
    Let $G\sim\normalfont\text{Bernoulli}(1/2)$, $B$ be an arbitrary probability distribution over subsets of $\NN^n_0$, and $H\in \NN$ be arbitrary. Then, for every $n\in \NN$, every $p\in \NN^n$, and every $i\in [n]$, the $G$-randomized divisor method $F^B$ satisfies $| \EX(F^B_i(p,H)) - q_i | < (H+1)/2.$
\end{restatable}

When parameterizing on the population vector, we can also obtain a bound on the deviation between the seats assigned to a state and the state's quota in terms of the ratio between the state's population and the minimum population among any state. This bound is minimized when $\EX(\delta)=1/2$, thus it is, in particular, valid for the deterministic Webster/Saint-Lagu\"e method as well.

\begin{restatable}{proposition}{propdevquotafixedpop}\label{prop:dev-quota-fixed-pop}
    Let $n\in \NN,~ p\in \NN^n,~ H\in \NN$, and $i\in [n]$ be arbitrary. Let also $G$ be an arbitrary probability distribution in $[0,1]$ such that for $\delta\sim G$ we have that $\EX(\delta) = 1/2$, and let $B$ be an arbitrary probability distribution over subsets of $\NN^n_0$. Then, the $G$-randomized divisor method $F^B$ satisfies
    \[
        | \EX(F^B_i(p,H)) - q_i | \leq \frac{1}{2}\left( \frac{p_i}{\min_{j\in[n]}p_j} +1 \right).
    \]
\end{restatable}

We defer the proof, which relies on providing bounds for the feasible multipliers that lead to an apportionment vector for each $\delta\in [0,1]$, to \Cref{app:prop-dev-quota-fixed-pop}.
The proposition implies that, whenever the ratio between the populations of every pair of states is bounded by a constant $r$, any randomized stationary divisor method with an expected shift equal to $1/2$ provides an expected deviation from the quota of at most $(r+1)/2$ for every state.

\subsection{Methods with Variable House Size}
\label{subsec:variable-H}

In this section, we relax our notion of a method with fixed house size (i.e., $\sum_{i=1}^{n}x_i=H$ for every $p,H$, and $x\in f(p,H)$) in order to allow for variable house size. 
That is, we do not ask for the apportionments $x \in f(p,H)$ to satisfy $\sum_{i=1}^{n}x_i=H$ for every instance $(p,H)$.
The notions of quota compliance, population monotonicity, and ex-ante proportionality naturally extend.

We now introduce a family of methods with variable house size, which we call \emph{randomized fixed-divisor methods}. These methods, when appropriately defined, satisfy population monotonicity, quota compliance, and ex-ante proportionality. Thus, they are able to overcome well-known impossibilities for deterministic methods with fixed house size. Further, specific methods provide additional deviation guarantees from the house size along with probabilistic bounds on this deviation.

A randomized fixed-divisor method is any method that randomizes on a specific set of 
methods that we call \emph{fixed-divisor methods}. We introduce some additional notation to define them. A \textit{signpost sequence} is a function $s: \NN_0 \to \RR_+$ satisfying two properties: (i) $s(t)\in [t,t+1]$ for every $t\in \NN_0$ and (ii) if $s(t')=t'$ for some $t'\in \NN_0$ then $s(t)<t+1$ for every $t\in \NN_0$ and, analogously, if $s(t')=t'+1$ for some $t'\in \NN_0$ then $s(t)>t$ for every $t\in \NN_0$. For a signpost sequence $s$ and a value $t\in \RR_{+}$, we let $N_s(t)=\{k\geq 0: s(k)< t\}$ denote the number of elements in the sequence that are strictly smaller than $t$.
A fixed-divisor method with signpost sequences $s_i(0),s_i(1),s_i(2),\ldots$ for every $i\in [n]$ receives a population vector $p=(p_1,\ldots,p_n)$ and a house size $H$, and returns $x=(x_1,\ldots,x_n)$ with
$x_i=N_{s_i}\left(q_i\right) \text{ for every } i \in [n].$ 
Note that fixed-divisor methods output a single apportionment vector for every instance: if $f$ is a fixed-divisor method, then $|f(p,H)|=1$ for every $p$ and $H$.\footnote{This is because the need to break ties to exactly match the house size is no longer an issue.} We refer to the vector output by $f$ for an instance $(p,H)$ directly as $f(p,H)$ (instead of $x\in f(p,H)$) to keep the notation simple.
Note that fixed-divisor methods differ from divisor methods studied in \Cref{sec:pop-mon} in two additional ways.
First and most importantly, the value up to which the signposts of a state are counted in fixed-divisor methods is---as the name suggests---fixed to the quota of the state; in divisor methods, this value is set \textit{ex-post} such that the total number of assigned seats is $H$. Second, in fixed-divisor methods, we allow for state-specific signpost sequences. 

Observe that, by linearity of expectation, for any ex-ante proportional randomized fixed-divisor method $F$, 
any population vector $p\in \NN^n$, and any house size $H\in \NN$, we have $\EE\left(\sum_{i=1}^n F_i(p,H)\right)=H$.
Moreover, if a randomized fixed-divisor method $F$ is quota-compliant, then for every $p$ and $H$, 
\[
    \PP\bigg(\bigg|\sum_{i=1}^{n} F_i(p,H) - H \bigg| \leq \max\bigg\{ H-\sum_{i=1}^{n} \left\lfloor q_i \right\rfloor, \sum_{i=1}^{n} \left\lceil q_i \right\rceil - H \bigg\} \bigg) = 1,
\]
which in particular yields $\PP\left(\big|\sum_{i=1}^{n} F_i(p,H) - H \big| < n\right) =1$. 
In simple words, a randomized fixed-divisor method that is ex-ante proportional and quota-compliant can deviate by at most $n-1$ seats from the original house size $H$ and meets $H$ in expectation.\footnote{This is, in turn, valid for any randomized method with variable house size that outputs a single vector for every instance.}

A randomized fixed-divisor method is fully given by the distribution of the signposts: it samples $s_i(k)$ for $i\in[n]$ and $k\in \NN_0$ from this distribution and runs the fixed-divisor method with this signpost sequence on the corresponding input $(p,H)$. 
\Cref{fig:random-fs-method} illustrates the application of a randomized fixed-divisor method with all signposts defined as $s_i(k)=k+\delta_i(k)$ with $\delta_i(k)$ being independent variables with distribution $U[0,1]$ for every $i\in[n]$ and $k\in \NN_0$.
We now state our main result regarding these methods. 
\vspace{-.2cm}
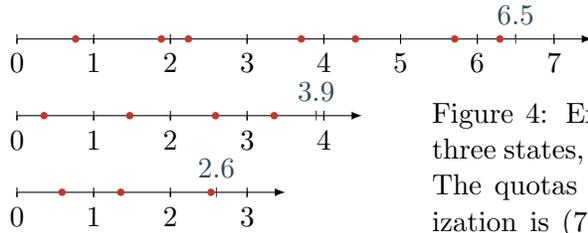
\begin{figure}[H]
\begin{minipage}[c]{0.34\textwidth}
\begin{flushleft}
\begin{tikzpicture}[scale=0.6]
\draw[-latex] (0,3.4) -- (12.75,3.4);
\foreach \x in {0,1,2,3,4,5,6,7}
\draw[shift={(1.7*\x,3.4)},color=black] (0pt,3pt) -- (0pt,-3pt);
\foreach \x in {0,1,2,3,4,5,6,7}
\draw[shift={(1.7*\x,3.4)},color=black] (0pt,0pt) -- (0pt,-3pt) node[below] 
{$\x$};
\draw[shift={(11.05,3.4)},color=blue] (0pt,3pt) -- (0pt,-3pt);
\draw[shift={(11.05,3.4)},color=blue] (0pt,0pt) -- (0pt,3pt) node[above] 
{$6.5$};

\draw[-latex] (0,1.7) -- (7.65,1.7);
\foreach \x in {0,1,2,3,4}
\draw[shift={(1.7*\x,1.7)},color=black] (0pt,3pt) -- (0pt,-3pt);
\foreach \x in {0,1,2,3,4}
\draw[shift={(1.7*\x,1.7)},color=black] (0pt,0pt) -- (0pt,-3pt) node[below] 
{$\x$};
\draw[shift={(6.63,1.7)},color=blue] (0pt,3pt) -- (0pt,-3pt);
\draw[shift={(6.63,1.7)},color=blue] (0pt,0pt) -- (0pt,3pt) node[above] 
{$3.9$};

\draw[-latex] (0,0) -- (5.95,0);
\foreach \x in {0,1,2,3}
\draw[shift={(1.7*\x,0)},color=black] (0pt,3pt) -- (0pt,-3pt);
\foreach \x in {0,1,2,3}
\draw[shift={(1.7*\x,0)},color=black] (0pt,0pt) -- (0pt,-3pt) node[below] 
{$\x$};
\draw[shift={(4.42,0)},color=blue] (0pt,3pt) -- (0pt,-3pt);
\draw[shift={(4.42,0)},color=blue] (0pt,0pt) -- (0pt,3pt) node[above] 
{$2.6$};

\foreach \x in {1.3, 3.2, 3.8, 6.3, 7.5, 9.7, 10.7}
\filldraw [red] (\x,3.4) circle (2pt);

\foreach \x in {0.6, 2.5, 4.4, 5.7}
\filldraw [red] (\x,1.7) circle (2pt);

\foreach \x in {1, 2.3, 4.3}
\filldraw [red] (\x,0) circle (2pt);
\end{tikzpicture}
\end{flushleft}
\hspace{3.5cm}

\end{minipage}
\begin{minipage}[c]{0.65\textwidth}
\vspace{1.4cm}

\caption{Example of a randomized fixed-divisor method with three states, populations $p=(50,30,20)$, and house size \mbox{$H=13$}. The quotas are $(6.5, 3.9, 2.6)$; the apportionment for this realization is $(7,4,3)$, with a deviation of one seat from $H=13$. 
Realizations of the signposts are denoted with red dots.}
\label{fig:random-fs-method}
\end{minipage}
\end{figure}

\vspace{-.2cm}
\begin{restatable}{theorem}{thmvariablesize}
\label{thm:variable-size}
    Let $\delta_i(k)$ be random variables with marginal distribution $U[0,1]$ for every $i\in[n]$ and $k\in \NN_0$. Then, the randomized fixed-divisor method $F$ with signpost sequences $(s_i(k))_{k\ge 0}$ defined as $s_i(k)=k+\delta_i(k)$ for every $i\in [n]$ is quota-compliant, population-monotone, and ex-ante proportional. 
    Furthermore, it is possible to define the distribution of the signpost sequences in a way that the corresponding randomized fixed-divisor method satisfies, in addition, the following two properties for every population vector $p=(p_1,\ldots,p_n)$, every house size $H$, and every $\Delta>0$:
    \[
        \PP\bigg(\bigg|\sum_{i=1}^{n} F_i(p,H) - H \bigg| > \frac{n+n\bmod 2}{2} \bigg) = 0,\quad
        \PP\bigg(\bigg|\sum_{i=1}^{n} F_i(p,H) - H\bigg| \geq \Delta \bigg) \leq 2\exp\bigg(-\frac{2\Delta^2}{n}\bigg).
    \]
\end{restatable}

Quota compliance and ex-ante proportionality follow rather easily from the definition of the method, while population monotonicity needs a slightly more careful analysis. However, the main challenge is to define the shifts $\delta_i$ such that the deviation from the house size is kept under control. For this purpose, we define these variables in pairs: $\delta_i$ is taken as a uniform, independently sampled random variable in $[0,1]$ if $i$ is odd, and as $1-\delta_{i-1}$ is $i$ is even. In this way, the marginal distribution of all variables is the same but their sum is restricted to lie between $(n-n\bmod 2)/2$ and $(n+n\bmod 2)/2$, which is key to guarantee a small deviation. The probabilistic bound follows from applying Hoeffding's concentration bound to associated random variables $Y_i\in \{0,1\}$ that take the value $1$ if state $i$ is allocated $\lceil q_i \rceil$ seats and $0$ otherwise. The construction of the variables $\{\delta_i\}_{i\in [n]}$ ensures that these variables $\{Y_i\}_{i\in [n]}$ are not positively correlated, which suffices to apply the concentration bound. We defer the formal proof to \Cref{app:thm-variable-size}.

The \textit{ex-post} deviation guarantee provided by the specific randomized fixed-divisor method constructed in the proof is almost the best possible among this class of methods.

\begin{restatable}{proposition}{proplbdevfixeddivisor}\label{prop:lb-dev-fixed-divisor}
    Let $f$ be a fixed-divisor method satisfying quota and let $n\in \NN$ be arbitrary. Then, for every $\varepsilon>0$ there exists $p\in \NN^n$ and $H\in \NN$ such that $|\sum_{i=1}^{n} f_i(p,H) - H | \geq n/2-1-\varepsilon.$
\end{restatable}

The proof of this result can be found in Appendix \ref{app:prop-lb-dev-fixed-divisor}. It relies on an adversarial construction of the population vector given the first signpost sequence of each state. Intuitively, if these signposts are such that $\sum_{i=1}^n s_i(0) \leq n/2$, we consider a population vector and a house size such that the quota of each state $i\in \{1,\ldots,n-1\}$ is slightly above $s_i(0)$, the fixed-divisor method assigns one seat to every state, and the deviation from the house size is roughly $n/2$. The population of the last state is taken to ensure an integral house size. The argument is analogous when $\sum_{i=1}^n s_i(0) > n/2$.

We finish this section by discussing further aspects of randomized fixed-divisor methods.

\paragraph{Connection to well-known deterministic methods with fixed house size.} Randomized fixed-divisor methods are a natural randomized version of widely-used deterministic methods that guarantee a (strict) subset of the properties that this method ensures. 
They can be seen as a randomization over divisor methods with a fixed multiplier $\lambda=H/P$ and state-specific rounding rules. In particular, when we ensure that every shift $\delta_i$ distributes uniformly in the interval $[0,1]$ as in the proof of \Cref{thm:variable-size}, these methods resemble uniform randomization over the family of stationary divisor methods studied in \Cref{sec:pop-mon}. Randomized fixed-divisor methods with such shifts may also be understood as a randomized variant of the Hamilton method, where each state first receives its lower quota and then, instead of assigning the remaining seats to the states with the largest remainders, each state receives one extra seat with a probability equal to its remainder.

\paragraph{Deviation from the house size.} \Cref{thm:variable-size} gives a probabilistic bound on the deviation between the total number of seats allocated via a randomized fixed-divisor method with suitable shifts and the house size $H$. This bound can be stated as follows: for every $\varepsilon>0$, $\PP(|H - \sum_{i=1}^{n} F_i(p,H)| \geq \sqrt{(n/2)\ln(2/\varepsilon)}) \leq \varepsilon.$
For instance, considering the number of seats and states taken into account for the US House of Representatives, with probability $0.8$ the deviation from the house size $H=435$ is at most $7$, and with probability $0.95$ is at most $9$.

\section{A Network Flow Approach to Monotonicity and Quota}\label{sec:HM}

In this section, we revisit the characterization of house-monotone and quota-compliant methods.
Recently, \citet{goelz2022apportionment} provided a characterization based on a matching polytope, while \citet{still1979} and \citet{balinski1979} had shown a characterization based on a non-polyhedral recursive construction. 
We provide a different characterization based on a simple network flow LP and show tight bounds of the size of the linear program needed.

To study house-monotone methods in a simpler way, we first show that we can actually restrict to methods outputting a single vector for every instance, as any house-monotone method can be easily expressed in terms of them. 
Formally, we define a \textit{solution} as a method $f$ such that $|f(p,H)|=1$ for every $p$ and $H$; we write $f(p,H)=x$ and $f(p,H) = \{x\}$ indistinctly when $f$ is a solution. For a solution $f$, it is easy to see that house monotonicity reduces to the following simple property.

\begin{restatable}{proposition}{prophmsolutions}\label{prop:hm-solutions}
    A solution $f$ is house-monotone if and only if for every $p \in \NN^n$ and $H\in \NN$ it holds that $f(p,H)\leq f(p,H+1)$.
\end{restatable}

Let $\calF^{\text{HM}}$ be the set of house-monotone methods and $\calF^{\text{Q}}$ be the set of quota-compliant methods. We obtain the following lemma, whose proof is deferred to \Cref{app:lem-hm-methods-solutions}.

\begin{restatable}{lemma}{lemhmmethodssolutions}\label{lem:hm-methods-solutions}
    For every $p\in \NN^n$ and $H\in \NN$,
    \[
        \{f(p,H): f \text{ is a house-monotone and quota-compliant solution}\} = \bigcup_{f \in \calF^{\text{HM}} \cap \calF^{Q}} f(p,H).
    \]
\end{restatable}

This lemma implies that, in order to characterize house-monotone and quota-compliant methods, it suffices to characterize house-monotone and quota-compliant solutions. We will focus on solutions throughout this section and state the implication of our characterization for methods at the end.

Given a population vector $p$ and an integer value $H$, consider the following linear program:
\begin{align}
\textstyle\sum_{i=1}^nx(i,t)&=1 \quad\quad\quad \text{ for every }t\in [H],\label{hm-seat}\\
\textstyle\sum_{\ell=1}^tx(i,\ell)&\ge \lfloor tp_i/P  \rfloor\; \text{ for every }i\in [n]\text{ and every }t\in [H],\label{hm-lower}\\ 
\textstyle\sum_{\ell=1}^tx(i,\ell)&\le \lceil tp_i/P  \rceil\; \text{ for every }i\in [n]\text{ and every }t\in [H]\label{hm-upper}\\ 
x(i,t)&\ge 0 \quad\quad\quad \text{ for every }i\in [n]\text{ and every }t\in [H].\label{hm-positive}
\end{align}
The key idea behind constructing this linear program is to model the procedure of assigning seats to states in a house-monotone and quota-compliant way.
Consider a feasible integer vector $x$ satisfying \eqref{hm-seat}-\eqref{hm-positive}.
In particular, we have that $x(i,t)\in \{0,1\}$
for every $i\in [n]$ and every $t\in [H]$, and constraint \eqref{hm-seat} guarantees that for every $t\in [H]$ there exists a unique $i_t\in [n]$ such that $x(i_t,t)=1$.
Consider the following apportionment solution: The first seat is assigned to state $i_1$, the second seat to state $i_2$, and so on, until the last seat $H$ is assigned to state $i_H$.
One by one, this assignment of the seats guarantees the number of each state's seats to be non-decreasing during the procedure, and, furthermore, constraints \eqref{hm-lower}-\eqref{hm-upper} guarantee that quota compliance is also satisfied.
Moreover, this linear program can be seen as the projection of a network flow linear program.
In \Cref{fig:flow}, we provide an example of the network flow associated with the linear program \eqref{hm-seat}-\eqref{hm-positive} in an instance with $n=3$ states and $H=6$ seats. 
We obtain the following lemma.
\begin{figure}[t]
    \centering
    \includegraphics[scale=0.42]{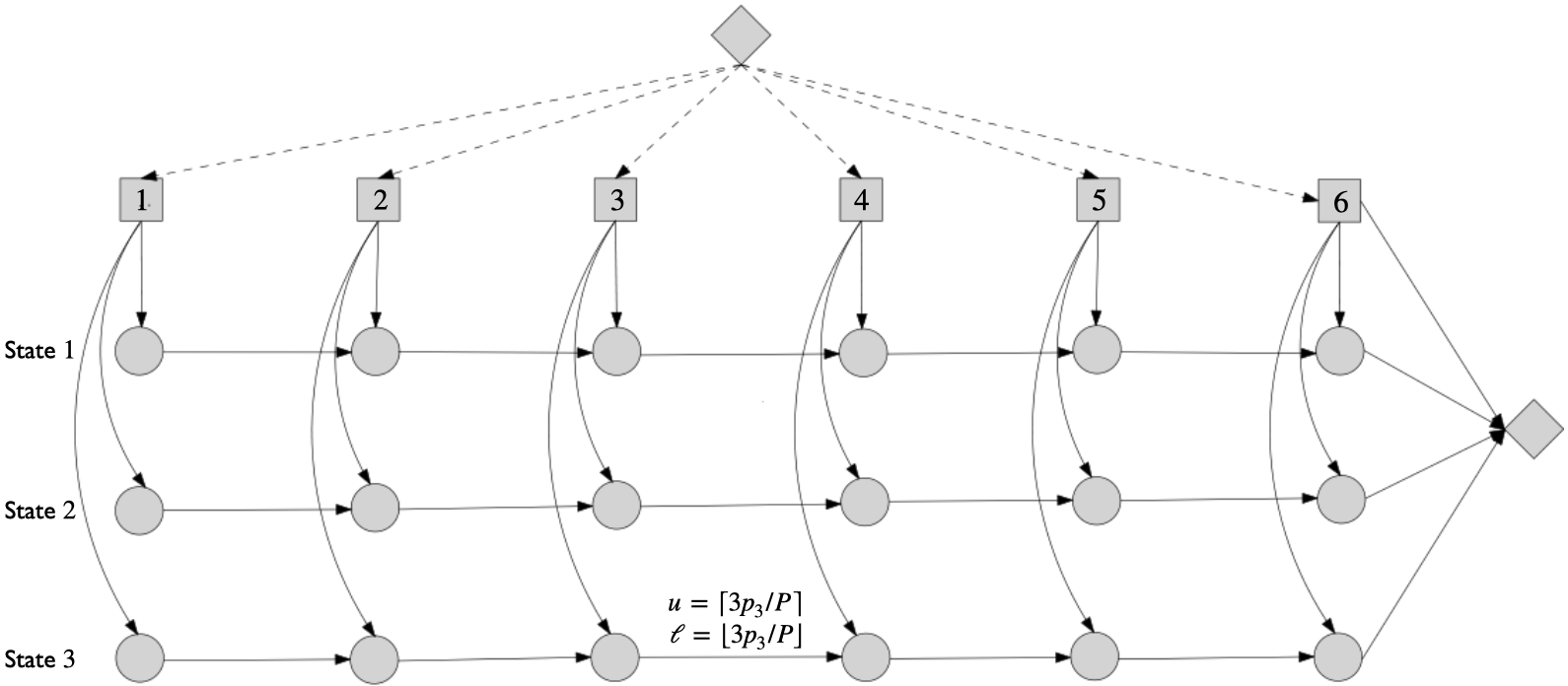}
    \caption{An example of the network flow formulation with $n=3$ states and $H=6$ seats. 
    The two rhombuses represent the source and the sink. 
    The source sends exactly one unit to each of the square nodes representing each of the $H=6$ seats, for which purpose we fix a lower and an upper capacity of exactly $1$ for the dashed edges. Then, for each of the three states, the circular nodes on the vertical layer corresponding to seat $t$ keep track of how many seats the state has been allocated up to that point. 
    The edges connecting the squared nodes to the circular nodes of the same vertical layer have a lower capacity of zero and an upper capacity of one, and every square node sends precisely one unit of flow to those circular nodes (representing the assignment of the seat to a state).
    The edge connecting two consecutive circular nodes in the same horizontal layer (representing the states) have lower and upper capacities of $\lfloor p_it/P\rfloor$ and $\lceil p_it/P\rceil$ for each state $i$ and seat $t$, respectively, to make sure that the allocation is quota-compliant.}
    \label{fig:flow}
\end{figure}
\begin{restatable}{lemma}{lemlpintegral}\label{lem:lp-integral}
For every instance $(p,H)$, the linear program \eqref{hm-seat}-\eqref{hm-positive} is integral.
\end{restatable}
The proof can be found in \Cref{app:lem-lp-integral}. 
The key observation is that \eqref{hm-seat}-\eqref{hm-positive} is the projection of a network flow linear program, and integrality is preserved after performing this projection. 
We provide a direct proof of the integrality based on total unimodularity.
Since every feasible solution of \eqref{hm-seat}-\eqref{hm-positive} is contained in the unit hypercube, \Cref{lem:lp-integral} implies that the set of extreme points in the linear program is exactly the set of integer feasible solutions.
We denote this set as $\calE(p,H)$.

Given a feasible solution $x$ of the linear program \eqref{hm-seat}-\eqref{hm-positive}, let $A_i(x,0)=0$ for each state $i\in [n]$, and when $H\ge 1$ let $A(x,H)$ be the vector such that 
$A_i(x,H)=\sum_{\ell=1}^Hx(i,\ell)$
for every $i\in [n]$.
In simple words, $A_i(x,H)$ is the number of seats that state $i$ has been allocated in the solution $x$ up to seat $H$.
Given a population vector $p$ and a house size $H$, we denote by $\calA(p,H)$ the set of apportionments for this instance that are generated by a house-monotone and quota-compliant solution, i.e.,
\[
    \calA(p,H) = \{f(p,H): f \text{ is a house-monotone and quota-compliant solution}\}.
\]
We now introduce a definition closely related to the one proposed by \citet{balinski1979,balinski2010fair} to characterize the necessary {\it lookahead} that a house-monotone and quota-compliant solution should consider when allocating each seat in order to ensure that it fulfills these properties when adding more seats.
Given an instance $(p,H)$, a non-negative integer vector $y=(y_1,\ldots,y_n)$ and an integer $k\ge 1$, we denote $S_k(p,H,y)=\{i\in [n]:\left\lfloor p_i(H+k)/P \right\rfloor>y_i\}$, and let $\tau(p,H,y)$ be defined as follows.
Let $k^{\star}(p,H,y)$ be the minimum integer $k\ge 1$ such that 
\begin{equation*}\label{eq:tau-def}
     \sum_{i\in S_k(p,H,y)} \bigg(\left\lfloor \frac{p_i(H+k)}{P} \right\rfloor - y_i \bigg) \geq k.
\end{equation*}
If $S_{k^{\star}}(p,H,y)\ne [n]$, let $\tau(p,H,y)=k^{\star}(p,H,y)$; otherwise, let $\tau(p,H,y)=1$.
Consider the function $\Phi$ defined recursively as follows: 
\begin{align*}
\Phi(p,0)&=0,\; \Phi(p,1)=\min\bigg\{ k \in \NN : \sum_{i=1}^n \left\lfloor \frac{p_ik}{P} \right\rfloor\geq k \bigg\},\\
\Phi(p,H+1)&=\max_{T\in [H]}\left\{T+\max\Big\{\tau(p,T,A(x,T)):x\in \calE(p,\Phi(p,H))\Big\}\right\} \text{ for every }H.
\end{align*}
The idea behind this definition is the following.
When $y=A(x,H)$ for some $x\in \calE(p,H)$ and the method has to allocate the next seat, it must consider the states $S_k(p,H,y)$ that will demand additional seats when adding $k$ extra seats, in order to be lower quota-compliant.
If the total number of demanded additional seats is strictly greater than $k$, the house-monotone method will not be quota-compliant with a house size $H+k$; the equality already forces the method to assign the next seat to a state in $S_k(p,H,y)$.
If this set of states is not $[n]$, we take $\tau(p,H,y)$ as the minimum $k$ for which the equality holds, which we denote as $k^*$, since for any $k>k^*$ the set $S_k(p,H,y)$ can only have additional states and thus considering quota compliance up to $H+k^*$ seats is restrictive enough.
If $S_k(p,H,y)=[n]$ for the smallest $k$ for which the equality holds, due to the monotonicity of this set in $k$ there is no restriction on which states can receive the next seat in order to have quota compliance, and thus including quota compliance constraints up to $H+1$ seats is enough.

The value $\Phi(p,H+1)$ captures the worst-case horizon---number of allocated seats $T$ for $T\leq H$, plus necessary additional seats at this point $\tau(p,T,A(x,T))$---that the method must consider in order to allocate the seat $H+1$ in a way that house monotonicity and quota compliance are not violated for any higher number of seats.
We now formally state our main result of this section.

\begin{restatable}{theorem}{thmhmlpcharacterization}\label{thm:hm-LP-characterization}
For every instance $(p,H)$, we have
$\calA(p,H)=\big\{A(x,H):x\in \calE(p,\Phi(p,H))\big\}.$
\end{restatable}

For a population vector $p$ and a house size $H$, this result 
states that all apportionments a house-monotone and quota-compliant solution can obtain are captured by the extreme points of \eqref{hm-seat}-\eqref{hm-positive} with a house size equal to $\Phi(p,H)$.
An apportionment solution obtained from any such extreme point $x$ can be implemented by assigning the $\ell$-th seat to the unique state $i_{\ell}$ such that $x(i_{\ell},\ell)=1$.

We defer the proof of \Cref{thm:hm-LP-characterization} to \Cref{app:HM-proof} and now address the natural questions that arise regarding the value $\Phi(p,H)$. 
The following proposition, whose proof can be found in \Cref{app:prop-phi}, provides three important properties of this function.

\begin{restatable}{proposition}{propphi}\label{prop:phi}
For every instance $(p,H)$ and positive integer $c$, the following holds:
\begin{enumerate}[label=(\alph*)]
    \item If $H\leq cP$, then $\Phi(p,H)\leq cP$.\label{phi-ub-cP}
    \item $\Phi(p,cP)=cP$.\label{phi-cP}
    \item If $H\geq 2$, then $\Phi(p,H)-H\le \max_{i\in [n]}\lceil P/p_i\rceil$.\label{phi-ub}
\end{enumerate}
\end{restatable}

Combining part \ref{phi-cP} of this proposition with \Cref{thm:hm-LP-characterization}, we conclude that $\calE(p,P)$ provides a full polyhedral description of all apportionments that are obtained from house-monotone and quota-compliant solutions up to a house size equal to $P=\sum_{j=1}^np_j$.
Part \ref{phi-ub} provides an upper bound on the lookahead value $\Phi(p,H)-H$, that describes the number of extra seats needed in the construction of the linear program \eqref{hm-seat}-\eqref{hm-positive} so that we can have a polyhedral description of all possible apportionments with house size $H$ that can be obtained by house-monotone and quota-compliant solutions.
In the worst case, this upper bound can be as large as $\calO(P)$; in the following proposition, we show that this is actually tight.

\begin{restatable}{proposition}{proptightness}\label{prop:tightness}
There is an instance $(p,H)$ with $\Phi(p,H)-H= \Omega(P)$.
\end{restatable}

The proof of this proposition can be found in \Cref{app:prop-tightness}; it consists of an example with $n=6$ states but it can be extended to an arbitrarily large number of states.
We remark that the worst-case lower bound $\Omega(P)$ for the value $\Phi(p,H)-H$ is achieved when we have one state with almost all the population, and the rest of the states are all tiny in comparison.
Typically, when the states represent districts or regions of a districting map, they are designed in a way that 
$\max_{i\in [n]}p_i/\min_{i\in [n]}p_i$ is a small value. 
In the case where this ratio is $\calO(1)$, the upper bound in \Cref{prop:phi} shows that $\Phi(p,H)-H=\calO(n)$, which in turns implies that $\Phi(p,H)= \calO(n)$ when $H=\Theta(n)$.
In particular, to find a polyhedral description of all the apportionments for a house size $H$ that can be obtained by house-monotone and quota-compliant solutions using \Cref{thm:hm-LP-characterization}, we need to write a linear program of size $\calO(nH)$, which is reasonable in practice.

Due to the correspondence between house-monotone methods and solutions, our theorem also fully characterizes the set of house-monotone and quota-compliant methods. Recall that $\calF^{\text{HM}}$ denotes the set of house-monotone methods and $\calF^{\text{Q}}$ denotes the set of quota-compliant methods. The following corollary is a direct consequence of \Cref{thm:hm-LP-characterization} and \Cref{lem:hm-methods-solutions}.

\begin{corollary}
    For every instance $(p,H)$, we have
\[ \bigcup_{f\in \calF^{\text{HM}} \cap \calF^{\text{Q}}} f(p,H) =\Big\{A(x,H):x\in \calE(p,\Phi(p,H))\Big\}.\]
\end{corollary}

Another consequence of \Cref{thm:hm-LP-characterization} is that we characterize the whole family of randomized methods that are house-monotone, quota-compliant, and ex-ante proportional up to a house size $H$ equal to the total population $P$.
We remark that this is also possible using the characterization by \citet{goelz2022apportionment}, but the simplicity of our network flow LP allows for a rather intuitive manner of recovering the apportionment solution.
Indeed, the quota $q$ can be easily mapped to a feasible solution of the linear program \eqref{hm-seat}-\eqref{hm-positive}.
Therefore, it can be written as a convex combination of the extreme points of the network flow LP. The randomized method that samples each of these extreme points with a probability equal to its coefficient in the convex combination is, by \Cref{thm:hm-LP-characterization}, a house-monotone, quota-compliant, and ex-ante proportional method. 

Formally, consider the family of randomized methods $\calM$ on the restricted domain in which the population vector $p$ and the house size $H$ satisfy $H\le P$.
We say that a method in $\calM$ is house-monotone, quota-compliant, and ex-ante proportional if the three properties hold for every $p$ and $H$ in their domain and denote by $\calM^{\star}$ this family of methods.
Given a method $F\in \calM$ and a population vector $p$, let $X_F:[n]\times [P]\to \{0,1\}$ be the random function such that $X_F(i,H)=F_i(p,H)-F_i(p,H-1)$ for every $H\in [P]$ and $i\in [n]$. 
For every $i\in [n]$ and every $t\in [P]$, let $Q(i,t)=p_i/P$.
Observe that $\sum_{i=1}^nQ(i,t)=\sum_{i=1}^np_i/P=1$ and $\sum_{\ell=1}^tQ(i,\ell)=p_it/P$.
Therefore, $Q$ is a feasible solution in the convex hull of $\calE(p,P)$.
For a set $S\subseteq \calE(p,P)$, we define
\[
    \Theta(S) = \left\{ \theta\in (0,1]^S: \sum_{x\in S}\theta_x=1,~ \sum_{x\in S}\theta_x x =Q\right\}
\] 
as the set of strictly positive coefficients such that the convex combination of the points in $S$ according to these coefficients is equal to $Q$.
We now let $\calS(p) = \{S\subseteq \calE(p,P): \Theta(S)\not=\emptyset\}$ be the set containing the subsets of $\calE(p,P)$ such that $Q$ can be obtained as a strictly positive convex combination of the elements in the subset.
\Cref{thm:hm-LP-characterization} implies the following result.

\begin{restatable}{theorem}{thmrandcharacterization}\label{thm:rand-characterization}
Let $F$ be a method in $\calM$.
Then, $F\in \calM^{\star}$ if and only if for every $p$ there exists $S_p\in\calS(p)$ and $\theta\in \Theta(S_p)$ such that for every $x\in S_p$ it holds $\PP(X_F=x)=\theta_x$.
\end{restatable}

We provide the proof and technical details of this result in Appendix \ref{app:rand-methods}.

\section{Discussion}

Our work underlines a hierarchy among the three studied classes of apportionment methods. 
First, we have studied in \Cref{sec:pop-mon} the well-established class of stationary divisor methods, showing that they have a small outcome space and every outcome satisfies lower or upper quota. However, this class is too restrictive to ensure strong ex-ante proportional outcomes (\Cref{subsec:deviations-quota}).
Second, for the class of population-monotone rules, we obtained strong ex-ante and ex-post proportionality guarantees if we allow to violate the house size (\Cref{subsec:variable-H}). Whether a similar result holds without violating the house size is open.
Lastly, the class of house-monotone rules is extremely rich. Even if we require ex-post quota, the space of ex-ante proportional rules can be large, as exemplified by our characterization (\Cref{sec:HM}). Identifying additional constraints capturing fairness or monotonicity that we can add to our network flow construction is an interesting direction for future work. 

Beyond these points, we believe that our structural insights can contribute to an ever-growing body of literature on generalizations of apportionment---such as committee elections, weighted fair allocation, or multidimensional apportionment---and best-of-both-worlds type guarantees for these settings, whose study has only begun recently.

\newpage
\appendix

\section{Proofs Deferred from Section \ref{sec:pop-mon}}

\subsection{Proof of \Cref{obs:convexity}}\label{app:obs-convexity}

\obsconvexity* 
    
    Let $p,H,\delta_1,\delta_2$ be as in the statement, and let $x\in f(p,H;\delta_1)\cap f(p,H;\delta_2)$ and $\delta\in [\delta_1,\delta_2]$ be arbitrary. Let $\lambda_1$ and $\lambda_2$ be multipliers corresponding to the output $x$ for the $\delta_1$- and $\delta_2$-divisor methods, respectively, i.e., $\lambda_1 \in \Lambda(x;\delta_1)$ and $\lambda_2\in \Lambda(x;\delta_2)$. 
    This is equivalent to
    \begin{align}
        x_i-1+\delta_1 & \leq \lambda_1 p_i \leq x_i+\delta_1 \quad \text{ for every } i\in [n], \label{eq:delta1}\\
        x_i-1+\delta_2 & \leq \lambda_2 p_i \leq x_i+\delta_2 \quad \text{ for every } i\in [n]. \label{eq:delta2}
    \end{align}
    We define
    \[
        \lambda = \lambda_1 + \frac{\delta - \delta_1}{\delta_2 - \delta_1}(\lambda_2 - \lambda_1).
    \]
    We then have, for every $i\in [n]$, that
    \begin{align*}
        \lambda p_i &= \lambda_1 p_i + \frac{\delta - \delta_1}{\delta_2 - \delta_1} (\lambda_2 - \lambda_1) p_i
         = \frac{\delta_2 - \delta}{\delta_2 - \delta_1}\lambda_1 p_i +  \frac{\delta - \delta_1}{\delta_2 - \delta_1}\lambda_2 p_i \\ 
        & \geq \frac{\delta_2 - \delta}{\delta_2 - \delta_1}(x_i - 1 + \delta_1) + \frac{\delta - \delta_1}{\delta_2 - \delta_1} (x_i - 1 + \delta_2)
        = x_i - 1 + \delta,
    \end{align*}
    where the inequality follows from \eqref{eq:delta1} and \eqref{eq:delta2}. Similarly, for each $i\in [n]$ we have that
        \begin{align*}
        \lambda p_i & = \lambda_1 p_i + \frac{\delta - \delta_1}{\delta_2 - \delta_1} (\lambda_2 - \lambda_1) p_i
         = \frac{\delta_2 - \delta}{\delta_2 - \delta_1}\lambda_1 p_i +  \frac{\delta - \delta_1}{\delta_2 - \delta_1}\lambda_2 p_i \\ 
        & \leq \frac{\delta_2 - \delta}{\delta_2 - \delta_1}(x_i + \delta_1) + \frac{\delta - \delta_1}{\delta_2 - \delta_1} (x_i + \delta_2)
        = x_i + \delta,
    \end{align*}
    where the inequality follows again from \eqref{eq:delta1} and \eqref{eq:delta2}. This shows that $x_i \in \llbracket \lambda p_i\rrbracket_{\delta}$. Since $\sum_{i=1}^n x_i=H$ follows immediately from the fact that $x\in f(p,H;\delta_1)\cap f(p,H;\delta_2)$, we conclude that $x \in f(p,H;\delta)$. \qed

\subsection{Proof of \Cref{prop:partition}}\label{app:prop-partition}

\quotaintervals* 

Let $(p,H) \in \NN^n\times \NN$ be arbitrary and consider
\begin{align*}
    \tau &= \min\{\delta\in [0,1]: \lambda_H(\delta)p_i \geq \lfloor q_i \rfloor \text{ for every } i\in [n] \},\\
    \bar{\tau} &= \max\{\delta\in [0,1]: \lambda_H(\delta)p_i \leq \lceil q_i \rceil \text{ for every } i\in [n] \}.
\end{align*}
We first observe that these values are well defined, i.e., that the sets over which the minimum and maximum are taken are non-empty. Indeed, suppose towards a contradiction that, for every $\delta\in [0,1]$, it holds that $\lambda_H(\delta)p_i < \lfloor q_i \rfloor$ for some $i\in [n]$. This yields that, for every $\delta\in [0,1]$, $\lambda_H(\delta)<H/P$. 
Applying this inequality for $\delta=1$, we obtain that
\[
    \sum_{i=1}^{n} | \{t\in \NN_0 : t + 1 \leq \lambda_H(1) p_i\} | = \sum_{i=1}^{n} \lfloor \lambda_H(1) p_i \rfloor < \sum_{i=1}^{n} q_i \leq H.
\]
On the other hand, the definition of $\lambda_H(\delta)$ states that
\[
    \sum_{i=1}^{n} | \{t\in \NN_0 : t + 1 \leq \lambda_H(1) p_i\} | \geq H,
\]
a contradiction.
We conclude that $\tau$ is well defined; the proof for $\bar{\tau}$ is completely analogous.

The fact that $\tau \leq \bar{\tau}$ follows directly from the definitions of these values and the fact that $\lambda_H(\delta)$ is a (piecewise linear) increasing function.

We finally proceed to show properties \ref{partition-uq} and \ref{partition-lq} in the statement. Since the proofs are, again, analogous one to the other, we only include the proof of property \ref{partition-lq}. 

We first consider the corner case $\delta=1$, i.e., we aim to show that for every $x\in f(p,H;1)$ and $i\in [n]$ it holds that $x_i \geq \lfloor q_i \rfloor$. This particular result is well known (see, e.g., \citet{balinski2010fair}); we prove it here for completeness. Let $x\in f(p,H;1)$ be arbitrary. 
 We already observed that $\lambda_H(1) \geq H/P$, so \Cref{obs:lineSets} implies that, for every $i\in [n]$,
\begin{equation}
    x_i \geq \lambda_H(1)p_i-1 \geq q_i-1.\label{eq:lb-xi}
\end{equation}
If we had equality throughout for some state $i$, we would have that $\lambda_H(1)=H/P$, thus \Cref{obs:lineSets} would imply $x_j\leq \lambda_H(1)p_j = q_j$ for every $j\in [n]$, which together with the fact that $x_i=q_i-1$ yields $\sum_{i\in[n]} x_i< H$, a contradiction. We conclude that, for each state $i\in [n]$, at least one of the inequalities in \eqref{eq:lb-xi} is strict, i.e., $x_i\geq \lfloor q_i \rfloor$.

We finally consider the case with $\tau<1$ and we let $\delta\in [\tau,1)$ and $x\in f(p,H;\delta)$ be arbitrary. We claim that, for every $i\in [n]$,
\[
    x_i \geq \lambda_H(\delta)p_i-\delta > \lambda_H(\delta)p_i-1 \geq \lambda_H(\tau)p_i-1 \geq \lfloor q_i \rfloor -1.
\]
Indeed, the first inequality follows from \Cref{obs:lineSets}, the second one from the fact that $\delta<1$, the third one from the fact that $\lambda_H$ is an increasing function, and the last one from the definition of $\tau$. Since $x\in \NN^n_0$, we conclude that $x_i \geq \lfloor q_i \rfloor$ for every $i\in [n]$.

\subsection{Proof of \Cref{prop:BreakingPointsAndTies}} \label{app:breakingTies}

\tiesAndBreaking*
\begin{proof}
We start by proving the first statement. To this end, let $\tau \in [0,1]$ be a breaking point. In the following we will show that (a) $\lambda_H(\delta)$ has a vertex at $\tau$ and (b) $|f(p,H;\tau)| >1$. By the definition of a breaking point, we know that (i) there exists $\varepsilon >0$ such that $f(p,H;\delta)$ is equal for any $\delta \in [\tau-\varepsilon,\tau)$ but is different for $\delta = \tau$, or (ii) there exists $\varepsilon >0$ such that the $f(p,H;\delta)$ is equal for $\delta \in (\tau,\tau +\varepsilon]$ but different for $\delta =\tau$. We assume (i) without loss of generality, as the proof is analogous for (ii). 
\begin{itemize}
    \item To show (a), we assume for contradiction that $f(p,H;\tau - \varepsilon) \setminus f(p,H;\tau)$ is non-empty and let $x$ be a vector in this set. By \Cref{obs:lineSets}, this implies that $\ell_{i,x_{i}-1}(\tau - \varepsilon') \leq \lambda_{H}(\tau - \varepsilon')$ for all $i \in [n]$ and $0<\varepsilon'<\varepsilon$. Continuity of the (piecewise) linear functions yields $\ell_{i,x_{i}-1}(\tau) \leq \lambda(\tau)$ for all $i \in [n]$, hence $x \in f(p,H;\tau)$, a contradiction. Therefore, $f(p,H;\tau) \setminus f(p,H;\tau - \varepsilon)$ is non-empty. Let $x'$ be an element of that set and $x \in f(p,H;\tau - \varepsilon)$ be arbitrary. Since there exists some $i \in [n]$ for which $x'_i > x_i$ and $x'$ is not a feasible outcome for $\tau - \varepsilon$, we know by \Cref{obs:lineSets} that $\ell_{i,x'_i - 1}(\tau - \varepsilon') > \lambda_H(\tau - \varepsilon')$ for all $0 < \varepsilon'<\varepsilon$ and $\ell_{i,x'_i - 1}(\tau) \leq \lambda_H(\tau)$. Thus, $\ell_{i,x'_i -1 }(\delta)$ intersects with $\lambda_H(\delta)$ at $\tau$ and therefore $\lambda_H$ has a vertex at $\tau$. 
    \item To show (b), let $x \in f(p,H;\tau-\varepsilon)$ and $x' \in f(p,H;\tau)$ such that $x \neq x'$. Then, there exists $i\in [n]$ such that $x_i > x'_i$. By \Cref{obs:lineSets}, it holds that $\ell_{i,x'_i}(\tau-\varepsilon') \leq \lambda_{H}(\tau-\varepsilon')$ for any $\varepsilon' < \varepsilon$. Thus, by the continuity of $\lambda_H$ and $\ell_{i,x'_i}$ this also implies $\ell_{i,x'_i}(\tau) \leq \lambda_H(\tau)$. Moreover, since $\mathcal{L}_{\leq H}(\tau)$ already contains at least $H$ lines different from $\ell_{i,x'_i}$ (namely, the lines $\ell_{j,0}, \dots, \ell_{j,x'_j-1}$ for each $j \in [n]$), this implies $|\mathcal{L}_{\leq H}(\tau)| > H$, which directly yields $|f(p,H;\tau)| >1$ by \Cref{obs:lineSets}.  
\end{itemize}

We now turn to prove the second statement, i.e., under the assumption that $p_i \neq p_j$ for all $i,j \in [n]$, the fact that $|f(p,H;\tau)| > 1$ for some $\tau \in [0,1]$ implies that $\tau$ is a breaking point. Note that $|f(p,H;\tau)|>1$ yields $|\mathcal{L}_{\leq H}(\tau)|>H$ by \Cref{obs:lineSets}, which in particular implies that $|\mathcal{L}_{H}(\tau)|>1$. Moreover, since no two lines from $\mathcal{L}$ corresponding to the same state intersect within the interval $[0,1]$ and $p_i \neq p_j$ for all $i,j \in [n]$, all lines in $\mathcal{L}_H(\tau)$ have different slopes. Fix $i \in [n]$ to be the state with the smallest population that has a line in $\mathcal{L}_H(\tau)$. This corresponds to the line with the steepest slope in $\mathcal{L}_H(\tau)$. Let $t \in \mathbb{N}_0$ be such that $\ell_{i,t} \in \mathcal{L}_H(\tau)$. Then, there exists $\varepsilon > 0$ such that for all $\varepsilon > \varepsilon' >0$ it holds that $\ell_{i,t} \in \mathcal{L}_{< H}(\tau-\varepsilon')$ and $\ell_{i,t} \not \in \mathcal{L}_{\leq H}(\tau + \varepsilon')$. Thus, any apportionment vector in $f(p,H,\tau - \varepsilon')$ gives $i$ at least $t+1$ seats but any apportionment vector in $f(p,H,\tau+\varepsilon')$ gives $i$ at most $t$ seats. Thus, $\tau$ is a breaking point.  
\end{proof}

\subsection{Proof of \Cref{obs:number-solutions}}
\label{app:obs-number-solutions}

\obsnumbersolutions* 

Let $n\in \NN$ be arbitrary and consider $\delta=0.5$, $p\in \NN^n$ defined as $p_i=2i-1$ for each $i\in [n]$, and $H= \lfloor n^2/2 \rfloor$. We claim that, for every $S\subseteq [n]$ with $|S|=\lfloor n/2 \rfloor$, $x\in \NN^n_0$ defined as $x_i=i-1+\chi(i\in S)$ satisfies $x\in f(p,H;\delta)$.\footnote{As before, $\chi$ denotes the indicator function, meaning that $\chi(u)=1$ if the condition $u$ holds and zero otherwise.} Indeed, fix such $S$ and $x$ arbitrarily. We have that this apportionment vector respects the house size as
\[
    \sum_{i=1}^{n} x_i = \sum_{i=1}^{n} i -n+|S| = \frac{n(n-1)}{2} +\left\lfloor \frac{n}{2}\right\rfloor = H.
\]
Taking a multiplier $\lambda = 1/2$, we have that $\lambda p_i=i-1/2$ for each $i\in [n]$. Thus, for $i\in S$ we have that
\[
    i-\frac{1}{2} = x_i-1+\delta \leq \lambda p_i \leq x_i+\delta = i+\frac{1}{2}.
\]
Similarly, for $i\in [n]\setminus S$ we have that
\[
    i-\frac{3}{2} = x_i-1+\delta \leq \lambda p_i \leq x_i+\delta = i-\frac{1}{2}.
\]
We obtain that $x\in f(p,H;\delta)$. Since $S$ can be any subset of $\lfloor n/2 \rfloor$ elements of $[n]$, we conclude the result.
\qed

\subsection{Proof of \Cref{lem:wlog-form-of-arrangement}}
\label{app:lem-wlog-form-of-arrangement}

\lemwlogFormArrangement*
\begin{proof}
    Let $\mathcal{L} = \{\ell_1, \dots, \ell_n\}$ be a line arrangement of $n \in \mathbb{N}$ non-vertical lines on the interval $[0,1]$. We also write $\ell_i(\delta) = m_i \delta + c_i$ for all $i \in [n]$. 

    We start by showing that we can apply a linear function to the slope of $\ell_i, i \in [n]$ and do not change the complexity of the $k$-level for each $k\in [n]$. That is, consider the line arrangement $\mathcal{L}'= \{\ell'_1, \dots, \ell'_n\}$, where, for every $i\in [n]$, $\ell'_i(\delta) = (\alpha m_i + \beta)\delta + c_i$ for some $\alpha, \beta \in \mathbb{R}$ with $\alpha\neq 0$. We claim that, for any distinct $i,j,k\in [n]$, the intersection point of lines $\ell_i$ and $\ell_j$ is above the line $\ell_k$ if and only the intersection point of $\ell'_i$ and $\ell'_j$ is above $\ell'_k$. 
    Indeed, the intersection point of $\ell_i$ and $\ell_j$ is given by $$\Big(\delta_0 = \frac{c_j-c_i}{m_i-m_j}, m_i \frac{c_j-c_i}{m_i-m_j} + c_i\Big).$$
    Moreover, $$\ell_k(\delta_0) = m_k\frac{c_j-c_i}{m_i-m_j} + c_k.$$ Thus, the intersection point of $\ell_i$ and $\ell_j$ is above $\ell_k$ if and only if $$(m_i-m_k) \frac{c_j-c_i}{m_i-m_j} > c_k-c_i.$$
    The intersection point of $\ell'_i$ and $\ell'_j$, on the other hand, is given by $$\Big( \delta'_0 = \frac{1}{\alpha}\frac{c_j-c_i}{m_i-m_j}, m_i \frac{c_j-c_i}{m_i-m_j} + \frac{\beta}{\alpha}\frac{c_j-c_i}{m_i-m_j} + c_i \Big).$$
    Moreover, $$\ell'_k(\delta'_0) = m_k\frac{c_j-c_i}{m_i-m_j} + \frac{\beta}{\alpha}\frac{c_j-c_i}{m_i-m_j}+ c_k.$$ Thus, the intersection point of $\ell'_i$ and $\ell'_j$ is above $\ell'_k$ if and only if $$(m_i-m_k) \frac{c_j-c_i}{m_i-m_j} > c_k-c_i,$$ which proves the claim. 

    We can now achieve the property \ref{en:con1} by first applying a transformation such that all slopes are 
    strictly positive, e.g., $\alpha = 1$ and $\beta = -\min\{m_i \mid i \in [n]\} + 0.01$, and then applying a transformation that compresses all slopes into the interval $(1,2)$, e.g., $m_i \mapsto 0.99\cdot \frac{m_i}{\max\{m_i \mid i \in [n]\}} + 1$.
    
    To obtain property \ref{en:con2}, it is easy to show via an analogous argument to the one above that we can apply any linear transformation with a positive factor to $c_i$, for each $i \in [n]$, and do not change the $k$-level. More precisely, we replace the lines obtained via the aforementioned transformation of the slopes, $\ell_i(\delta) = m_i \delta + c_i$ for each $i\in [n]$, by $\ell'_i(\delta) = m_i \delta + \alpha (c_i + \beta)$, where $\alpha > 0$. We first ensure that the coefficients $c_i'= c_i+\beta$ are strictly positive by choosing, e.g., $\beta = - \min \{c_i \mid i \in [n]\} + 0.01$. Then, we choose the smallest $\alpha \in \mathbb{Q}$ such that $\alpha \frac{c'_i}{m_i} \in \mathbb{N}$, for all $i \in [n]$. The resulting line arrangement satisfies properties \ref{en:con1}-\ref{en:con3}.
\end{proof}

\subsection{Power-mean Divisor Methods}
\label{app:power-mean}

In this section, we provide more details about the extension of the upper bound on the number of breaking points of an instance to the case of power-mean divisor methods. For $q\in \RR$, we define $s_q \colon \NN \to \RR_+$ by
\[
    s_q(t) = \begin{cases}
        \lim_{q\to 0} \big(\frac{t^q}{2}+\frac{(t+1)^q}{2}\big)^{1/q} = \sqrt{t(t+1)} & \text{if } q=0,\\
        0 & \text{if } t=0,~ q<0,\\
        \big(\frac{t^q}{2}+\frac{(t+1)^q}{2}\big)^{1/q} & \text{otherwise.} 
    \end{cases}
\]
Note that for every $q\in \RR$ we have $s_q(t)\in [t,t+1]$ for each $t\in \NN$, and that the sequence $s_q(0),s_q(1),s_q(2),\ldots$ is strictly increasing. The rounding rule parameterized by $q\in \RR$ is then defined as
\[
    \llbracket r \rrbracket_q = \begin{cases}
        \{0 \} & \text{ if } r < s_q(0),\\
        \{t \} & \text{ if } s_q(t-1) < r < s_q(t) \text{ for some } t\in \NN,\\
        \{t,t+1\} & \text{ if } r = s_q(t) \text{ for some } t\in \NN_0.
    \end{cases}
\]
For $q\in \RR$, the \textit{$q$-power-mean divisor method} is a family of functions $f(\cdot,\cdot;q)$ (one function for each number $n \in \mathbb{N}$)
such that for every $p\in \NN^n$ and $H\in \NN$
\[
    f(p,H;q) = \bigg\{ x\in \NN^n_0 ~\bigg|~ \text{there exists } \lambda>0 \text{ s.t.\ } x_i \in \llbracket \lambda p_i\rrbracket_q \text{ for every }i\in [n] \text{ and } \sum_{i=1}^{n} x_i = H \bigg\}.
\]
It is not hard to see that $q=-\infty$ yields the Adams method, $q=-1$ the Dean method, $q=0$ the Huntington-Hill method, $q=1$ the Webster/Sainte-Lagu\"e method, and $q=\infty$ the Jefferson/D'Hondt method \citep{marshall2002majorization}. Similarly to the stationary case,\footnote{We have replaced the inductive definition given for stationary divisor methods with an alternative one since the domain of $q$ is not compact. This new definition could be equivalently employed for stationary divisor methods; we used the inductive one for simplicity.} the \textit{breaking points} of an instance $(p,H)$ are all values $\tau\in \RR$ for which there exists $\varepsilon>0$ such that, for all $\varepsilon' \in (0,\varepsilon]$, we have that
\[
    f(p,H;\tau-\varepsilon) = f(p,H;\tau-\varepsilon') \neq f(p,H;\tau+\varepsilon') = f(p,H;\tau+\varepsilon).
\]

For $i \in [n], t \in \{0,\dots,H-1\}$, we consider the functions 
$\ell^+_{i,t}\colon \RR_{++} \to \RR_+$ and $\ell^-_{i,t}\colon \RR\setminus\RR_{+} \to \RR_+$, both equal to $\frac{s_q(t)}{p_i}$ when evaluated at~$q$ but with different domains.
We also consider the families
\[
    \mathcal{L}^+(p,H) = \big\{\ell^+_{i,t}(q) \;\vert\; i \in [n], t \in \{0,\dots,H-1\}\big\},\ \mathcal{L}^-(p,H) = \big\{\ell^-_{i,t}(q) \;\vert\; i \in [n], t \in \{1,\dots,H-1\}\big\}.
\]
The following lemma states the key observation for our extension to power-mean divisor methods.

\begin{lemma}\label{lem:pseudolines}
    Let $(p,H)$ be an instance with $p_i\neq p_j$ for all $i,j\in [n]$ with $i\neq j$. Then, $\mathcal{L}^+(p,H)$ and $\mathcal{L}^-(p,H)$ are pseudoline arragements.
\end{lemma}

\begin{proof}
    We need to prove that, for any pair of curves $\ell$ and $\ell'$ in either $\mathcal{L}^+(p,H)$ or $\mathcal{L}^-(p,H)$, the curves intersect at most once. We will make use of the following claim.

    \begin{claim}\label{claim:exp-eq-1sol}
        For every $a,b,c,d\in \RR_{++}$ with $a<b$, $c<d$, and $(a,b)\neq (c,d)$, the equation in $q\in \RR_{++}$
        \[
            a^q+b^q = c^q+d^q
        \]
        has at most one solution.
    \end{claim}

    \begin{proof}
        Let $a,b,c,d$ be as in the statement. Suppose toward a contradiction that there are two solutions $q,r$ with $0<q<r$, i.e., 
        \begin{align}
            a^q+b^q & = z = c^q+d^q,\label{eq:power-mean-q}\\
            a^r+b^r & = c^r+d^r\label{eq:power-mean-r}
        \end{align}
        for some $z\in \RR$. If $a=c$, any of these equations yields $b=d$, hence $(a,b) = (c,d)$, a contradiction. In the following, we thus assume $a\neq c$.
        From equation \eqref{eq:power-mean-q}, we know that
        \begin{equation}
            a^r+b^r = (a^q)^{r/q} + (z-a^q)^{r/q}, \qquad
            c^r+d^r = (c^q)^{r/q} + (z-c^q)^{r/q}.\label{eq:power-mean-abcd}
        \end{equation}
        We further know that
        \begin{equation}
            a^q < \frac{z}{2},\qquad c^q < \frac{z}{2}\label{eq:power-mean-bounds-ac}
        \end{equation}
        due to $a<b$ and $c<d$.

        We now observe that the function $g\colon (0,z/2) \to \RR$ defined as $g(w)=w^{r/q} + (z-w)^{r/q}$ is strictly decreasing. Indeed, for $w\in (0,z/2)$ we have
        \[
            g'(w) = \frac{r}{q}\big(w^{r/q-1} + (z-w)^{r/q-1}\big) < 0,
        \]
        where the last inequality uses that $w<z/2$ and $q<r$. Therefore, inequalities \eqref{eq:power-mean-bounds-ac} and $a\neq c$ yield 
        \[
            (a^q)^{r/q} + (z-a^q)^{r/q} \neq (c^q)^{r/q} + (z-c^q)^{r/q}.
        \]
        Equalities \eqref{eq:power-mean-abcd} then imply $a^r+b^r \neq c^r+d^r$, a contradiction to equation \eqref{eq:power-mean-r}.
    \end{proof}
    
    To show that $\calL^+(p,H)$ is a pseudoline arrangement, we need to prove that, for every $i,j\in [n]$ and $t,u\in \{0,\dots,H-1\}$ with $(i,t) \neq (j,u)$, the equation in $q>0$
    \[
        \frac{1}{p_i}\bigg(\frac{t^q}{2}+\frac{(t+1)^q}{2}\bigg)^{\frac{1}{q}} = \frac{1}{p_j}\bigg(\frac{u^q}{2}+\frac{(u+1)^q}{2}\bigg)^{\frac{1}{q}},
    \]
    has at most one solution. For this domain of~$q$, the equation is equivalent to
    \[
        \bigg(\frac{t}{p_i}\bigg)^q+\bigg(\frac{t+1}{p_i}\bigg)^q = \bigg(\frac{u}{p_j}\bigg)^q+\bigg(\frac{u+1}{p_j}\bigg)^q.
    \]
    Furthermore, $\big(\frac{t}{p_i},\frac{t+1}{p_i}\big) \neq \big(\frac{u}{p_j},\frac{u+1}{p_j}\big)$: If we had $p_jt=p_iu$ and $p_j(t+1)=p_i(u+1)$ we would have $(p_i,t) = (p_j,u)$ and thus $(i,t) = (j,u)$, a contradiction. Thus, \Cref{claim:exp-eq-1sol} directly implies that the equation has at most one solution. 
    
    Similarly, to show that $\calL^-(p,H)$ is a pseudoline arrangement, we need to prove that, for every $i,j\in [n]$ and $t,u\in \{1,\dots,H-1\}$ with $(i,t) \neq (j,u)$, the equation in $q<0$
    \[
        \frac{1}{p_i}\bigg(\frac{t^q}{2}+\frac{(t+1)^q}{2}\bigg)^{\frac{1}{q}} = \frac{1}{p_j}\bigg(\frac{u^q}{2}+\frac{(u+1)^q}{2}\bigg)^{\frac{1}{q}},
    \]
    has at most one solution. For this domain of~$q$, the equation is equivalent to
    \[
        \bigg(\frac{p_i}{t+1}\bigg)^{-q} + \bigg(\frac{p_i}{t}\bigg)^{-q} = \bigg(\frac{p_j}{u+1}\bigg)^{-q} + \bigg(\frac{p_j}{u}\bigg)^{-q}.
    \]
    The fact that $\big(\frac{p_i}{t+1},\frac{p_i}{t}\big) \neq \big(\frac{p_j}{u+1},\frac{p_j}{u}\big)$ follows from the previous case. Thus, \Cref{claim:exp-eq-1sol} again implies that the equation has at most one solution.
\end{proof}

We define, similarly to the stationary case,
\begin{align*}
\lambda^+_H(q) & = \min \big\{ \lambda\in \RR \;\big\vert\; |\{\ell \in \mathcal{L}^+(p,H) \mid \ell(q) \leq \lambda\}| \geq H \big\} \text{ for } q\in \RR_{++},\\ 
\lambda^-_H(q) & = \min \big\{ \lambda\in \RR \;\big\vert\; |\{\ell \in \mathcal{L}^-(p,H) \mid \ell(q) \leq \lambda\}| \geq H \big\} \text{ for } q\in \RR\setminus \RR_{+}.
\end{align*}
As before, for all $q\in \RR_{++}$, $\lambda^+_H(q)$ is equal to $\ell(q)$ for some $\ell\in \calL^+(p,H)$, and for all $q\in \RR\setminus \RR_{+}$, $\lambda^-_H(q)$ is equal to $\ell(q)$ for some $\ell\in \calL^-(p,H)$. Thus, the number of breaking points of an instance is at most the number of vertices of the piecewise linear function $\lambda^+_H$ plus the number of vertices of the piecewise linear function $\lambda^-_H$ plus one (due to a potential breaking point at $q=0$). We make use of the following property analogous to that used in the proof of the upper bound of \Cref{thm:main}: For every instance $(p,H)$, there exist at most $2n-1$ lines in $\calL^+(p,H)$ that intersect with $\lambda^+_H$ and at most $2n-1$ lines in $\calL^-(p,H)$ that intersect with $\lambda^-_H$. This property can be shown in a completely analogous way to the case of line arrangements detailed in \Cref{subsec:ub-bp}, as its proof only uses that, for every state $i$ and integer $t$, the function $\ell_{i,t}$ is increasing and its largest value is at most the smallest value of $\ell_{i,t+1}$, and both of these properties hold for the functions in $\calL^+(p,H)$ and for the functions in $\calL^-(p,H)$.

The last ingredient we need is a bound on the complexity of the $k$-level of a pseudoline arrangement. \citet{tamaki2003characterization} extended the result by \citet{dey1998improved} and showed that, for an arrangement of $m$ pseudolines, the complexity of the $k$-level is bounded by $\calO(m^{4/3})$ for any $k\in \{0,\ldots,m\}$. Combining the previous two observations, we can conclude that the number of breaking points for both $q<0$ and $q>0$ is upper bounded by $\calO(n^{4/3})$, thus the total number of breaking points is upper bounded by $\calO(n^{4/3})$ as well. Note that this holds even when the populations of different states coincide, as having coincident curves cannot increase the complexity of the arrangement.\footnote{This is the same reason why having coincident lines given by the constant functions $\ell_{i,0}(q)=0$ for all $i\in [n]$ when $q<0$ is not an issue for the complexity either.}

\section{Proofs Deferred from Section \ref{sec:randomization-variable}}

\subsection{Proof of \Cref{prop:det-dev-quota}}\label{app:prop-det-dev-quota}

\propdetdevquota*

Fix $H\in \NN$ and $\varepsilon>0$. To show part \ref{prop:det-dev-quota-i}, we consider $n=H$, $M = \left\lceil \frac{1}{\varepsilon}(H-1)(H-\varepsilon)\right\rceil$, and define $p\in \NN^n$ as $p_1=M$ and $p_i=1$ for $i\in \{2,\ldots,H\}$. 
We claim that $x_1\leq 1$ for every $x\in f(p,H;0)$. Indeed, if $x_1\geq 2$ for some $x\in f(p,H;0)$, then for every $\lambda \in \Lambda(x;\delta)$ we would have $\lambda p_1 \geq 1$, i.e., $\lambda\geq 1/M$. This would imply that, for every $i \in \{2,\ldots,H\}$, $\lambda p_i >0$ and thus $x_i\geq 1$. But this yields $\sum_{i=1}^n x_i > H$, a contradiction. 
We conclude that, for every $x \in f(p,H;0)$,
\[
    | x_1 - q_1 | \geq \frac{M}{M+H-1}H - 1 \geq \frac{\frac{1}{\varepsilon}(H-1)(H-\varepsilon)}{\frac{1}{\varepsilon}(H-1)(H-\varepsilon)+H-1}H - 1 = H-1-\varepsilon.
\]

To see part \ref{prop:det-dev-quota-ii}, we let $\delta\in (0,1]$ be arbitrary, we fix $M=\left\lceil \frac{H}{\varepsilon}-1 \right\rceil $, and we let $n\in \NN$ with $n>M+1$ be such that
\[
    \frac{H-1}{n-1-M}M< \delta.
\]
Observe that this implies 
\begin{equation}
    H-1+\delta < \frac{n-1}{M}\delta.\label{eq:non-emtpy-interval-lambda-det}
\end{equation}
Consider now $p=\left(n-1,M,\ldots,M\right)$. We claim that $x_1=H$ for every $x \in f(p,H;\delta)$. Indeed, if $x_1\leq H-1$ for some $x\in f(p,H;\delta)$, then for every $\lambda \in \Lambda(x;\delta)$ we would have $\lambda p_1 \leq H-1+\delta$, i.e., $\lambda< \frac{\delta}{M}$. This would imply that, for every $i \in \{2,\ldots,H\}$, $\lambda p_i < \delta$ and thus $x_i=0$. But this yields $\sum_{i=1}^n x_i < H$, a contradiction. 
We conclude that, for every $x \in f(p,H;\delta)$,
\[
    | x_1 - q_1 | = H-\frac{1}{M+1}H \geq H - \frac{1}{\frac{H}{\varepsilon}-1+1}H = H-\varepsilon.\eqno\qed
\]

\subsection{Proof of \Cref{prop:rand-dev-quota}}\label{app:prop-rand-dev-quota}

\propranddevquota* 

    Let $G,B,H$, and $\varepsilon$ be as in the statement, and let $F^B$ be the $G$-randomized divisor method with tie-breaking distribution $B$. If $G$ is a distribution with all probability mass in one point, then the result follows directly from \Cref{prop:det-dev-quota}, so we assume that this is not the case in the following. 
    We prove the proposition by making use of the following two claims, which provide large expected deviations for cases when there is a considerable probability mass below or above a value $\xi\in [0,1]$, respectively. We use similar instances as in the proof of \Cref{prop:det-dev-quota}, but additional conditions are required in order to ensure these large expected deviations.

    \begin{claim}\label{claim:bad-ins-F0-large}
        For every $\xi\in (0,1)$ such that 
        $G(\xi) > \frac{H}{\left\lceil 1/\xi-1 \right\rceil +H-1}$, there exist $n\in \NN$ and $p\in \NN^n$ such that
        \[
        | \EX(F^B_1(p,H)) - q_1 | \geq 
        G(\xi)(H-1) - \frac{H-1}{\left\lceil 1/\xi-1 \right\rceil +H-1}H.
        \]
    \end{claim}

    \begin{proof}

        We consider $n = H$ and define $p\in \NN^n$ with 
        $p_1=\left\lceil 1/\xi-1 \right\rceil$  and $p_i=1$ for $i\in \{2,\ldots,H\}$. 
        We claim that for every $\delta'\in [0,\xi]$ it holds that $x_1\leq 1$ for every $x\in f(p,H;\delta')$. Indeed, fix $\delta'\in [0,\xi]$ arbitrarily and suppose towards a contradiction that $x_1\geq 2$ for some $x\in f(p,H;\delta')$. Then, for every $\lambda \in \Lambda(x;\delta')$ we have that $\lambda p_1 \geq 1+\delta'\geq 1$, i.e., $\lambda\geq 1/p_1$. This would imply that, for every $i \in \{2,\ldots,H\}$, $\lambda p_i \geq 1/p_1 > \xi \geq \delta'$ and thus $x_i\geq 1$. But this yields $\sum_{i=1}^n x_i > H$, a contradiction. 
        Since $G(\xi) > \frac{H}{\left\lceil 1/\xi-1 \right\rceil +H-1}$, it holds that 
        \[
            G(\xi) + (1-G(\xi))H < \frac{\left\lceil 1/\xi-1 \right\rceil}{\left\lceil 1/\xi-1 \right\rceil+H-1}H = q_1,
        \]
        and thus
        \begin{align*}
            | \EX(F^B_1(p,H)) - q_1 | & \geq \frac{\left\lceil 1/\xi-1 \right\rceil}{\left\lceil 1/\xi-1 \right\rceil+H-1}H - G(\xi) - (1-G(\xi))H \\
            & \geq G(\xi)(H-1) - \frac{H-1}{\left\lceil 1/\xi-1 \right\rceil +H-1}H.
            \tag*{\raisebox{-.5\baselineskip}{\qedhere}}
        \end{align*}
    \end{proof}

    \begin{claim}\label{claim:bad-ins-F0-small}
        For every $\xi\in (0,1)$ such that $G(\xi)\in (0,1)$ and every $M>\frac{G(\xi)}{1-G(\xi)}$, there exist $n\in \NN$ and $p\in \NN^n$ such that
        \[
        | \EX(F^B_1(p,H)) - q_1 | \geq \left(1-G(\xi) -\frac{1}{M+1}\right)H.
        \]
    \end{claim}

    \begin{proof}
        Let $\xi$ and $M$ be as in the statement and let $n\in \NN$ be such that $n>M$ and $HM < (n-1)\xi$. 
        We define $p\in \NN^n$ as $p_1=n-1$  and $p_i=M$ for $i\in \{2,\ldots,n\}$. We claim that for every $\delta'\in [\xi,1]$ it holds that $x_1=H$ for every $x\in f(p,H;\delta')$. Indeed, fix $\delta'\in [\xi,1]$ arbitrarily and suppose towards a contradiction that $x_1\leq H-1$ for some $x\in f(p,H;\delta')$. Then, for every $\lambda \in \Lambda(x;\delta')$ we have that $\lambda p_1 \leq H-1+\delta'\leq H$, i.e., $\lambda\leq H/p_1=H/(n-1)$. This would imply that, for every $i \in \{2,\ldots,H\}$, $\lambda p_i \leq HM/(n-1) < \xi \leq \delta'$ and thus $x_i=0$. But this yields $\sum_{i=1}^n x_i < H$, a contradiction. 
        Since $M>\frac{G(\xi)}{1-G(\xi)}$, it holds that \[(1-G(\xi))H + G(\xi)\cdot 0 > \frac{1}{M+1}H = q_1,\] and thus
        \[
            | \EX(F^B_1(p,H)) - q_1 | \geq (1-G(\xi))H + G(\xi)\cdot 0 -\frac{1}{M+1}H = 
            \left(1-G(\xi) -\frac{1}{M+1}\right)H.\tag*{\raisebox{-.5\baselineskip}{\qedhere}}
        \]
    \end{proof}

    We now finish the proof by combining these two results.
    Let $\xi \in (0,1)$ be arbitrary; we will take it arbitrarily close to $0$ to conclude later.
    If 
    $G(\xi) \leq \frac{H}{\left\lceil 1/\xi-1 \right\rceil +H-1}$, then 
    we obtain from \Cref{claim:bad-ins-F0-small} that
    \begin{equation}
        \max_{n\in \NN, p\in \NN^n} | \EX(F^B_1(p,H)) - q_1 | \geq \left(1-\frac{H}{\left\lceil 1/\xi-1 \right\rceil +H-1} -\frac{1}{M+1}\right)H.\label{eq:lb-dev-G-small}
    \end{equation}
    for every $M>\frac{G(\xi)}{1-G(\xi)}$. 
    Otherwise, we  obtain from \Cref{claim:bad-ins-F0-large} and \Cref{claim:bad-ins-F0-small} that
    \[
        \max_{\substack{n\in \NN,\\ p\in \NN^n}} | \EX(F^B_1(p,H)) - q_1 | \geq \max\left\{ G(\xi)(H-1) - \frac{H-1}{\left\lceil 1/\xi-1 \right\rceil +H-1}H,~ \left(1-G(\xi) -\frac{1}{M+1}\right)H \right\}
    \]
    for every $M>\frac{G(\xi)}{1-G(\xi)}$.
    Since the first expression in the maximum is increasing in $G(\xi)$ and the second is decreasing in this value, a lower bound for the maximum is obtained by taking $G(\xi)$ such that both are equal. Denoting this value of $G(\xi)$ as $\rho$, such that
    \begin{align*}
         & \rho(H-1) - \frac{H-1}{\left\lceil 1/\xi-1 \right\rceil +H-1}H = \left(1-\rho -\frac{1}{M+1}\right)H\\
         \Longleftrightarrow \ & \rho = \left( 1 - \frac{1}{M+1} + \frac{H-1}{\left\lceil 1/\xi-1 \right\rceil +H-1} \right)\frac{H}{2H-1},
    \end{align*}
    we obtain that
    \begin{equation}
        \max_{n\in \NN, p\in \NN^n} | \EX(F^B_1(p,H)) - q_1 | \geq \left(1-\rho-\frac{1}{M+1}\right)H\label{eq:lb-dev-G-large}
    \end{equation}
    for every $M>\frac{G(\xi)}{1-G(\xi)}$.
    Putting expressions \eqref{eq:lb-dev-G-small} and \eqref{eq:lb-dev-G-large} together, we obtain that
    \begin{align*}
        \max_{\substack{n\in \NN,\\ p\in \NN^n}} | \EX(F^B_1(p,H)) - q_1 | \geq \min \Bigg\{ & \left(1-\frac{H}{\left\lceil 1/\xi-1 \right\rceil +H-1} -\frac{1}{M+1}\right)H,\\
        &\left(1-\left( 1 - \frac{1}{M+1} + \frac{H-1}{\left\lceil 1/\xi-1 \right\rceil +H-1} \right)\frac{H}{2H-1}-\frac{1}{M+1}\right)H \Bigg\}
    \end{align*}
    for every $\xi\in (0,1)$ and $M>\frac{G(\xi)}{1-G(\xi)}$. Taking $\xi$ arbitrarily close to $0$ and $M$ arbitrarily large, the first term in the maximum tends to $H$, whereas the second term in the maximum tends to \[ \left(1-\frac{H}{2H-1} \right)H = \frac{H-1}{2H-1}H = \left(1-\frac{1}{2H-1}\right) \frac{H}{2}. \] That is, for any $\varepsilon>0$ we can reach a deviation of $(1-1/(2H-1))H/2 - \varepsilon$, as claimed.

\subsection{Proof of \Cref{prop:jeff-adams}}
\label{app:prop-jeff-adams}

\propjeffadams*

Let $\delta, H, n,p$, and $i$ be as in the statement.
We know from \Cref{prop:partition} that $0 \leq x_i \leq \lceil q_i\rceil$ for every $x\in f(p,H;0)$ and that $\lfloor q_i\rfloor \leq x_i \leq H$ for every $x\in f(p,H;1)$.
Therefore, for every $i\in [n]$
\begin{align*}
    | \EX(F^B(p,H)) - q_i | & \leq \max_{x\in f(p,H;0),y\in f(p,H;1)} \Big|\frac{1}{2}x_i+\frac{1}{2}y_i - q_i \Big| \\
    & \leq \max \left\{ \frac{1}{2}q_i+\frac{1}{2}(q_i-\lfloor q_i\rfloor), \frac{1}{2}(\lceil q_i\rceil-q_i) + \frac{1}{2}(H-q_i) \right\}\\
    & = \max \left\{ \frac{1}{2}(2q_i-\lfloor q_i\rfloor), \frac{1}{2}(H+\lceil q_i\rceil-2q_i) \right\} < \frac{1}{2}(H+1),
\end{align*}
where the maximum in the first step accounts for both possible worst-cases, either $x_i=0$ and $y_i=\lfloor q_i\rfloor$ for $x\in f(p,H;0)$ and $y\in f(p,H;1)$, or $x_i=\lceil q_i \rceil$ and $y_i=H$ for $x\in f(p,H;0)$ and $y\in f(p,H;1)$, and we used in the last step that $q_i\leq H$.
\qed

\subsection{Proof of \Cref{prop:dev-quota-fixed-pop}}\label{app:prop-dev-quota-fixed-pop}

\propdevquotafixedpop*

    Let $n,~p$, and $H$ be as in the statement. We let $\Lambda: [0,1] \rightarrow 2^\mathbb{R}$ be a function that maps every $\delta \in [0,1]$ to the set of multipliers producing an apportionment via the $\delta$-divisor method, i.e., \[ \Lambda(\delta) = \bigcup_{x\in f(p,H;\delta)} \Lambda(x;\delta). \] We refer to such values of $\lambda$ as \textit{feasible} for the instance.
    Consider the functions $\lambda_{\min}(\delta)$ and $\lambda_{\max}(\delta)$, pointwise minimal and maximal for any $\delta \in [0,1]$, respectively. Specifically, for $\delta\in [0,1]$ we define $\lambda_{\min}(\delta) = \min \Lambda(\delta)$ and $\lambda_{\max}(\delta) = \max \Lambda(\delta)$. 
    
    Suppose towards a contradiction that there exists $\delta\in [0,1]$ such that for every $i\in [n]$, we have that $p_i\lambda_{\min}(\delta) > q_i +\delta$.
    Fixing $x \in f(p,H;\delta)$ arbitrarily, this would imply that for every $i\in [n]$ it holds that $x_i \geq p_i\lambda_{\min} - \delta > q_i$,
    and thus $\sum_{i\in [n]} x_i > H$, a contradiction. We conclude that 
    \begin{equation}
        \lambda_{\min}(\delta) \leq \frac{H}{P} +\max_{j\in[n]} \frac{\delta}{p_j} = \frac{H}{P} + \frac{\delta}{\min_{j\in[n]} p_j}.\label{eq:ub-delta-min}
    \end{equation}
    Similarly, suppose that there exists $\delta\in [0,1]$ such that for every $i\in [n]$, $p_i\lambda_{\max}(\delta) < q_i - (1-\delta)$.
    Fixing $x \in f(p,H;\delta)$ arbitrarily, this would imply that for every $i\in [n]$ it holds that $x_i \leq p_i\lambda_{\max} + 1-\delta < q_i$,
    and thus $\sum_{i\in [n]} x_i < H$, a contradiction. We conclude that 
    \begin{equation}
        \lambda_{\max}(\delta) \geq \frac{H}{P} -\max_{j\in[n]} \frac{1-\delta}{p_j} = \frac{H}{P} - \frac{1-\delta}{\min_{j\in[n]} p_j}.\label{eq:lb-delta-max}
    \end{equation}
    
    We now fix $G$ and $B$ as in the statement, and $\delta\sim G$. We make use of the following claim.

    \begin{claim} \label{lem:expected-delta}
        For every function $\lambda: [0,1] \to \RR$ such that $\lambda(\delta)\in \Lambda(\delta)$ for each $\delta\in [0,1]$, and for every state $i \in [n]$, the $G$-randomized divisor method $F^B$ satisfies
        \[
        \EX(F^B_i(p,H)) \in [p_i \EX(\lambda(\delta)) - 1/2, p_i \EX(\lambda(\delta)) + 1/2].
        \]
    \end{claim}

    \begin{proof}
        We know that for any $\delta' \in [0,1]$ and every $x\in f(p,H,\delta')$ it holds that \[ x_i + \delta' - 1 \leq \lambda(\delta') p_i \leq x_i + \delta'. \]
        Taking expectation over $\delta\sim G$ and using linearity of expectation, together with the fact that $\EX(\delta)=1/2$, we get that 
        $\EX(F^B_i(p,H)) - 1/2 \leq \EX(\lambda(\delta))p_i \leq \EX(F^B_i(p,H)) + 1/2.$
    \end{proof}

    Taking expectation over $\delta\sim G$ on both sides, \eqref{eq:ub-delta-min} and \eqref{eq:lb-delta-max} imply
    \[
        \EX(\lambda_{\min}(\delta)) \leq \frac{H}{P} + \frac{1}{2\min_{j\in[n]} p_j},\quad \EX(\lambda_{\max}(\delta)) \geq \frac{H}{P} - \frac{1}{2\min_{j\in[n]} p_j}.
    \]
    Therefore, for every function $\lambda:[0,1]\to \mathbb{R}$ such that $\lambda_{\min}(\delta') \leq \lambda(\delta') \leq \lambda_{\max}(\delta')$ for every $\delta'\in [0,1]$, i.e., that is feasible in the aforementioned way, it holds that
    \[
        \frac{H}{P} - \frac{1}{2\min_{j\in[n]} p_j} \leq \EX(\lambda(\delta)) \leq \frac{H}{P} + \frac{1}{2\min_{j\in[n]} p_j}.
    \]
    Then, we conclude from \Cref{lem:expected-delta} that
    \[
        \EX(F^B(p,H)) \in \left[ q_i - \frac{1}{2} \left(\frac{p_i}{\min_{j\in[n]} p_j} +1\right),\ q_i + \frac{1}{2} \left(\frac{p_i}{\min_{j\in[n]} p_j} +1\right) \right].\eqno\qed
    \]

\subsection{Proof of \Cref{thm:variable-size}}
\label{app:thm-variable-size}

\thmvariablesize*

We start with the first claim. Let $\delta_i(k)$ be random variables with marginal distribution $U[0,1]$ and $s_i(k)=k+\delta_i(k)$ for every $i\in[n]$ and $k\in \NN_0$, and let $F$ denote the randomized fixed-divisor method with signpost sequences $s_i(0),s_i(1),s_i(2),\ldots$ for every $i\in [n]$.
    Consider an arbitrary instance given by $p=(p_1,\ldots,p_n)$ and a house size $H$.
    
    In order to prove quota compliance and ex-ante proportionality, we also fix an arbitrary state $i\in [n]$. 
    For every $k\ge 1$ and $r\in [k,k+1]$ we have $\PP(s_i(k)< r) = r-k$.
    If $q_i$ is integer, this implies $\PP(N_{s_i}(q_i) = q_i) = 1$. Otherwise, we obtain
    \begin{align*}
        \PP\left(N_{s_i}\left(q_i\right) = \left\lceil q_i\right\rceil\right) & = \PP\left(s_i\left(\left\lfloor q_i \right\rfloor\right) < q_i \right) = q_i - \left\lfloor q_i \right\rfloor, \\
        \PP\left(N_{s_i}\left(q_i\right) = \left\lfloor q_i\right\rfloor\right) & = \PP\left(s_i\left(\left\lfloor q_i \right\rfloor\right) \geq q_i \right) = \left\lceil q_i \right\rceil - q_i, \\
        \PP\left(N_{s_i}\left(q_i\right) = r \right) & = 0 \text{ for every } r\not\in \left\{ \left\lfloor q_i\right\rfloor, \left\lceil q_i\right\rceil \right\}.
    \end{align*}
    These properties directly imply $\PP(F_i(p,H)\in \{\lfloor q_i \rfloor, \lceil q_i \rceil\}) = 1$, i.e., $F$ is quota-compliant.
    If $q_i$ is integer, since $\PP(F_i(p,H)=q_i)=1$ we immediately obtain $\EE(F_i(p,H)) = q_i$.
    On the other hand, if $q_i$ is fractional, we have that $\EE(F_i(p,H)) = \EE\left(N_{s_i}\left(q_i\right) \right) = \left\lceil q_i\right\rceil \left( q_i - \left\lfloor q_i \right\rfloor \right) + \left\lfloor q_i\right\rfloor \left( \left\lceil q_i \right\rceil - q_i \right) = q_i,$
    so we conclude that $F$ is ex-ante proportional.
    
    We now show that $F$ satisfies population monotonicity.
    Let $p,p'$ be two population vectors, $H,H'$ two house sizes, and let $P' = \sum_{\ell=1}^n p'_{\ell}$ be the total population in the instance $(p',H')$. Let $i,j\in [n]$ be two states such that $p'_i/p'_j \geq p_i/p_j$, or equivalently, $p'_i/p_i \geq p'_j/p_j$.
    Suppose that $\PP(F_i(p,H)> F_i(p',H'))>0$ and $\PP(F_j(p,H)< F_j(p',H'))>0$.
    This implies that $\PP( N_{s_i}(q_i) > N_{s_i}(q'_i)) > 0$ and $\PP( N_{s_j}(q_j) < N_{s_j}(q'_j)) > 0,$
    where $q'_\ell=p'_\ell H'/P'$ for $\ell\in [n]$.
    Since the signpost sequences are the same for both instances, these expressions imply $q_i > q'_i$ and $q_j < q'_j$.
    Putting these inequalities together yields $p'_i/p_i < (H/H')\cdot P'/P < p'_j/p_j,$
    a contradiction.
    We conclude that $\PP(F_i(p,H)> F_i(p',H'))=0$ or $\PP(F_j(p,H)< F_j(p',H'))=0$, i.e., $F$ is population-monotone.

    We now prove the second claim. To do so, we define the random variables
    \[
        \delta_i = \begin{cases}
            U[0,1] & \text{ if } $i$ \text{ is odd,}\\
            1-\delta_{i-1} & \text{ if } $i$ \text{ is even,}
        \end{cases} \quad \text{for every } i\in [n],
    \]
    we define $s_i(k)=k+\delta_i$ for every $i\in [n]$ and $k\in \NN_0$, and we denote as $F$ the randomized fixed-divisor method with signpost sequences $s_i(0), s_i(1), s_i(2),\ldots$ for every $i\in [n]$.

    For each $i\in [n]$, we let $r_i = q_i - \lfloor q_i \rfloor$ denote the fractional part of state $i$'s quota. Since these signpost sequences are a subclass of those studied in the previous paragraph, for every $i\in [n]$ it holds that $\PP(\lfloor q_i \rfloor \leq F_i(p,H) \leq  \lceil q_i \rceil)=1$ and that $F_i(p,H) =  \lceil q_i \rceil$ if and only if $s_i(\lfloor q_i \rfloor) < q_i$, which is equivalent to $\delta_i < r_i$. To see that the deviation from the house size is never greater than $n/2+ n\bmod 2$, we observe that
    \begin{align}
        \Bigg|\sum_{i=1}^{n} F_i(p,H) - H \Bigg| = \Bigg| \sum_{i=1}^{n} \left( F_i(p,H) - q_i \right) \Bigg| = \Bigg| |\{i\in [n]: r_i> \delta_i\}| - \sum_{i=1}^{n} r_i \Bigg|.\label{eq:dev-house-size}
    \end{align}
    On the other hand, from the definition of $\delta_i$ for each $i\in [n]$ we have that
    \begin{align}
        \frac{n-n \bmod 2}{2} \leq \sum_{i=1}^{n} \delta_i \leq \frac{n+n \bmod 2}{2}.\label{eq:sum-delta}
    \end{align}
    Let $\pi:[n]\to[n]$ be a bijection such that $\delta_{\pi(1)}\leq \delta_{\pi(2)}\leq \cdots \leq \delta_{\pi(n)}$. We now distinguish two cases to conclude the bound. If $|\{i\in [n]: r_i> \delta_i\}| \geq (n+n \bmod 2)/2$, we define $k = |\{i\in [n]: r_i> \delta_i\}| - (n+n \bmod 2)/2$ and we observe that
    \[
        \sum_{i=1}^{n} r_i > \sum_{i=1}^{\frac{n+n \bmod 2}{2}+k}\delta_{\pi(i)} = \sum_{i=1}^{n}\delta_{\pi(i)} - \sum_{i=\frac{n+n \bmod 2}{2}+k+1}^{n}\delta_{\pi(i)} \geq k,
    \]
    where the last inequality comes from \eqref{eq:sum-delta} and the fact that $\delta_i\leq 1$ for every $i\in [n]$. Replacing $k$ in the last expression and accounting for both cases $|\{i\in [n]: r_i> \delta_i\}| \geq \sum_{i=1}^{n} r_i$ and $|\{i\in [n]: r_i> \delta_i\}| < \sum_{i=1}^{n} r_i$, we conclude that
    \[
        \Bigg| |\{i\in [n]: r_i> \delta_i\}| - \sum_{i=1}^{n} r_i \Bigg| \leq \max\left\{ \frac{n+n\bmod 2}{2}, n-1-\frac{n+n\bmod 2}{2}\right\} = \frac{n+n\bmod 2}{2}.
    \]
    On the other hand, if $|\{i\in [n]: r_i> \delta_i\}| \leq (n-n\bmod 2)/2$, we define $k = (n-n\bmod 2)/2 - |\{i\in [n]: r_i> \delta_i\}|$ and we observe that
    \begin{align*}
        \sum_{i=1}^{n} r_i & \leq \frac{n-n\bmod 2}{2}-k + \sum_{i=\frac{n-n\bmod 2}{2}-k+1}^{n}\delta_{\pi(i)} = \sum_{i=1}^{n}\delta_{\pi(i)} + \sum_{i=1}^{\frac{n-n\bmod 2}{2}-k}(1-\delta_{\pi(i)}) \\
        & \leq
        \frac{n+n \bmod 2}{2} + \frac{n-n \bmod 2}{2}-k = 
        n- k,
    \end{align*}
    where the last inequality comes from \eqref{eq:sum-delta} and the fact that $\delta_i\geq 0$ for every $i\in [n]$. Replacing $k$ in the last expression and accounting again for both cases $|\{i\in [n]: r_i> \delta_i\}| \geq \sum_{i=1}^{n} r_i$ and $|\{i\in [n]: r_i> \delta_i\}| < \sum_{i=1}^{n} r_i$, we conclude that
    \[
        \Bigg| |\{i\in [n]: r_i> \delta_i\}| - \sum_{i=1}^{n} r_i \Bigg| \leq \max\left\{ \frac{n-n\bmod 2}{2}, \frac{n+n\bmod 2}{2}\right\} = \frac{n+ n\bmod 2}{2}.
    \]
    In either of these cases, we conclude the desired bound from \eqref{eq:dev-house-size}.

    Let now $Z=\sum_{i=1}^{n}F_i(p,H) - \sum_{i=1}^{n} \lfloor q_i \rfloor$ be the random variable equal to the number of seats that are allocated above the sum of the floor of the quotas.
    Then, we have that
    \[
        Z = \big| \left\{ i\in [n]: s_i\left(\left\lfloor q_i \right\rfloor \right) < q_i \right\} \big| = \sum_{i=1}^{n} \chi \left( s_i\left(\left\lfloor q_i \right\rfloor \right) < q_i \right) = \sum_{i=1}^{n} Y_i,
    \]
    where $\chi$ denotes the indicator function, meaning that $\chi(u)=1$ if the condition $u$ holds and zero otherwise; and for every $i\in [n]$, the random variables $Y_i$ distribute as Bernoulli random variables with a success probability equal to $q_i-\lfloor q_i\rfloor$. These variables are independent except for pairs of variables $Y_i, Y_{i+1}$ with $i\leq n-1$ odd. In these cases, these variables are negatively correlated. To see this, let $i\in [n-1]$ be an arbitrary odd number and observe that
    \begin{align*}
        \PP\left[ Y_i=0 \wedge Y_{i+1}=0 \right] & = \PP\left[ \delta_i \geq r_i \wedge 1-\delta_i \geq r_{i+1} \right] = \PP\left[ r_i \leq \delta_i \leq 1-r_{i+1} \right]
        = 1-r_i-r_{i+1} \\
        & \leq (1-r_i)(1-r_{i+1}) = \PP[\delta_i \geq r_i]\PP[1-\delta_i \geq r_{i+1}] = \PP\left[ Y_i=0 \right] \PP\left[Y_{i+1}=0 \right],
    \end{align*}
    if $r_i+r_{i+1}\leq 1$, and the first expression is 0, otherwise.
    Similarly,
    \begin{align*}
        \PP\left[ Y_i=1 \wedge Y_{i+1}=1 \right] & = \PP\left[ \delta_i < r_i \wedge 1-\delta_i < r_{i+1} \right] = \PP\left[ 1-r_{i+1} < \delta_i < r_i\right]
        = r_i+r_{i+1}-1 \\
        & \leq r_ir_{i+1} = \PP[\delta_i< r_i]\PP[1-\delta_i< r_{i+1}] = \PP\left[ Y_i=1 \right] \PP\left[Y_{i+1}=1 \right],
    \end{align*}
    if $r_i+r_{i+1}\geq 1$, and the first expression is 0, otherwise.
    The expected value of $Z=\sum_{i=1}^{n} Y_i$ is $H-\sum_{i=1}^{n} \lfloor q_i\rfloor$.
    Applying the Hoeffding's concentration bound on $Z$ \citep{hoeffding1994probability}, which is valid for negatively correlated variables too as proven by \citet{panconesi1997randomized}, we obtain that for every $\Delta>0$
    \[
        \PP\left( \Bigg| H - \sum_{i=1}^{n} \left\lfloor q_i \right\rfloor - Z \Bigg| \geq \Delta \right) \leq 2\exp\left( -\frac{2\Delta^2}{n}\right).
    \]
    This implies the result, as $H-\sum_{i=1}^n F_i(p,H) = H - \sum_{i=1}^n \lfloor q_i\rfloor - Z$.

\subsection{Proof of \Cref{prop:lb-dev-fixed-divisor}}\label{app:prop-lb-dev-fixed-divisor}

\proplbdevfixeddivisor*

    Let $n\in \NN$ be arbitrary and let $f$ be a fixed-divisor method with signpost sequences $s_i(k)$ for each $i\in[n]$ and $k\in \NN_0$ satisfying quota. Let also $\varepsilon>0$ be an arbitrary value.
    Since $f$ satisfies quota, for every $i\in [n]$ we have that $s_i(0)\in [0,1]$; otherwise, if $\PP[s_{i'}(0)>1]>0$ for some $i'\in [n]$, taking $p\in \NN^n$ defined as $p_i=1$ for every $i\in [n]$ and $H=n$ we would have $f_{i'}(p,H)=0$, a contradiction since $q_{i'}=1$. We denote $\delta_i = s_i(0)\in [0,1]$ in what follows.
    
    We distinguish two cases for the proof. We first consider the case with $\sum_{i=1}^n \delta_i \leq n/2$. We construct an instance as follows. For $i\in [n-1]$, we let $\varepsilon_i> 0$ be such that $\sum_{i=1}^{n-1} \varepsilon_i \leq \varepsilon$
    and such that $r_i = \delta_i+\varepsilon_i$ is a rational number, say $r_i = a_i/b_i$ for some integers $a_i,b_i$ for each $i\in [n]$. We further define
    \[
        r_n = \left\lceil \sum_{i=1}^{n-1} r_i \right\rceil - \sum_{i=1}^{n-1} r_i,
    \]
    which is also a rational value, so we let $a_n,b_n$ be integers such that $r_n = a_n/b_n$.
    We let $M$ be the least common multiple of all values $b_1,\ldots,b_n$ and let $p\in \NN^n$ and $H\in \NN$ be defined as
    \[
        p_i = Mr_i \text{ for every } i\in [n], \quad H=\sum_{i=1}^{n}r_i.
    \]
    It is straightforward from the previous definitions that all these values are indeed integers and that, for every $i\in [n]$, it holds that $q_i = r_i$. 
    Since $r_i>\delta_i$ for every $i\in [n-1]$, we have that $f_i(p,H)=1$ for every $i\in [n-1]$. If $f_n(p,H)=0$, we further have that $r_n \leq \delta_n$ and thus
    \[
        \sum_{i=1}^{n}r_i \leq \sum_{i=1}^{n} \delta_i + \sum_{i=1}^{n-1}\varepsilon_i \leq \frac{n}{2} + \varepsilon.
    \]
    Therefore, we obtain
    \[
        \sum_{i=1}^{n} f_i(p,H) - H = n-1 - \sum_{i=1}^{n} r_i \geq n-1-\left( \frac{n}{2} + \varepsilon \right) = \frac{n}{2}-1-\varepsilon.
    \]
    Otherwise, if $f_n(p,H)=1$, since $r_n< 1$ we know that
    \[
        \sum_{i=1}^{n}r_i = \sum_{i=1}^{n-1}(\delta_i+\varepsilon_i) + r_n < \frac{n}{2} + \varepsilon + 1,
    \]
    and thus
    \[
        \sum_{i=1}^{n} f_i(p,H) - H = n - \sum_{i=1}^{n} r_i > n-\left( \frac{n}{2} + \varepsilon + 1 \right) = \frac{n}{2}-1-\varepsilon.
    \]

    We now consider the case with $\sum_{i=1}^n \delta_i > \frac{n}{2}$, whose proof is analogous. We construct an instance as follows. For $i\in [n-1]$, we let $\varepsilon_i\geq 0$ be such that $\sum_{i=1}^{n-1} \varepsilon_i \leq \varepsilon$
    and such that $r_i = \delta_i-\varepsilon_i$ is a non-negative rational number, say $r_i = a_i/b_i$ for some integers $a_i,b_i$ for each $i\in [n]$. We further define
    \[
        r_n = \left\lceil \sum_{i=1}^{n-1} r_i \right\rceil - \sum_{i=1}^{n-1} r_i,
    \]
    which is also a rational value, so we let $a_n,b_n$ be integers such that $r_n = a_n/b_n$.
    We let $M$ be the least common multiple of all values $b_1,\ldots,b_n$ and let $p\in \NN^n$ and $H\in \NN$ be defined as
    \[
        p_i = Mr_i \text{ for every } i\in [n], \quad H=\sum_{i=1}^{n}r_i.
    \]
    It is straightforward from the previous definitions that all these values are indeed integers and that, for every $i\in [n]$, it holds that $q_i = r_i$. 
    Since $r_i\leq \delta_i$ for every $i\in [n-1]$, we have that $f_i(p,H)=0$ for every $i\in [n-1]$. If $f_n(p,H)=1$, we further have that $r_n > \delta_n$ and thus
    \[
        \sum_{i=1}^{n}r_i > \sum_{i=1}^{n} \delta_i - \sum_{i=1}^{n-1}\varepsilon_i > \frac{n}{2} - \varepsilon.
    \]
    Therefore, we obtain
    \[
        H - \sum_{i=1}^{n} f_i(p,H) = \sum_{i=1}^{n} r_i - 1 > \frac{n}{2} - 1 - \varepsilon.
    \]
    Otherwise, if $f_n(p,H)=0$, since $r_n \geq 0$ and $\delta_n\leq 1$ we know that
    \[
        \sum_{i=1}^{n}r_i = \sum_{i=1}^{n-1}(\delta_i-\varepsilon_i) + r_n > \frac{n}{2} - 1 - \varepsilon,
    \]
    and thus
    \[
        H - \sum_{i=1}^{n} f_i(p,H) = \sum_{i=1}^{n} r_i > \frac{n}{2} - 1-\varepsilon. \eqno\qed
    \]
    
\section{Proofs Deferred from Section \ref{sec:HM}}

\subsection{Proof of \Cref{prop:hm-solutions}}
\label{app:prop-hm-solutions}

\prophmsolutions*

Let $f$ be a house-monotone solution, and fix $p$ and $H$ arbitrarily. Since $f$ outputs a single vector for every instance, we denote $f(p,H)=\{x^H\}$. As $f$ is house-monotone, there exists $y \in f(p,H+1)$ such that $x^H\leq y$. But $f(p,H+1)$ contains a single vector as well, so we conclude. 

For the converse, we let $f$ be a solution such that $f(p,H)\leq f(p,H+1)$ for every $p\in \NN^n$ and $H\in \NN$. We fix $p\in \NN^n$ arbitrarily and we consider values $H_1,H_2,H_3 \in \NN$ such that $H_1<H_2<H_3$. The hypothesis about $f$ directly implies that $f(p,H_1)\leq f(p,H_2) \leq f(p,H_3)$, so we conclude that $f$ is house-monotone.

\subsection{Proof of \Cref{lem:hm-methods-solutions}}
\label{app:lem-hm-methods-solutions}

\lemhmmethodssolutions*

In what follows, we let $p\in \NN^n$ and $H\in \NN$ be arbitrary. We first show the simplest inclusion. Let $x\in \{f(p,H): f \text{ is a house-monotone and quota-compliant solution}\}$ and let $f$ be any house-monotone solution such that $x = f(p,H)$. As a house-monotone and quota-compliant solution is trivially a house-monotone and quota-compliant method, $f \in \calF^{\text{HM}} \cap \calF^{\text{Q}}$, hence $x \in \bigcup_{f \in \calF^{\text{HM}} \cap \calF^{\text{Q}}} f(p,H)$.

In order to prove the other inclusion, we consider an arbitrary vector $x \in \bigcup_{f \in \calF^{\text{HM}} \cap \calF^{\text{Q}}} f(p,H)$, and we fix $f\in \calF^{\text{HM}} \cap \calF^{\text{Q}}$ to be any house-monotone and quota-compliant method such that $x \in f(p,H)$. In the following, we let $q_i(h)=p_ih/P$ denote the quota of state $i$ for every $i\in [n]$ and $h\in \NN$, and we denote $x^H=x$. We show the existence of a house-monotone and quota-compliant solution $g$ such that $x =g(p,H)$ by using the next simple claims.

\begin{claim}\label{claim:hm-sequence-below}
    For every $h \in \{1,\ldots,H-1\}$, there exists a sequence $x^{H-1},\ldots,x^h$ such that, for every $k\in \{H-1,\ldots,h\}$ it holds that $x^k \in f(p,k)$, $x^k\leq x^{k+1}$, and $x^k_i \in \{\lfloor q_i(k)\rfloor, \lceil q_i(k)\rceil \}$ for each $i\in [n]$.
\end{claim}

\begin{proof}
    We proceed by induction, so let $h$ be as in the statement. For the base case $k=H-1$, we observe that house monotonicity of $f$ implies the existence of a vector $x^{H-1}\in f(p,H-1)$ such that $x^{H-1}\leq x^H$, while quota compliance of $f$ implies that $x^{H-1}_i \in \{\lfloor q_i(H-1)\rfloor, \lceil q_i(H-1)\rceil\}$ for each $i\in [n]$. If we assume that the result holds for every $k \in \{h+1,\ldots,H-1\}$, that is, there exists $x^k\in f(p,k)$ with $x^k \leq x^{k+1}$, then house monotonicity of $f$ implies the existence of $x^{h}\in f(p,h)$ with $x^{h} \leq x^{h+1}$. Quota compliance of $f$ directly yields $x^{h}_i \in \{\lfloor q_i(h)\rfloor, \lceil q_i(h)\rceil \}$ for each $i\in [n]$, so we conclude.
\end{proof}

\begin{claim}\label{claim:hm-sequence-above}
    For every $h \in \{H+1,H+2,\ldots\}$, there exists a sequence $x^{H+1},\ldots,x^h$ such that, for every $k\in \{H+1,\ldots,h\}$ it holds that $x^k \in f(p,k)$, $x^{k-1}\leq x^k$, and $x^k_i \in \{\lfloor q_i(k)\rfloor, \lceil q_i(k)\rceil \}$ for each $i\in [n]$.
\end{claim}

\begin{proof}
    We proceed by induction, so let $h$ be as in the statement. For the base case $k=H+1$, we observe that house monotonicity of $f$ implies the existence of a vector $x^{H+1}\in f(p,H+1)$ such that $x^{H}\leq x^{H+1}$, while quota compliance of $f$ implies that $x^{H+1}_i \in \{\lfloor q_i(H+1)\rfloor, \lceil q_i(H+1)\rceil\}$ for each $i\in [n]$. If we assume that the result holds for every $k \in \{H+1,\ldots,h-1\}$, that is, there exists $x^k\in f(p,k)$ with $x^{k-1} \leq x^{k}$, then house monotonicity of $f$ implies the existence of $x^{h}\in f(p,h)$ with $x^{h-1} \leq x^{h}$. Quota compliance of $f$ directly yields $x^{h}_i \in \{\lfloor q_i(h)\rfloor, \lceil q_i(h)\rceil \}$ for each $i\in [n]$, so we conclude.
\end{proof}

\Cref{claim:hm-sequence-below} and \Cref{claim:hm-sequence-above} yield the existence of an infinite sequence $x^1,x^2,\ldots$ such that $x^H=x$ and that, for every $h\in \NN$, $x^h \in f(p,h)$, $x^h\leq x^{h+1}$, and $x^{h}_i \in \{\lfloor q_i(h)\rfloor, \lceil q_i(h)\rceil \}$ for each $i\in [n]$. Hence, taking $g(p,h)=x^h$ for every $h\in \NN$ and $g(q,h)$ as an arbitrary house-monotone solution for every $q\in \NN^n\setminus \{p\}$ and $h\in \NN$, we have that $g$ is a house-monotone solution with $x = g(p,H)$, as claimed.

\subsection{Proof of \Cref{lem:lp-integral}}
\label{app:lem-lp-integral}

\lemlpintegral*

We prove that the constraint matrix of \eqref{hm-seat}-\eqref{hm-positive} is totally unimodular.
Let $(p,H)\in \NN^n \times\NN$ be an instance and consider the variable vector $z\in \RR^{nH}$ defined as
\[
    z_j = x\left( \left\lfloor \frac{j-1}{H} \right\rfloor +1, j- \left\lfloor \frac{j-1}{H} \right\rfloor H \right),
\]
i.e., $z=(x(1,1),x(1,2),\ldots,x(1,H),x(2,1),x(2,2),\ldots,x(2,H),\ldots,x(n,1),x(n,2)\ldots,x(n,H))$. \linebreak For every $i\in [n]$ and $t\in [H]$, let $e_{i,t}, e'_t \in \{0,1\}^{nH}$ be defined for each $j\in [nH]$ as
\[
    (e_{i,t})_j = \left\{ \begin{array}{ll}
             1 &  \text{ if } \left\lfloor \frac{j-1}{H}\right\rfloor +1 = i \text{ and } j- \left\lfloor \frac{j-1}{H} \right\rfloor H \leq t, \\[1ex]
             0 &  \text{ otherwise,}
             \end{array}
   \right. \quad
   (e'_{t})_j = \left\{ \begin{array}{ll}
             1 &  \text{ if } j- \left\lfloor \frac{j-1}{H} \right\rfloor H = t, \\[1ex]
             0 &  \text{ otherwise.}
             \end{array}
   \right.
\]
Then, the linear program \eqref{hm-seat}-\eqref{hm-positive} is equivalent to 
\begin{align*}
(e_{i,t})^T z &\le \lceil tp_i/P  \rceil\; \text{ for every }i\in [n]\text{ and every }t\in [H],\\ 
-(e_{i,t})^T z &\le - \lfloor tp_i/P \rfloor\; \text{ for every }i\in [n]\text{ and every }t\in [H],\\  
(e'_t)^Tz &\leq 1 \quad\quad\quad \text{ for every }t\in [H],\\
-(e'_t)^Tz &\leq 1 \quad\quad\quad \text{ for every }t\in [H],\\
z&\ge 0.
\end{align*}
It is thus enough to prove that the matrix with rows $e_{i,t}$ for every $i\in [n],~ t\in [H],~ -e_{i,t}$ for every $i\in [n],~ t\in [H],~ e'_t$ for every $t\in [H]$, and $-e'_t$ for every $t\in [H]$, is totally unimodular.
We define this matrix as $A$, such that defining $i(k)=\lfloor (k-1)/H \rfloor +1$ and $t(k)= k- \lfloor (k-1)/H \rfloor H$ for every $k\in [2(n+1)H]$, we have for each $k\in [2(n+1)H]$ that the $k$-th row of $A$ is
\[
    A_k = \left\{ \begin{array}{ll}
             e_{i(k),t(k)}  &  \text{ if } k\in [nH], \\[1ex]
             -e_{i(k-nH),t(k-nH)} &  \text{ if } k\in \{nH+1,\ldots,2nH\}, \\[1ex]
             e'_{k-2nh} & \text{ if } k\in \{2nH+1,\ldots,(2n+1)H\}, \\[1ex]
             -e'_{k-(2n+1)H} & \text{ if } k\in \{(2n+1)H+1,\ldots,2(n+1)H\}.
             \end{array}
   \right.
\]
To do so, we use the characterization of total unimodularity due to \citet{ghouila-houri1962}: A constraint matrix $A \in \RR^{k \times \ell}$ is totally unimodular if and only if for every subset of rows $R \subseteq \{ 1, \dots, k\}$ there is a sign-assignment, i.e., a function $\sigma: R\to \{-1,1\}$, such that
$\sum_{i \in R} \sigma(i)A_{ij} \in \{-1,0,1\} \text{ for all }j \in [\ell].$

Let $R\subseteq [2(n+1)H]$ denote a subset of the rows indices of the constraint matrix, and consider $R'=R\cap [2nH]$.
For each $i\in [n]$, let $R_i\subseteq R'$ be the subset containing the indices of the first $2nH$ rows of $A$ that are in the set $R$ for which a component corresponding to state $i$ is non-zero, i.e.,
\[
    R_i = \left\{ k\in R \cap [2nH]: A_{kj}\not = 0 \text{ for some } j \text{ such that } \left\lfloor \frac{j-1}{H} \right\rfloor +1=i \right\}.
\]
From the definition of $A$, this last condition holds if and only if $k\in \{(i-1)H+1,\ldots,(i-1)H+H\} \cup \{(n+i-1)H+1,\ldots,(n+i-1)H+H\}$, thus
\begin{equation}
    R_i = R \cap \left( \{(i-1)H+1,\ldots,(i-1)H+H\} \cup \{(n+i-1)H+1,\ldots,(n+i-1)H+H\} \right).\label{eq:rows-state-i}
\end{equation}

We claim that, for every $i\in [n]$, there is a function $\sigma_i: R_i \to \{-1,1\}$ such that, for every $j \in [2nH]$ it holds
$\sum_{k \in R_i} \sigma_i(k)A_{kj} \in \{0,1\}.$
We first conclude the result assuming this claim, and show it afterwards.
From \eqref{eq:rows-state-i}, we know that $R_i\cap R_{i'}=\emptyset$ for every $i\not=i'$, and furthermore, $R\cap [2nH]=\bigcup_{i=1}^{n} R_i$.
If the claim is true, defining $\sigma': R\cap [2nH] \to \{0,1\}$ simply as $\sigma'(k) = \sigma_i(k)$ for the unique $i\in [n]$ such that $k\in R_i$, we conclude 
$\sum_{k \in R\cap [2nH]} \sigma'(k)A_{kj} \in \{0,1\}.$
We now define $\sigma: R\to \{-1,1\}$ as follows. For each $k\in [2(n+1)H]$, we let
\[
    \sigma(k) = \left\{ \begin{array}{rl}
             \sigma'(k)  &  \text{ if } k\in R', \\[1ex]
             -1 & \text{ if } k\in \{2nH+1,\ldots,(2n+1)H\}, \\[1ex]
             1 & \text{ if } k\in \{(2n+1)H+1,\ldots,2(n+1)H\} \text{ and } k-H\not\in R, \\[1ex]
             -1 & \text{ if } k\in \{(2n+1)H+1,\ldots,2(n+1)H\} \text{ and } k-H \in R.
             \end{array}
   \right.
\]
We let $j\in [2nH]$ be arbitrary, and show that
$\sum_{k \in R} \sigma(k)A_{kj} \in \{-1,0,1\}.$
Indeed, from the definition of $A$ we have that, defining $t=j+\lfloor (j-1)/H\rfloor H$,
\[
    \sum_{k \in R\setminus R'} \sigma(k)A_{kj} = \left\{ \begin{array}{ll}
    0  \quad \text{ if } R\cap \{2nH+t,2(n+1)H+t\} = \emptyset, \\[1ex]
    \sigma(2nH+t) \quad  \text{ if } R\cap \{2nH+t,2(n+1)H+t\} = \{2nH+t\}, \\[1ex]
    - \sigma(2(n+1)H+t) \quad  \text{ if } R\cap \{2nH+t,2(n+1)H+t\} = \{2(n+1)H+t\}, \\[1ex]
    \sigma(2nH+t) - \sigma(2(n+1)H+t) \quad  \text{ if } \{2nH+t,2(n+1)H+t\} \subseteq R.
    \end{array}
   \right.
\]
From the way we defined $\sigma(k)$, we obtain that this sum equals $0,-1,-1$ and $0$ for each of the four cases, in their respective order.
Since, in addition, we can write
\[
   \sum_{k \in R} \sigma(k)A_{kj} = \sum_{k \in R'} \sigma(k)A_{kj} + \sum_{k \in R\setminus R'} \sigma(k)A_{kj},
\]
where the value of the first sum is either $0$ or $1$ from the claim and the second sum has been proven to be either $-1$ or $0$, we conclude the result.

We now prove the claim. Let $i\in [n]$ and define the partition $(R^1_i, R^2_i, R^3_i, R^4_i)$ as follows
\begin{align*}
    R^1_i & = \{k\in R_i\cap \{(i-1)H,\ldots,(i-1)H+H\}: k+nH \not\in R_i\},\\
    R^2_i & = \{k\in R_i\cap \{(n+i-1)H,\ldots,(n+i-1)H+H\}: k-nH \not\in R_i\},\\
    R^3_i & = ( R_i \cap \{(i-1)H,\ldots,(i-1)H+H\}) \setminus R^1_i,\\
    R^4_i & = ( R_i \cap \{(n+i-1)H,\ldots,(n+i-1)H+H\}) \setminus R^2_i.
\end{align*}
We define $\sigma_i: R_i\to \{-1,1\}$ as follows.
We fix $\sigma_i(k) = 1$ for every $k\in R^3_i \cup R^4_i$, so that $\sigma_i(k)A_k+\sigma_i(k+nH)A_{k+nH}$ vanishes.
Elements in $R^1_i \cup R^2_i$ are arranged in an auxiliary matrix $A'\in \RR^{nH\times 2nH}$ such that, for every $k\in [nH]$, $A'_k =A_k$ if $k\in R^1_i$, $A'_k=A_{k+nH}$ if $k\in R^2_i$, and $A'_k=0$ otherwise.
The intuition behind this construction is that $A'$ has its non-zero rows sorted in decreasing order of non-zero entries, thus assigning signs in a way that two consecutive rows have only zeros and non-zero values of opposite sign ensures obtaining a vector with entries in $\{-1,0,1\}$. Furthermore, starting this assignment with a vector whose entries are in $\{0,1\}$ guarantees obtaining a vector with the same property.
Formally, for each $k\in R^1_i\cup R^2_i$ we let $\pi(k)$ denote the index of the row $A_k$ in the matrix $A'$ restricted to non-zero rows, i.e., $\pi(k)$ is such that $A'_{\pi(k)+z_k}=A_k$, where for each $k\in [nH],~ z_k=|\{\ell\leq k: A'_{\ell}=0\}|$.
Denote as $k_n$ the total number of non-zero rows of $A'$, i.e., $k_n = [nH] - z_{nH}$.
Denote $r\sim_P s$ if two integers $r$ and $s$ have the same parity, i.e., they are both even or both odd, and $r\not\sim_P s$ otherwise.
We then define $\sigma_i(k)$ as follows for these rows.
For $k\in R^1_i\cup R^2_i$, we let
\[
    \sigma_i(k) = \left\{ \begin{array}{rl}
             1  &  \text{ if } k\in R^1_i \text{ and }\pi(k)\sim_P k_N, \text{ or } k\in R^2_i \text{ and }\pi(k) \not\sim_P k_N,\\[1ex]
             -1 & \text{ otherwise}.
             \end{array}
   \right.
\]
This construction implies that, for every $t\in [H]$, defining $j=(i-1)H+t$ we have
\begin{align*}
    \sum_{k\in R_i} \sigma_i(k)A_{kj} & = \sum_{k=(i-1)H+t}^{iH} \sigma_i(k)A_{kj} + \sum_{k=(n+i-1)H+t}^{(n+i)H} \sigma_i(k)A_{kj}\\
    & = \sum_{\substack{k\in R^1_i:\\ k\geq (i-1)H+t}} \sigma_i(k)A_{kj} + \sum_{\substack{k\in R^2_i:\\ k\geq (n+i-1)H+t}} \sigma_i(k)A_{kj}\\
    & = \sum_{k=(i-1)H+t}^{iH} \sigma_i(k)A'_{kj} \\
    & = \left\{ \begin{array}{rl}
             1  &  \text{ if } |\{k\in R^1_i: k\geq (i-1)H+t\}|+|\{k\in R^2_i: k\geq (n+i-1)H+t\}| \text{ is odd,}\\[1ex]
             0 & \text{ otherwise}.
             \end{array}
   \right.
\end{align*}
The first equality comes from the definition and $R_i$ and $A$, the second one from the fact that the rows in $R^3_i$ and $R^4_i$ vanish due to having opposite signs, and the last two from the definition of $A'$.
This concludes the proof of the claim and the proof of the lemma.\qed

\subsection{Proof of Theorem \ref{thm:hm-LP-characterization}}\label{app:HM-proof}

\thmhmlpcharacterization*

\citet{balinski1979,balinski2010fair} provide a characterization of all house-monotone and quota-compliant solutions based on a recursive construction. 
In their approach, a solution is constructed by describing the procedure in which the seats, one by one, are assigned to the states.
We summarize their approach and introduce some notation that will be useful for our proof.
Given a population vector $p=(p_1,\ldots,p_n)$, a house size $H$, and a vector $y=(y_1,\ldots,y_n)$ such that $\sum_{i=1}y_i=H$, let 
\begin{align*}
\calL(p,H,y)&=\left\{i\in [n]:\left\lfloor\frac{p_i}{P}(H+k^*(p,H,y))\right\rfloor>y_i\right\},\\
\calU(p,H,y)&=\left\{i\in [n]:\frac{p_i}{P}(H+1)>y_i\right\}.
\end{align*}
In words, $\calU(p,H,y)$ is the subset of states that can receive the seat $H+1$ without violating their upper quota $\lceil p_i(H+1)/P\rceil$, whereas $\calL(p,H,y)$ is the subset of states that must receive the seat $H+1$ such that the lower quota of any state $j$ with $k^*(p,H,y)$ extra seats, $\lfloor p_j(H+1)/P\rceil$, is fulfilled.
Balinski and Young proved that $f$ is a house-monotone and quota-compliant solution if and only if the values $f(p,\cdot)$ are constructed as follows: $f(p,0) = 0$, and for every house size $H$, $f_{i}(p,H+1)=f_i(p,H)+1$ for one state $i\in \calL(p,H,f(p,H))\cap \calU(p,H,f(p,H))$, and $f_j(p,H+1)=f_j(p,H)$ for every $j\ne i$.

We denote by $K(p,H)$ the set of feasible solutions of \eqref{hm-seat}-\eqref{hm-positive}.
To prove \Cref{thm:hm-LP-characterization}, we use the result by Balinski and Young together with the following lemma, which provides a structural property of the set of extreme points of the linear program \eqref{hm-seat}-\eqref{hm-positive}.

\begin{lemma}\label{lem:x-to-f}
Let $x\in \calE(p,\Phi(p,H))$.
Then, for every $T\le H-1$, we have $x(i,T+1)=0$ for every $i\notin \calL(p,T,A(x,T))\cap \calU(p,T,A(x,T))$.
\end{lemma}

\begin{proof}
Suppose that $x(k,T+1)=1$ for some state $k \notin \calL(p,T,A(x,T))\cap \calU(p,T,A(x,T))$. 
If $k \notin \calU(p,T,A(x,T))$, we have that
\[
    A_k(x,T+1)=A_k(x,T)+x(k,T+1)\ge  \frac{p_i}{P}(T+1)+1 > \left\lceil \frac{p_i}{P}(T+1) \right\rceil,
\]
where the first inequality comes from the definition of the set $\calU(p,T,A(x,T))$.
This implies that the constraint \eqref{hm-upper} for $i=k$ and $t=T+1$ is violated, contradicting the fact that $x\in K(p,\Phi(p,H))$. 
In the following, we suppose that $k \notin \calL(p,T,A(x,T))$.
In particular, since $x(k,T+1)=1$, constraint \eqref{hm-seat} implies that
\begin{equation}
\sum_{i\in \calL(p,T,A(x,T))}x(i,T+1)=0.\label{eq:h-1-zero}
\end{equation} 
Since $T\leq H-1$, from the definition of $\Phi(p,H)$ we have that $T + \tau(p,T,A(x,T))\le \Phi(p,H)$, thus denoting $\bar{T} = T + \tau(p,T,A(x,T))$ we have that the linear program \eqref{hm-seat}-\eqref{hm-positive} includes constraints for every value of $t\in \{T-1,\ldots,\overline{T}\}$.
We obtain 
\begin{equation}
    \sum_{\ell=T+1}^{\overline{T}}\sum_{i\in\calL(p,T,A(x,T))}x(i,\ell) =\sum_{\ell=T+2}^{\overline{T}}\sum_{i\in\calL(p,T,A(x,T))}x(i,\ell)\le \overline{T}-T-1, \label{eq:ub-lower-quota-states}
\end{equation}
where the first equality follows from \eqref{eq:h-1-zero} and the inequality from constraints \eqref{hm-seat}.
On the other hand, we have
\begin{align*}
\sum_{i\in \calL(p,T,A(x,T))}\left(\left\lfloor\frac{p_i}{P}\overline{T}\right\rfloor-A_i(x,T)\right) & \le \sum_{i\in \calL(p,T,A(x,T))}\left(A_i(x,\overline{T})-A_i(x,T)\right) \\
& =\sum_{i\in\calL(p,T,A(x,T))}\sum_{\ell=T+1}^{\overline{T}}x(i,\ell),
\end{align*}
where the inequality holds since $x$ satisfies the constraints \eqref{hm-lower} for every $t\le \overline{T}$.
Together with inequality \eqref{eq:ub-lower-quota-states}, this implies that 
\[
    \sum_{i\in \calL(p,T,A(x,T))}\left(\left\lfloor\frac{p_i}{P}(T+\tau(p,T,A(x,T)))\right\rfloor-A_i(x,T)\right)\le \tau(p,T,A(x,T))-1.
\]
If $\calL(p,T,A(x,T))=[n]$, \eqref{eq:h-1-zero} is an immediate contradiction to constraint \eqref{hm-seat} with $t=T+1$.
Otherwise, we have $\tau(p,T,A(x,T))=k^*(p,T,A(x,T))$ and thus
\[
    \sum_{i\in \calL(p,T,A(x,T))}\left(\left\lfloor\frac{p_i}{P}(T+k^*(p,T,A(x,T)))\right\rfloor-A_i(x,T)\right)\le k^*(p,T,A(x,T))-1,
\] 
a contradiction to the definition of $k^*(p,T,A(x,T))$.
This finishes the proof of the lemma.
\end{proof}

We are now ready to prove \Cref{thm:hm-LP-characterization}.

\begin{proof}[Proof of \Cref{thm:hm-LP-characterization}]
Consider a house-monotone and quota-compliant solution $f$.
For every house size $T\le \Phi(p,H)$ and every $i\in [n]$, let $z(i,T)=f_i(p,T)-f_i(p,T-1)$.
For every $T\in [\Phi(p,H)]$, the house monotonicity of $f$ implies that $\sum_{i=1}^n z(i,T)=\sum_{i=1}^n(f_i(p,T)-f_i(p,T-1))=1$, and therefore constraints \eqref{hm-seat} are satisfied.
For every $i\in [n]$, we have $\sum_{\ell=1}^Tz(i,\ell)=f_i(p,T)$ and therefore constraints \eqref{hm-lower}-\eqref{hm-upper} are all satisfied since $f$ is quota-compliant.
Since $z$ is non-negative, we conclude that $z\in K(p,\Phi(p,T))$, and the integrality of $z$ implies that $z\in \calE(p,\Phi(p,T))$.
Since $A(z,H)=f(p,H)$, we conclude that $f(p,H)\in \{A(x,H):x\in \calE(p,\Phi(p,H))\}$.

For the other direction, let $x\in \calE(p,\Phi(p,H))$.
By \Cref{lem:x-to-f}, we have $A_i(x,T+1)=A_i(x,T)+x(i,T+1)=A_i(x,T)$ for every value $i\notin \calL(p,t,A(x,T))\cap \calU(p,t,A(x,T))$.
Then, the sequence $A(x,0), A(x,1), \ldots, A(x,H)$ can be obtained by the recursive construction of Balinski and Young, from where we conclude that $A(x,H)\in \calA(p,H)$.
This finishes the proof of the theorem.
\end{proof}

\subsection{Proof of \Cref{prop:phi}}\label{app:prop-phi}

\propphi*

Let $p,~H$, and $c$ be as in the statement of the proposition.
    To see \ref{phi-ub-cP}, let $T\in [H-1]$ and note that taking $k=cP-T$, we have that for every $x\in \calE(p,\Phi(p,H-1))$ and $i\in [n]$
    \[
        \left\lfloor \frac{p_i(T+k)}{P} \right\rfloor - A_i(x,T) = cp_i - A_i(x,T) =  \left\lceil \frac{p_i(T+k)}{P} \right\rceil - A_i(x,T) \geq \left\lceil \frac{p_iT}{P} \right\rceil - A_i(x,T) \geq 0,
    \]
    where the last inequality follows from the fact that $x\in \calE(p,\Phi(p,H-1))$ and $T\leq H-1$.
    This yields
    \[
        \sum_{i\in S_k(p,T,A(x,T))} \left(\left\lfloor \frac{p_i(T+k)}{P} \right\rfloor - A_i(x,T) \right) = \sum_{i=1}^{n} (cp_i-A_i(x,T)) = cP-T = k.
    \]
    We conclude that $\tau(p,T,A(x,T))\leq k=cP-T$ for every $x\in \calE(p,\Phi(p,H-1))$ and, therefore,
    \[
        T+\max\Big\{\tau(p,T,A(x,T)):x\in \calE(p,\Phi(p,H-1))\Big\}\leq cP.
    \]
    This implies
    \[
        \Phi(p,H)=\max_{T\in [H-1]}\left\{T+\max\Big\{\tau(p,T,A(x,T)):x\in \calE(p,\Phi(p,H-1))\Big\}\right\} \leq cP,
    \]
    which concludes the proof of part \ref{phi-ub-cP}.

    We now prove part \ref{phi-cP}.
    Taking $x\in \calE(p,\Phi(p,cP-1))$ and $k=1$, we have that for every $i\in [n]$
    \[
        \left\lfloor \frac{p_i(cP-1+k)}{P} \right\rfloor - A_i(x,cP-1) = cp_i-A_i(x,cP-1) =  \left\lceil \frac{p_i(cP-1+k)}{P} \right\rceil - A_i(x,cP-1) \geq 0,
    \]
    where the last inequality follows from the fact that $x\in \calE(p,\Phi(p,cP-1))$.
    This yields
    \[
        \sum_{i\in S_k(p,cP-1,A(x,cP-1))} \left(\left\lfloor \frac{p_i(cP-1+k)}{P} \right\rfloor - A_i(x,cP-1) \right) = \sum_{i=1}^{n} (cp_i-A_i(x,cP-1)) = 1 = k.
    \]
    Since $\tau(p,cP-1,A(x,cP-1))$ is the minimum value of $k$ that satisfies 
    \[
        \sum_{i\in S_k(p,cP-1,A(x,cP-1))} \left(\left\lfloor \frac{p_i(cP-1+k)}{P} \right\rfloor - A_i(x,cP-1) \right) \geq k
    \]
    in case $S_k(p,cP-1,A(x,cP-1))\not=[n]$ and $1$ otherwise, we conclude that $\tau(p,cP-1,A(x,cP-1))=1$ for every $x\in \calE(p,\Phi(p,H))$. Therefore,
    \[
        cP-1+\max\Big\{\tau(p,cP-1,A(x,cP-1)):x\in \calE(p,\Phi(p,cP-1))\Big\}= cP.
    \]
    Putting this together with the result from part \ref{phi-ub-cP}, we conclude that
    \[
        \Phi(p,cP) =  \max_{T\in [cP-1]}\left\{T+\max\Big\{\tau(p,T,A(x,T)):x\in \calE(p,\Phi(p,cP-1))\Big\}\right\} = cP.
    \]

We finally prove part \ref{phi-ub}.
Consider a house size $H\ge 2$. 
Following \citet{balinski2010fair}, for every extreme point $x\in \calE(p,\Phi(p,H-1))$ and every $i\in [n]$, we have $\lfloor p_i(H-1)/P\rfloor\le A_i(x,H-1)\le \lceil p_i(H-1)/P\rceil$, and then, for every integer value $k> \max_{i\in [n]}\lceil A_i(x,H-1)P/p_i-H+1 \rceil$ we have $H-1+k> \max_{i\in [n]}A_i(x,H-1)P/p_i$. 
This implies that 
$p_i(H-1+k)/P> A_i(x,H-1)$ for every $i\in [n]$.
Therefore, we have
\begin{align*}
\sum_{i=1}^n \max\left\{ \left\lfloor \frac{p_i(H-1+k)}{P} \right\rfloor - A_i(x,H-1) ,0\right\} & = \sum_{i=1}^n \left(\left\lfloor \frac{p_i(H-1+k)}{P} \right\rfloor - A_i(x,H-1)\right) \\
& < \sum_{i=1}^n \left( \frac{p_i(H-1+k)}{P} - A_i(x,H-1)\right)\\
&=H-1+k-\sum_{i=1}^nA_i(x,H-1)=k.
\end{align*}
We conclude that $\tau(p,H-1,A(x,H-1))\le \max_{i\in [n]}\lceil A_i(x,H-1)P/p_i-H+1 \rceil$.
For every $i\in [n]$, it holds that
\begin{align*}
\left\lceil \frac{A_i(x,H-1)}{p_i}P-H+1 \right\rceil&\le \left\lceil \left\lceil\frac{p_i(H-1)}{P}\right\rceil\frac{1}{p_i}P-H+1 \right\rceil\\
&\le \left\lceil \left(\frac{p_i(H-1)}{P}+1\right)\frac{1}{p_i}P-H+1 \right\rceil=\left\lceil \frac{P}{p_i}\right\rceil,
\end{align*}
and therefore,
\begin{align*}
\Phi(p,H)&=\max_{T\in [H-1]} \left\{T+\max\Big\{\tau(p,T,A(x,T)):x\in \calE(p,\Phi(p,H-1))\Big\} \right\}\\
&\le \max_{T\in [H-1]} \left\{T+\max_{i\in [n]}\lceil P/p_i\rceil \right\} \leq H-1+\max_{i\in [n]}\lceil P/p_i\rceil,
\end{align*}
from where we conclude that $\Phi(p,H)-H\le \max_{i\in [n]}\lceil P/p_i\rceil$.\qed

\subsection{Proof of \Cref{prop:tightness}}\label{app:prop-tightness}

\proptightness*

We consider the instance with $m=6$ states, where $p=(P-6,2,1,1,1,1)$ and $H=P/3+2$, with $P$ any value divisible by $6$.
We consider $x$ defined as $x(1,t)=1$ for every $t\in [P/3],~ x(i,t)=0$ for every $i\ne 1$ and $t\in [P/3],~ x(3,P/3+1)=x(4,P/3+2)=1,~ x(i,P/3+1)=0$ for every $i\ne 3$, and $x(i,P/3+2)=0$ for every $i\ne 4$.
We first observe that $x\in \calE(p,\Phi(p,H-1))$.
To do so, it is enough to show that for every $T\in [H-1]$ it holds $\tau(p,T,A(x,T))=1$, and we do so by proving that, for every $T\in [H-1]$ and $k\in \NN$ we have
\begin{equation}
    \min\left\{ k\in \NN: \sum_{i\in S(p,T,A(x,T))} \left( \left \lfloor \frac{p_i(T+k)}{P} \right \rfloor - A_i(x,T) \right) \geq k\right\} =P-T.\label{eq:condition-k-T}
\end{equation}
This immediately implies $k^*=P-T$ and thus $S_k(p,T,A(x,T))=[n]$ and $\tau(p,T,A(x,T))=1$.

Observe that when we take $k=P-T,~ (p_i(T+k))/P=p_i \in \NN$, so the $i$-th terms in the sum of the left-hand side of \eqref{eq:condition-k-T} is $p_i-A_i(x,T)$, which summed over $i$ gives $P-T=k$.
Therefore, in order to show \eqref{eq:condition-k-T} we only have to prove that, whenever $k<P-T$, we have
\begin{equation}
    \sum_{i\in S(p,T,A(x,T))} \left( \left \lfloor \frac{p_i(T+k)}{P} \right \rfloor - A_i(x,T) \right) < k.\label{eq:inequality-k-T}
\end{equation}

We first compute, for each state, the lower quota or an upper bound on it for different numbers of seats $h\in [H-1]$:
\begin{align}
    \left\lfloor \frac{p_1h}{P} \right \rfloor & = \left\lfloor \frac{(P-6)h}{P} \right \rfloor \leq h-1 \quad \text{if } h\in [P/3], \label{eq:d1-h-tiny}\\
    \left\lfloor \frac{p_1h}{P} \right \rfloor & = \left\lfloor \frac{(P-6)h}{P} \right \rfloor = h-3 \quad \text{if } h\in \{P/3+1,\ldots,P-1\}, \label{eq:d1-h-medium}\\
    \left\lfloor \frac{p_1h}{P} \right \rfloor & = \left\lfloor \frac{(P-6)h}{P} \right \rfloor \leq h-4 \quad \text{if } h\in \{P/2+1,\ldots,P-1\},\label{eq:d1-h-large}\\
    \left\lfloor \frac{p_2h}{P} \right \rfloor & = \left\lfloor \frac{2h}{P} \right \rfloor = 0 \quad \text{if } h\in [P/2-1], \label{eq:d2-h-small}\\
    \left\lfloor \frac{p_2h}{P} \right \rfloor & = \left\lfloor \frac{2h}{P} \right \rfloor =1 \quad \text{if } h\in \{P/2,\ldots,P-1\}, \label{eq:d2-h-large}\\
    \left\lfloor \frac{p_ih}{P} \right \rfloor & = \left\lfloor \frac{h}{P} \right \rfloor = 0 \quad \text{if } h\in [P-1], \text{ for every } i\in \{3,4,5,6\}. \label{eq:d3456-h}
\end{align}

We now prove \eqref{eq:inequality-k-T}.
We start with the case $T\in [P/3]$, so $A_1(x,T)=T$ and $A_i(x,T)=0$ for every $i\in [n]\setminus \{1\}$.
If $k\in [P/2-T-1]$, then from \eqref{eq:d1-h-tiny}, \eqref{eq:d2-h-small}, and \eqref{eq:d3456-h} we have 
\[
    \sum_{i\in S(p,T,A(x,T))} \left( \left \lfloor \frac{p_i(T+k)}{P} \right \rfloor - A_i(x,T) \right) = \left \lfloor \frac{p_1(T+k)}{P} \right \rfloor - A_1(x,T) \leq T+k-1 - T = k-1<k.
\]
If $k\in \{P/2-T,\ldots, P-T-1\}$, then from \eqref{eq:d1-h-medium}, \eqref{eq:d1-h-large}, \eqref{eq:d2-h-large}, and \eqref{eq:d3456-h} we have 
\begin{align*}
    \sum_{i\in S(p,T,A(x,T))} \left( \left \lfloor \frac{p_i(T+k)}{P} \right \rfloor - A_i(x,T) \right) & = \left \lfloor \frac{p_1(T+k)}{P} \right \rfloor - A_1(x,T) + \left \lfloor \frac{p_2(T+k)}{P} \right \rfloor - A_2(x,T) \\
    & \leq T+k-3- T + 1 - 0 = k-2<k.
\end{align*}
We now address the case $T = P/3+1$, so $A_1(x,T)=P/3,~ A_3(x,T)=1$, and $A_i(x,T)=0$ for every $i\in [n]\setminus \{1,3\}$.
If $k\in [P/6-2]$, then from \eqref{eq:d1-h-medium}, \eqref{eq:d2-h-small}, and \eqref{eq:d3456-h} we have 
\[
    \sum_{i\in S(p,T,A(x,T))} \left( \left \lfloor \frac{p_i(T+k)}{P} \right \rfloor - A_i(x,T) \right) = \left \lfloor \frac{p_1(T+k)}{P} \right \rfloor - A_1(x,T) \leq \frac{P}{3}+k-2 - \frac{P}{3} = k-1<k.
\]
If $k\in \{P/6-1,\ldots, 2P/3-2\}$, then from \eqref{eq:d1-h-medium}, \eqref{eq:d1-h-large}, \eqref{eq:d2-h-large}, and \eqref{eq:d3456-h} we have 
\begin{align*}
    \sum_{i\in S(p,T,A(x,T))} \left( \left \lfloor \frac{p_i(T+k)}{P} \right \rfloor - A_i(x,T) \right) & = \left \lfloor \frac{p_1(T+k)}{P} \right \rfloor - A_1(x,T) + \left \lfloor \frac{p_2(T+k)}{P} \right \rfloor - A_2(x,T) \\
    & \leq \frac{P}{3}+k-2- \frac{P}{3} + 1 - 0 = k-1<k.
\end{align*}
This concludes the proof of \eqref{eq:inequality-k-T} and thus $x\in \calE(p,\Phi(p,H-1))$.

We now show that $\tau(p,H,A(x,H))=P/2-H=P/6-2$, which yields $\Phi(p,H+1)\geq H+P/6-2$ and thus $\Phi(p,H+1)- (H+1) = \Omega(P)$.
Note that $A_1(x,H)=P/3,~ A_3(x,H)=A_4(x,H)=1$, and $A_i(x,H)=0$ for every $i\in [n]\setminus \{1,3,4\}$.
Indeed, for $k\in [P/6-3]$ we have, due to \eqref{eq:d1-h-medium}, \eqref{eq:d2-h-small}, and \eqref{eq:d3456-h}, that
\[
    \sum_{i\in S(p,H,A(x,H))} \left( \left \lfloor \frac{p_i(H+k)}{P} \right \rfloor - A_i(x,H) \right) = \left \lfloor \frac{p_1(H+k)}{P} \right \rfloor - A_1(x,H) = \frac{P}{3}+k-1-\frac{P}{3} = k-1<k.
\]
However, taking $k=P/6-2$, we have 
\begin{align*}
    \sum_{i\in S(p,H,A(x,H))} \left( \left \lfloor \frac{p_i(H+k)}{P} \right \rfloor - A_i(x,H) \right) & = \left \lfloor \frac{p_1(H+k)}{P} \right \rfloor - A_1(x,H) + \left \lfloor \frac{p_2(H+k)}{P} \right \rfloor - A_2(x,T)\\
    & = \frac{P}{2}-3-\frac{P}{3} +1 - 0 = \frac{P}{6}-2 = k,
\end{align*}
so we conclude $\tau(p,H,A(x,H))=P/6-2$.
This concludes the proof of the proposition.\qed

\subsection{Proof of Theorem \ref{thm:rand-characterization}}\label{app:rand-methods}

\thmrandcharacterization*

We first show the following lemma.

\begin{lemma}\label{lem:ex-ante}
Consider a population vector $p$, and let $Y$ be a random variable taking values over $\calE(p,P)$, and distributed according to $\theta \in \Theta(S)$ for some $S\in \calS(p)$. Then, for every $H\le P$, the following holds:
\begin{enumerate}[label=(\alph*)]
    \item For every $i\in [n]$, we have\, $\EE(\sum_{\ell=1}^H Y(i,\ell))=q_i$.\label{hm-random-a}
    \item For every $i\in [n]$, we have\, $\lfloor q_i\rfloor \le \sum_{\ell=1}^HY(i,\ell)\le \lceil q_i\rceil$.\label{hm-random-b}
    \item For every $i\in [n]$ and every $H\le P-1$, we have\, $\sum_{\ell=1}^HY(i,\ell)\le \sum_{\ell=1}^{H+1}Y(i,\ell)$.\label{hm-random-c}
\end{enumerate}
\end{lemma}

\begin{proof}
Parts \ref{hm-random-b} and \ref{hm-random-c} are direct consequences of $Y$ taking values over the set of extreme points $\calE(p,P)$ of the linear program \eqref{hm-seat}-\eqref{hm-positive} with $H=P$.
For every $i\in [n]$, we have
\begin{align*}
\EE\left(\sum_{\ell=1}^HY(i,\ell)\right)&=\sum_{x\in S}\theta_x\sum_{\ell=1}^HY(i,\ell)=\sum_{\ell=1}^H\sum_{x\in S}\theta_xY(i,\ell)=\sum_{\ell=1}^H Q(i,\ell)=q_i,
\end{align*}
where the third equality follows since $\theta\in \Theta(S)$. This proves part \ref{hm-random-a}.
\end{proof}

We now show how to construct a randomized method using our linear programming approach. 
For every population vector $p$, consider an infinite sequence $S(p)=(S_1,S_2,\ldots)$ such that $S_j\in \calS(p)$ for every positive integer $j$.
Then, consider the infinite sequence of independent random variables $Y_1,Y_2,\ldots$ where $Y_j$ is distributed according to some $\theta\in \Theta(S_j)$ for every positive integer $j$, and consider the infinite sequence $h(S(p))=(h_1,h_2,h_3,\ldots)$ where $h_{\ell}=i_{\ell}$ if $Y_{\lfloor \ell/P\rfloor}(i_{\ell},\ell)=1$.
We consider the probability space $\Omega$ such that, for each $p$, we have the random sequence $h(S(p))$ distributed according to the previously defined distribution, and they are all independent across $p$.

For a given realization in $\omega\in \Omega$, in our randomized solution, we provide the following apportionment for a population vector $p$ and a house size $H$: For every $i\in [n]$, $f_i(p,H)=\{\ell\le H:h_{\ell}=i\}$, where $(h_1,h_2\ldots)=h(S(p))$ is the sequence obtained for $p$ according to our previous construction.
By \Cref{lem:ex-ante} we have that this method is house-monotone, quota-compliant, and ex-ante proportional.
We remark that the choice of the sequences $S(p)$ for every population vector $p$ generates different randomized methods.
\citet{goelz2022apportionment} provide a randomized method based on a different type of dependent rounding, based on adapting the bipartite rounding method by \citet{gandhi2006pipage}.

\begin{proof}[Proof of Theorem \ref{thm:rand-characterization}]
Suppose that $F$ is a randomized method in $\calM^{\star}$, that is, house-monotone, quota-compliant, and ex-ante proportional in the restricted domain.
Then, \Cref{thm:hm-LP-characterization} implies that $X_F\in \calE(p,P)$, and let $S$ be the support of $X_F$.
For every $y\in S$, let $\theta_y=\PP(X_F=y)$.
Clearly, we have that $\sum_{y\in S}\theta_y=1$, $\theta>0$, and since $F$ is ex-ante proportional we have $Q=\EE(X_F)=\sum_{y\in S}\theta_y y$.
We conclude that $\theta\in \Theta(S)$ and thus $S\in \calS(p)$.

Now suppose that $F\in \calM$ is such that for every $p$ there exists $S_p\in \calS(p)$ and $\theta\in \Theta(S_p)$ such that for every $x\in S_p$ we have $\PP(X_F=x)=\theta_x$.
By directly applying \Cref{lem:ex-ante}, we obtain that $F$ is ex-ante proportional, quota-compliant, and house-monotone in the restricted domain., i.e., $F\in \calM^*$. 
This finishes the proof of the theorem.
\end{proof}

\bibliographystyle{abbrvnat}
\bibliography{references}

\end{document}